\documentclass[10pt,twocolumn]{article}
\usepackage{fullpage}
\usepackage[utf8]{inputenc}
\usepackage{graphicx}
\usepackage{amsmath, amsthm, amssymb, color,colortbl}
\usepackage[left=2cm, right=2cm, bottom=2cm,top=2cm]{geometry}
\usepackage{titling}
\usepackage{multicol}
\usepackage{nicematrix}
\usepackage{adjustbox}
\usepackage{siunitx}
\usepackage{multirow}
\usepackage{makecell}
\usepackage[utf8]{inputenc}
\usepackage[T1]{fontenc}
\usepackage{sidecap}
\usepackage{blkarray}
\usepackage{parskip}
\usepackage{setspace}
\usepackage{tikz}
\usetikzlibrary{decorations.pathmorphing, calc,arrows,angles,quotes, bbox}
\tikzset{>=stealth}

\usepackage{enumerate}
\usepackage[shortlabels]{enumitem}
\usepackage{tikz}
\usetikzlibrary{quantikz2}
\usepackage{booktabs}
\usepackage{ragged2e}

\usepackage{soul}

\usepackage{array}
\newcolumntype{P}[1]{>{\centering\arraybackslash}m{#1}}
\newcolumntype{U}[1]{>{\centering\arraybackslash}p{#1}}
\newcolumntype{L}[1]{>{\centering\arraybackslash}l{#1}}

\usepackage{xr-hyper}
\usepackage[colorlinks=true, citecolor=blue,linkcolor=blue]{hyperref}

\usepackage{physics, cleveref}
\usepackage{float}
\usepackage{cancel}
\usepackage{appendix}
\usepackage[ruled]{algorithm2e}
\usepackage{import}
\usepackage{relsize}
\usepackage[all]{nowidow}
\usepackage[caption=false]{subfig}
\usepackage{thmtools, thm-restate}
\usepackage{appendix}
\usepackage{balance}

\def\autorefapp#1{\hyperref[#1]{Appendix~\ref{#1}}}

\newcommand{\killpunct}[1]{}

\renewcommand{\bra}[1]{\langle #1 |}
\renewcommand{\ket}[1]{|#1\rangle}
\DeclareMathOperator{\arccosh}{\mathrm{arccosh}}

\renewcommand{\braket}[2]{\langle#1 |  #2\rangle}
\def\ketbra#1{ |{#1}\rangle\!\langle{#1}| }
\def\ketAbraB#1#2{|{#1}\rangle\!\langle{#2}| }

\newcommand{\nrm}[1]{\lVert #1 \rVert}
\renewcommand{\vec}[1]{\boldsymbol{#1}}
\newcommand{\spn}{\mathrm{span}}
\newcommand{\CPiNOT}{\mathrm{C}_{\Pi}\mathrm{NOT}}

\newcommand{\EV}{\mathbb{E}}

\newtheorem{lemma}{Lemma}

\newtheorem{corollary}{Corollary}
\newtheorem{theorem}{Theorem}

\usepackage[
    style=wikibibstyle, 
    citestyle=numeric-comp,
    sorting=none,
    backref=false,
    minnames=3,
    maxnames=13,
    giveninits=true,
    defernumbers=true
    ]{biblatex}
\usepackage{abstract}
\usepackage{lipsum}   

\usepackage[affil-it]{authblk} 

\begin{document}

\title{\textbf{A shortcut to an optimal quantum linear system solver}}

\author{\vspace{6pt}  Alexander M. Dalzell \vspace{-6pt}}
\affil{  AWS Center for Quantum Computing, Pasadena, CA}

\date{}

\twocolumn[
\maketitle 
\vspace{-14pt}
\begin{onecolabstract}
\vspace{10pt}
Given a linear system of equations $A\vec{x} = \vec{b}$, quantum linear system solvers (QLSSs) approximately prepare a quantum state $\ket{\vec{x}}$ for which the amplitudes are proportional to the solution vector $\vec{x}$. Asymptotically optimal QLSSs have query complexity $O(\kappa \log(1/\varepsilon))$, where $\kappa$ is the condition number of $A$, and $\varepsilon$ is the approximation error. However, runtime guarantees for existing optimal and near-optimal QLSSs do not have favorable constant prefactors, in part because they rely on complex or difficult-to-analyze techniques like variable-time amplitude amplification and adiabatic path-following. Here, we give a conceptually simple QLSS that does not use these techniques. If the solution norm $\nrm{\vec{x}}$ is known exactly, our QLSS requires only a single application of \textit{kernel reflection}---a straightforward extension of the eigenstate filtering (EF) technique of previous work---and the query complexity of the QLSS is $(1+O(\varepsilon))\kappa \ln(2\sqrt{2}/\varepsilon)$. 
If the norm is unknown, our method allows it to be estimated up to a constant factor using $O( \log\log(\kappa))$ applications of \textit{kernel projection}---a direct generalization of EF---yielding a straightforward QLSS with near-optimal $O(\kappa \log\log(\kappa)\log\log\log(\kappa)+\kappa\log(1/\varepsilon))$ total complexity. Alternatively, by reintroducing a concept from the adiabatic path-following technique, we show that $O(\kappa)$ complexity can be achieved for estimating the norm up to a constant factor, yielding an optimal QLSS with $O(\kappa\log(1/\varepsilon))$ complexity while still avoiding the need to analyze the adiabatic theorem. Finally, we give explicit upper bounds on the constant prefactors of the complexity statements: we show that the query complexity of our optimal QLSS is at most $56\kappa + 1.05\kappa \ln(1/\varepsilon) + o(\kappa)$, saving more than an order of magnitude compared to existing QLSS complexity guarantees. 
\vspace{14pt}
\end{onecolabstract}
]

\section{Introduction}

Quantum linear system solvers (QLSSs) efficiently produce a quantum state $\ket{\vec{x}}$ that encodes the solution to a linear system of equations, $A\vec{x} = \vec{b}$. 
Since their discovery in 2008 \cite{harrow2009QLinSysSolver}, they have been a key driver of enthusiasm for quantum computing. After all, the need to numerically solve large linear systems appears in a multitude of applications and already represents a key use case for advanced classical computational hardware, such as graphics processing units. However, compared to their classical counterparts, quantum algorithms begin at an orders-of-magnitude disadvantage due to slower physical clock speeds and severe overheads from quantum error correction (see, e.g., \cite{babbush2021FocusBeyondQuadratic}). Thus, providing a substantial quantum speedup with the QLSS will require making optimizations at every level of the computational stack. 

Toward this end, the asymptotic performance of the QLSS has been steadily improved over time. The key parameters that determine the QLSS complexity are the condition number
$\kappa$ of the matrix to be inverted---that is, the ratio of its largest and smallest singular value---and the error $\varepsilon$ sought by the solution. The original QLSS had cost $O(\kappa^2/\varepsilon)$ queries to the input data comprising $A$ and $\vec{b}$ \cite{harrow2009QLinSysSolver}. A sequence of improvements \cite{ambainis2010VTAA, childs2015QLinSysExpPrec, chakraborty2018BlockMatrixPowers, subasi2019QAlgSysLinEqsAdiabatic, an2022QLSStimeDepAdiabatic, lin2019OptimalQEigenstateFiltering,chakraborty2023quantumRegularized,costa2021OptimalLinearSystem} has reduced this complexity to $O(\kappa \log(1/\varepsilon))$ \cite{costa2021OptimalLinearSystem},\footnote{Optimal $O(\kappa \log(1/\epsilon))$ complexity was also obtained by alternative methods in work concurrent \cite{cunningham2024eigenpathTraversal} and subsequent \cite{low2026vtaa} to the first version of this paper. } which matches lower bounds showing that linear-in-$\kappa$ is optimal \cite{harrow2009QLinSysSolver,Orsucci2021solvingclassesof}.\footnote{ It has been shown that $O(\kappa\log(1/\varepsilon))$ is \textit{jointly} optimal in both $\kappa$ and $1/\varepsilon$ \cite{costa2023discreteConstantFactors,mori2026sparsityLowerBound}. However, in this paper, unless otherwise stated, when we write ``optimal'' and ``near-optimal'' we mean only with respect to $\kappa$, where near-optimal refers to scaling $\kappa \, \mathrm{polylog}(\kappa)$. }

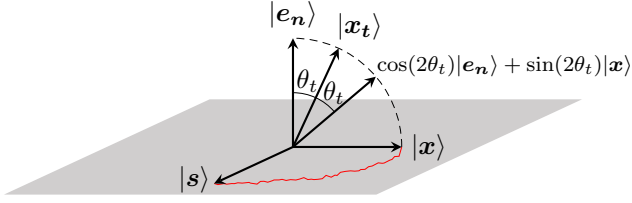
\begin{figure}[t]
    \centering
\definecolor{bkgray}{HTML}{D0CECE}
\begin{tikzpicture}[scale=1.26,thick,baseline=0mm,bezier bounding box]
\def\sinAngle{0.42}
\def\cosAngle{0.907}
\def\sinThetat{0.42}
\def\cosThetat{0.907}
\def\thetatDegrees{24.8}
\def\cosThetatOverTwo{0.9766}
\def\sinThetatOverTwo{0.2147}
\def\height{1}
\def\width{1.75}
\def\radius{5.16/4.52}
\def\offsetX{0}
\def\arcRadius{\radius*0.5}
%
\coordinate (A) at (-\width/2-\height*\cosAngle/\sinAngle+\offsetX,-\height/2);
\coordinate (B) at (-\width/2+\offsetX,\height/2);
\coordinate (C) at (\width/2+\height*\cosAngle/\sinAngle+\offsetX,\height/2);
\coordinate (D) at (\width/2+\offsetX,-\height/2);
\tikzset{rwpath/.style={red, very thin, decorate, decoration = {random steps, segment length=1.5pt, amplitude=0.7pt}}}
\draw[draw=none,fill=bkgray] (A)--(B)--(C)--(D)--cycle;
\coordinate (E) at (0,\radius);
\coordinate (F) at (\radius*\sinThetat, \radius*\cosThetat);
\coordinate (G) at (\radius*2*\sinThetat*\cosThetat,\radius*\cosThetat*\cosThetat-\radius*\sinThetat*\sinThetat);
\coordinate (H) at (\radius,0);
\coordinate (I) at (-\radius*0.8*\cosAngle,-\radius*0.8*\sinAngle);
%
\draw[rwpath] (H) arc(0:-90:{{\radius + \radius*0.8*\cosAngle} and {\radius*0.8*\sinAngle}});
%
\draw[black, thick, ->] (0,0) -- (E);
\draw[black, ->] (0,0) -- (F);
\draw[black, ->] (0,0) -- (G);
\draw[black, ->] (0,0) -- (H);
\draw[black, ->] (0,0) -- (I);
 \draw[black,densely dashed,very thin] (H) arc[start angle=0, end angle=90, radius=\radius];
 %
\node at (\radius*1.25,0) {$\ket{\vec{x}}$};
\node at (0,\radius*1.22) {$\ket{\vec{e_n}}$};
\node at (\radius*\sinThetat*1.25+0.1,\radius*\cosThetat*1.25-0.03) {$\ket{\vec{x_t}}$};
\node at (-\radius*0.8*\cosAngle*1.25,-\radius*0.8*\sinAngle*0.9) {$\ket{\vec{s}}$};
\node at (1.25*\radius*2*\sinThetat*\cosThetat+\width*0.65,1.25*\radius*\cosThetat*\cosThetat-1.25*\radius*\sinThetat*\sinThetat-0.07) {\footnotesize $\cos(2\theta_t)\ket{\vec{e_n}} + \sin(2\theta_t)\ket{\vec{x}}$};
%
\draw[ultra thin, black] (\arcRadius*2*\sinThetat*\cosThetat,\arcRadius*\cosThetat*\cosThetat-\arcRadius*\sinThetat*\sinThetat) arc(90-2*\thetatDegrees:90:\arcRadius);
\node at (\arcRadius*\sinThetatOverTwo*1.25,\arcRadius*\cosThetatOverTwo*1.25) {$\theta_t$};
\node at (\arcRadius*\sinThetatOverTwo*\cosThetat*1.25+\arcRadius*\cosThetatOverTwo*\sinThetat*1.25,\arcRadius*\cosThetatOverTwo*\cosThetat*1.25-\arcRadius*\sinThetatOverTwo*\sinThetat*1.25) {$\theta_t$};
\end{tikzpicture}

    \caption{Illustration of the conceptual idea of the quantum linear system solver in this paper.  Roughly speaking, previous methods approximately follow an adiabatic path from an initial state $\ket{\vec{s}}$ to the solution $\ket{\vec{x}}$, depicted by the red line, and then refine the resulting state with eigenstate filtering (EF). In contrast, our method begins in a known state $\ket{\vec{e_n}}$ that is deliberately orthogonal to $\ket{\vec{x}}$. Given an estimate $t$ for $\nrm{\vec{x}}$, it uses the quantum singular value transformation (QSVT) to reflect about the solution $\ket{\vec{x_t}}$ to an augmented linear system, arriving at $\cos(2\theta_t)\ket{\vec{e_n}}+ \sin(2\theta_t)\ket{\vec{x}}$, where $\theta_t = \arctan(\nrm{\vec{x}}/t)$ is the initial angle between $\ket{\vec{e_n}}$ and $\ket{\vec{x_t}}$. Then, it projects onto the image of $I-\ketbra{\vec{e_n}}$ (shaded plane) with success probability $\sin^2(2\theta_t)$, to arrive at $\ket{\vec{x}}$, up to errors induced by imperfect reflection. The success probability is $\Omega(1)$ as long as $t$ is a constant-factor approximation of $\nrm{\vec{x}}$. }
    \label{fig:two_step}
\end{figure}

The optimal QLSS of Ref.~\cite{costa2021OptimalLinearSystem} achieves this state-of-the-art complexity by combining two techniques: (i) \textit{eigenstate filtering} (EF), and (ii) \textit{adiabatic path-following} via the ``quantum walk method.'' It utilizes EF \cite{lin2019OptimalQEigenstateFiltering} to approximately project onto the ideal state $\ket{\vec{x}}$ at cost $O(\kappa \log(1/\varepsilon))$. However, to apply this technique, one must first prepare an ansatz state $\ket{\vec{x_{\rm ans}}}$ that has constant overlap $\gamma = |\braket{\vec{x_{\rm ans}}}{\vec{x}}| = \Omega(1)$, so that EF succeeds with substantial probability, and need not be repeated more than $1/\gamma^2 = O(1)$ times on average. Extending ideas first explored in Refs.~\cite{subasi2019QAlgSysLinEqsAdiabatic,an2022QLSStimeDepAdiabatic}, this ansatz state is prepared by applying a sequence of unitary ``walk'' operators, where the eigenstates of the operators follow a smoothly varying ``eigenpath.'' The discrete adiabatic theorem guarantees that if the path is traversed sufficiently slowly, the state of the system will approximately track the eigenpath. For the QLSS, one chooses the sequence of walk operators so that the path begins at a simple-to-prepare state related only to the vector $\vec{b}$. As the sequence progresses, information about $A$ is introduced, and the path ends at a state from which $\ket{\vec{x}}$ can be easily extracted.  Ultimately, the method prepares a state $\ket{\vec{x_{\rm ans}}}$ achieving overlap $\gamma$ at total cost $O(\kappa/\sqrt{1-\gamma^2})$. 

However, rigorous treatment of adiabatic algorithms are notoriously difficult \cite{albash2018AQCreview}. Accordingly, the analysis of the discrete adiabatic theorem in Ref.~\cite{costa2021OptimalLinearSystem} is intricate and introduces a large constant prefactor, on the order of $10^5$ (see \cite[Eq.~(L2) of arXiv v1]{jennings2023QLSS}), into the rigorous bound on the complexity of preparing the ansatz state via the quantum walk method. A rigorous analysis \cite{jennings2023QLSS} of an alternative adiabatic path-following strategy called the ``randomization method'' \cite{subasi2019QAlgSysLinEqsAdiabatic} has a similar conclusion: ansatz state preparation is far more costly than the subsequent EF step, even for very small values of $\varepsilon$. 
This large cost of the QLSS---and in particular ansatz state preparation---contributes to the large resource estimates reported for QLSS-based applications, such as financial portfolio optimization \cite{dalzell2022socp} and solving differential equations \cite{jennings2023costDiffEQ}. 
Numerical simulations \cite{costa2023discreteConstantFactors} on small random linear systems of size up to $16 \times 16$ suggest that for both the randomization method and the quantum walk method, the constant prefactors involved are much smaller than their explicit upper bounds. Nevertheless, it remains desirable to have a tighter upper bound that is guaranteed to hold for larger systems, and in the worst case.

In this paper, we introduce a new conceptual idea, depicted in \autoref{fig:two_step}, that eliminates the need for ansatz state preparation via adiabatic path-following. Briefly, the idea is to \textit{augment} the linear system by adding a single additional variable and a single uncoupled linear constraint on that variable. The solution to the augmented linear system simultaneously has constant overlap with the solution $\ket{\vec{x}}$ and with an easy-to-prepare state we label $\ket{\vec{e_n}}$. Thus, we can quickly navigate from $\ket{\vec{e_n}}$ to $\ket{\vec{x}}$ by approximately reflecting about the solution to the augmented system---implementing this reflection is an extension of the EF technique and has cost  $O(\kappa \log(1/\varepsilon))$. The result is a simper QLSS that relies only on EF-like operations combined with a few straightforward linear algebra observations. 

A caveat is that augmenting the linear system in the right way requires knowing (or guessing) an estimate of the Euclidean norm $\nrm{\vec{x}}$ of the solution to the linear system. Thus, in a sense, our work might be viewed as \textit{replacing} the ansatz-state-preparation step in prior QLSSs with a solution-norm-estimation step. This trade is beneficial since $\nrm{\vec{x}}$ is a single real number that need only be learned once, rather than a high-dimensional quantum state that must be reprepared each time the linear system is to be solved. If the norm is known to within a constant multiplicative factor, our QLSS is asymptotically optimal, with $O(\kappa \log(1/\varepsilon))$ cost. If the norm is unknown, we also show that our framework offers a simple way to learn it (up to a constant factor) in near-optimal cost $O(\kappa\log\log(\kappa)\log\log\log(\kappa))$.
Alternatively, by using our method together with some of the intuition from the adiabatic approach, we show that the norm can be learned up to a constant factor in optimal $O(\kappa)$ complexity, while still avoiding the intricate analysis of the adiabatic theorem.  
Furthermore, the approximation ratio of the norm estimate can be improved from constant to $1+\varepsilon$ incurring additional  cost equal to $O(\kappa\log(1/\varepsilon)/\varepsilon)$, shaving multiple $\log(\kappa)$ and $\log(1/\varepsilon)$ factors off of the previous state-of-the-art method for norm estimation from Ref.~\cite{chakraborty2018BlockMatrixPowers}, and nearly matching the lower bound of $\Omega(\kappa/\varepsilon)$ that we show in the appendix.  Depending on which norm estimation method is used, this gives a full QLSS with optimal or near-optimal query complexity.

In \autoref{sec:QLSP}, we establish notation and background information about the quantum linear system problem, and in \autoref{sec:KP_KR_main} we introduce \textit{kernel projection} and \textit{kernel reflection}, the EF-like operations that constitute the main technical ingredient of our QLSS. In \autoref{sec:QLSS_given_norm_main}, we present our main algorithm for solving the quantum linear system problem when an estimate for the solution norm is known, and in \autoref{sec:estimating_the_norm}, we present several algorithms for estimating the norm, along with compact, non-optimized proofs of correctness. In \autoref{sec:constant_factors}, we provide a discussion of constant prefactors, based on a more detailed version of our QLSS that we develop in the appendix.  Finally, in \autoref{sec:conclusion} we provide some concluding remarks.

\section{The quantum linear system problem}\label{sec:QLSP}

The input to the linear system problem is an $m \times n$ matrix $A \in \mathbb{C}^{m\times n}$ and a vector $\vec{b} \in \mathbb{C}^m$. 
The solution is a vector $\vec{x} \in \mathbb{C}^n$ for which $A\vec{x} = \vec{b}$, assuming such a solution exists.\footnote{If no solution exists (because $\vec{b}$ is not in the column space of $A$), we ideally seek the vector for which the least-squares metric $\nrm{A\vec{x} - \vec{b}}$ is minimized. However, our QLSS does not offer a solution in this case, and throughout we assume that $\vec{b}$ is in the column space of $A$. The least-squares case can be handled by the near-optimal QLSS of Ref.~\cite{chakraborty2018BlockMatrixPowers}.} 
When more than one such solution exists, we seek the $\vec{x}$ with minimum Euclidean norm $\nrm{\vec{x}}$. 

Let $s = \lceil \log_2(1+\max(m,n))\rceil$ (here we add one to leave room for an extra row/column, as discussed later). Consider an $s$-qubit system with orthonormal computational basis states $\{\ket{\vec{e_j}}\}_{j=0}^{2^s-1}$. For $k \leq 2^s$, we associate $k$-dimensional vectors $\vec{u} = (u_0,\ldots,u_{k-1})\in \mathbb{C}^k$ with normalized $s$-qubit states $\ket{\vec{u}} = \nrm{\vec{u}}^{-1}\sum_{j=0}^{k-1} u_j\ket{\vec{e_j}}$, for which the coefficients in the computational basis are proportional to the corresponding vector entries. For $k_r,k_c \leq 2^s$, we associate $k_r \times k_c$ matrices $M$ with the operator $\sum_{i=0}^{k_r-1}\sum_{j=0}^{k_c-1} m_{ij} \ketAbraB{\vec{e_i}}{\vec{e_j}}$, where $m_{ij}$ are the matrix entries of $M$. Then, following standard conventions, we assume we have access to the data in the $m$-dimensional vector $\vec{b}$ via a unitary operator $U_{\vec{b}}$ for which $U_{\vec{b}}\ket{\vec{e_0}}=\ket{\vec{b}}$. We also assume that $\nrm{A} \leq 1$ (where $\nrm{\cdot}$ denotes spectral norm for matrix arguments),  and that we have access to the matrix $A$ via a $(\alpha,a)$-block-encoding of $A$, that is, an $(a+s)$-qubit unitary\footnote{Since $\max(m,n) < 2^s$, this assumes that the block-encoded matrix already has at least one row and column of zero-padding. If this is not the case, this padding can be created with one ancilla qubit; see \autorefapp{sec:how_to_pad}.} operation $U_A$ for which
\begin{equation} \nonumber
A =  \alpha(\bra{0}^{\otimes a} \otimes I_{2^s})U_A (\ket{0}^{\otimes a} \otimes I_{2^s})
\end{equation}
where $I_{d}$ denotes the identity operator on a Hilbert space of dimension $d$, here the Hilbert space for the $s$-qubit system. Note that unitarity requires $\alpha \geq \nrm{A}$.

The quantum linear system problem (QLSP) takes as input an error parameter $\varepsilon$ and asks to produce a quantum state $\ket{\vec{\tilde{x}}}$ for which the trace distance $\frac{1}{2}\nrm{\ketbra{\vec{x}}-\ketbra{\vec{\tilde{x}}}}_1$ is at most $\varepsilon$, while minimizing the number of queries to $U_A$ and $U_{\vec{b}}$, as well as their controlled versions and their inverses. We also accept randomized algorithms that output a mixed state $\tilde{\rho}$ for which  $\frac{1}{2}\nrm{\ketbra{\vec{x}}-\tilde{\rho}}_1\leq \varepsilon$. Our analysis could have been performed for other related metrics, such as $\nrm{\ket{\vec{x}}-\ket{\vec{\tilde{x}}}}$, as used in Ref.~\cite{costa2021OptimalLinearSystem}.

Additionally, henceforth we assume the convention that $\nrm{\vec{b}} = 1$ and $\alpha=1$.  This convention is also without loss of generality since the quantum linear system problem asks only to prepare a quantum state proportional to $\vec{x}$, making it insensitive to scaling of $A$ and $\vec{b}$. We suppose there is a known value of $\kappa$ such that all nonzero singular values of $A$ fall in the interval $[\kappa^{-1},1]$---thus $\kappa$ is an upper bound on the condition number of $A$, when $A$ is restricted to act on inputs orthogonal to its kernel.\footnote{Actually, we need only consider singular values corresponding to singular vectors on which $\vec{b}$ has support, when computing the required choice of $\kappa$. } These conventions imply that 
\begin{equation}
    1 \leq \nrm{\vec{x}} \leq \kappa\,.
\end{equation}

\section{Kernel projection and kernel reflection}\label{sec:KP_KR_main}

Kernel projection (KP) is the term we give to the technique of eigenstate filtering, generalized to non-Hermitian and potentially non-square matrices. As its name suggests, when combined with postselection, it leads to approximate projection onto the kernel of a matrix. Similarly, kernel reflection (KR) leads to approximate reflection about the kernel. In either case, the first step is to construct a matrix $G$ out of $A$ and $\vec{b}$ for which $\vec{x}$ is in the kernel. Namely, following prior work \cite{subasi2019QAlgSysLinEqsAdiabatic,an2022QLSStimeDepAdiabatic,lin2019OptimalQEigenstateFiltering}, we let
\begin{equation}\label{eq:G}
    G = Q_{\vec{b}}A 
\end{equation}
where the $m \times m$ matrix $Q_{\vec{b}} = I_m - \vec{b}\vec{b}^\dagger $
is the projector onto the subspace orthogonal to $\vec{b}$. The kernel of $G$ is the span of the kernel of $A$ and the vector $\vec{x}$. Here, we note that since we have defined $\vec{x}$ to be the solution to $A\vec{x} = \vec{b}$ of minimum norm $\nrm{\vec{x}}$, $\vec{x}$ is orthogonal to the kernel of $A$.  We can easily construct a $(1,a+1)$-block-encoding $U_G$ for $G$ using one query to each of $U_A$, $U_{\vec{b}}$, and $U_{\vec{b}}^\dagger$; see \autoref{app:block_encodings}.

Furthermore, we can assert that the smallest nonzero singular value of $G$ is at least $\kappa^{-1}$, by the following argument (see also  \cite{subasi2019QAlgSysLinEqsAdiabatic,an2022QLSStimeDepAdiabatic}). Suppose $\sigma$ is a nonzero singular value of $G$, with normalized left singular vector $\vec{u}$ and right singular vector $\vec{v}$. Since $\sigma$ is nonzero, $\vec{v}$ is orthogonal to $\vec{x}$ and to the kernel of $A$. Moreover, by definition of singular vectors, $\sigma \vec{u} = Q_{\vec{b}} A \vec{v} = A\vec{v} - (\vec{b}^\dagger A\vec{v})\vec{b} = A(\vec{v} - (\vec{b}^\dagger A\vec{v})\vec{x})$, implying that $\vec{u}$ is in the image of $A$, and also that $Q_{\vec{b}} \vec{u} = \vec{u}$.  By the definition of singular values, $\sigma^2 = \vec{u}^\dagger Q_{\vec{b}} A A^\dagger Q_{\vec{b}} \vec{u} = \vec{u}^\dagger AA^\dagger\vec{u}\geq 1/\kappa^2$, where the last inequality follows from the fact that $\vec{u}$ is in the image of $A$, and thus can be expressed as a linear combination of eigenvectors of $AA^\dagger$, all of which have eigenvalue at least $1/\kappa^2$.

With these facts established, we can now define the technique of kernel projection. It takes as input two parameters $(\kappa,\eta)$ and only works under the promise that the nonzero singular values of $G$ lie in the interval $[\kappa^{-1},1]$ (which we proved above in the case $G=Q_{\vec{b}}A$). Consider an arbitrary normalized $s$-qubit state $\ket{\phi} = \gamma \ket{\vec{w}} + \nu \ket{\vec{w_\perp}}$, where $\vec{w}$ is a unit vector in the kernel of $G$ (typically $\ket{\vec{w}} = \ket{\vec{x}}$), and $\vec{w_\perp}$ is a unit vector orthogonal to the kernel of $G$. KP enacts the transformation
\begin{equation}\label{eq:impact_KP}
\gamma \ket{\vec{w}} + \nu \ket{\vec{w}_{\perp}} \overset{\rm KP}{\mapsto} \gamma \ket{\vec{w}} + \nu\delta_1\ket{\vec{w_{\perp}}}+\nu\delta_2\ket{\vec{w_{\perp}'}}
\end{equation}
where $\delta_1$ and $\delta_2$ are real parameters (dependent on $\vec{w}$), that satisfy $\sqrt{\delta_1^2+\delta_2^2}\leq \eta$, and $\vec{w_{\perp}'}$ is another unit vector orthogonal to $\vec{w_{\perp}}$ and to the kernel of $G$. In other words, KP leaves the kernel of $G$ untouched while shrinking the norm of the portion of the state orthogonal to the kernel by a factor of $\eta$ (or more). The right-hand side is subnormalized: the probability that KP succeeds is given by its squared-norm $|\gamma|^2+|\nu|^2(\delta_1^2+\delta_2^2)$. This probability can be small if the initial state $\ket{\phi}$ has small overlap $\gamma$ with the kernel of $G$, necessitating the step of ansatz state preparation in prior work \cite{lin2019OptimalQEigenstateFiltering,costa2021OptimalLinearSystem}. 

Similarly, KR can be understood as enacting the transformation 
\begin{equation}\label{eq:impact_KR}
\gamma \ket{\vec{w}} + \nu \ket{\vec{w}_{\perp}} \overset{\rm KR}{\mapsto} \gamma \ket{\vec{w}} - \nu(1-\delta_1')\ket{\vec{w_{\perp}}}+\nu\delta_2'\ket{\vec{w_{\perp}'}}\,.
\end{equation}

\noindent For fixed $\vec{w}$ and $\vec{w_{\perp}}$, the vector $\vec{w_{\perp}'}$ in Eq.~\eqref{eq:impact_KR} is the same as the one in Eq.~\eqref{eq:impact_KP}. Moreover, we have the identities $\delta_1' = \frac{2\eta+2\delta_1}{1+\eta}$ and $\delta_2' = \frac{2\delta_2}{1+\eta}$, from which we can derive relations 
\begin{equation}\label{eq:delta_deltaprime_relations}
 \delta_1' \geq 0, \;\;\;  \sqrt{\delta_1'^2+\delta_2'^2} \leq \frac{4\eta}{1+\eta}, \;\;\; |\delta'_2|\leq \frac{2\eta}{1+\eta}   
\end{equation}
These identities are justified in \autoref{app:KP_KR}, and follow from the close relationship between KP and KR. 
Namely, both KP and KR are enacted as an application of the quantum singular value transformation (QSVT) \cite{gilyen2018QSingValTransf, martyn2021GrandUnificationQAlgs}. Briefly, the QSVT procedure involves preparing the state $\ket{0}^{\otimes(a+2)} \otimes \ket{\phi}$, applying a sequence of gates including $U_G$ and $U_G^\dagger$, and then postselecting on the first register being $\ket{0}^{\otimes(a+2)}$---success or failure of KP and KR is heralded by the outcome of these measurements. The circuit is constructed to preserve the right singular vectors of $G$ while applying a certain polynomial transformation to the singular values, and the total cost of the procedure scales with the degree of the polynomial. To perform KP, we choose a polynomial $p$ for which $p(0)=1$ and $|p(x)| \leq \eta$ for all $x \geq \kappa^{-1}$. To perform KR, we choose a related polynomial where $p(0) = 1$ and $-1\leq p(x) \leq -1+4\eta/(1+\eta)$ for all $x \geq \kappa^{-1}$.  
In both cases, the degree of this polynomial is $2\ell$, where 
\begin{equation}
    \ell = \lceil \kappa \ln(2/\eta)/2\rceil
\end{equation}
The cost of implementing the KP or KR unitary is $\ell$ calls to the block-encoding $U_G$, $\ell$ calls to its inverse $U_G^{\dagger}$,  $4\ell$ multi-controlled Toffoli gates, and $2\ell$ single-qubit rotation gates. 
Details of this implementation can be found in \autoref{app:KP_KR}.

KP is equivalent to EF when the matrix $A$ is Hermitian. EF was developed in Ref.~\cite{lin2019OptimalQEigenstateFiltering}, where it was applied to the QLSP to give the first QLSS with error dependence strictly linear in $\log(1/\varepsilon)$. The idea was to use existing QLSSs to first produce a state $\ket{\vec{x_{\rm ans}}}$ with $\gamma = \Omega(1)$ overlap with $\ket{\vec{x}}$, and then apply EF to project to a state $\varepsilon$-close to $\ket{\vec{x}}$, succeeding with probability roughly $|\gamma|^2$. Later, Ref.~\cite{costa2021OptimalLinearSystem} showed that a QLSS based on the discrete adiabatic theorem can produce a state $\ket{\vec{x_{\rm ans}}}$ with $O(\kappa/\sqrt{1-\gamma^2})$ complexity, which, combined with EF and taking $\gamma = \Omega(1)$, gives an optimal QLSS with overall $O(\kappa\log(1/\varepsilon))$ complexity. 


\section{Optimal QLSS given constant-factor estimate for norm}\label{sec:QLSS_given_norm_main}

The main conceptual idea presented in this paper is to form an augmented linear system by introducing one new variable to the system, and adding one new (uncoupled) linear constraint involving that variable. This adds a known singular vector to the linear system, with a known, tunable singular value. If the known singular value is chosen appropriately, then the solution to the augmented linear system simultaneously has substantial overlap with the solution $\ket{\vec{x}}$ of the original linear system, and with the known singular vector. This produces a navigable path from the known singular vector to the solution $\ket{\vec{x}}$ that bypasses the need for sophisticated ansatz-preparation methods---a shortcut to an optimal QLSS.  

Recall that $\nrm{\vec{b}} = 1$ and $\nrm{A} \leq 1$ by convention, and that this implies that $1 \leq \nrm{\vec{x}} \leq \kappa$. Let $t \in [1,\kappa]$ be an estimate of $\nrm{\vec{x}}$ and define 
\begin{equation}
    \theta_t = \arctan\left(\frac{\nrm{\vec{x}}}{t}\right)\,.
\end{equation}
Let $A_t$ be an $(m+1)\times(n+1)$ matrix, where the upper left $m \times n$ block is $A$, the lower right entry is $1/t$, and the other entries are zero:
\begin{equation}\label{eq:A_t}
    A_t = A + \frac{1}{t}\vec{e}_m \vec{e_n}^\dagger = 
    \begin{bmatrix}
        A & 0 \\
        0 & t^{-1}
    \end{bmatrix}\,.
\end{equation}
All nonzero singular values of $A_t$ lie in the interval $[\kappa^{-1},1]$. 
Furthermore, define
\begin{equation}
    \vec{b'} = \frac{1}{\sqrt{2}}\left(\vec{b} + \vec{e_m}\right) = 
    \begin{bmatrix}
        \vec{b}/\sqrt{2} \\
        1/\sqrt{2}
    \end{bmatrix}
\end{equation}
It is then easy to verify that the solution to the augmented system $A_t\vec{x_t} = \vec{b'}$ is
\begin{equation}
    \vec{x_t} = \frac{1}{\sqrt{2}}\left(\vec{x} + t \vec{e_n}\right) = 
    \begin{bmatrix}
        \vec{x} /\sqrt{2}\\
        t / \sqrt{2}
    \end{bmatrix} \,.
\end{equation}
The vectors $\vec{x}$, $\vec{x_t}$, and $\vec{e_n}$ are depicted in \autoref{fig:two_step}. All three lie in the plane spanned by $\vec{x}$ and $\vec{e_n}$ with $\theta_t$ the angle between $\vec{e_n}$ and $\vec{x_t}$. 

With one ancilla qubit and one controlled query to $U_{\vec{b}}$, we can easily construct a unitary $U_{\vec{b'}}$ that prepares $\ket{\vec{b'}}$. Similarly, with one controlled query to $U_A$, we can construct a $(1,a+1)$-block-encoding $U_{A_t}$ for $A_t$. This block-encoding can be turned into a $(1, a+2)$-block-encoding $U_{G_t}$ of the matrix  (cf.~Eq.~\eqref{eq:G})
\begin{equation}\label{eq:G_t}
    G_t = Q_{\vec{b'}}A_t
\end{equation}
for which $\vec{x_t}$ lies in the kernel, and the nonzero singular values are contained in $[\kappa^{-1},1]$; these block-encoding constructions are provided in \autoref{app:block_encodings}.

Our main algorithm proposes to begin in the state $\ket{\vec{e_n}}$ and end in the orthogonal state $\ket{\vec{x}}$. In this journey through the $n$-dimensional Hilbert space, the vector $\ket{\vec{x_t}}$ is the essential landmark that charts the correct path. The QSVT-based technique of KR provides the vehicle for traversing this path, using the matrix $G_t$ as its compass, as $G_t$ encodes $\ket{\vec{x_t}}$ into its kernel. As we will see, KR is most effective when the state $\ket{\vec{x_t}}$ lies equally far from $\ket{\vec{e_n}}$ and $\ket{\vec{x}}$, corresponding to $\theta_t = \pi/4$. If $\theta_t$ is larger or smaller, the algorithm requires more repetitions due to reduced success probability. 

Formally, the algorithm takes as input the value of $\kappa$, as well as a choice for $t \in [1,\kappa]$ and for $\eta \in (0,1]$. The procedure has three steps, described in \autoref{algo:main_algo}. 

\SetKwInput{Input}{Input}
\SetKwInput{Output}{Output}
\SetKwInput{QCost}{Query Cost}
\SetKwInput{Set}{Set}
\SetKwFunction{ApprSolve}{ApprSolve}
\SetStartEndCondition{ }{}{}%
\SetKwProg{Fn}{def}{\string:}{}
\SetKwFunction{Range}{range}
\SetKw{KwTo}{in}\SetKwFor{For}{for}{\string:}{}%
\SetKwIF{If}{ElseIf}{Else}{if}{:}{elif}{else:}{}%
\SetKwFor{While}{while}{:}{fintq}%
\newcommand{\forcond}{$i$ \KwTo\Range{$n$}}
\AlgoDontDisplayBlockMarkers\SetAlgoNoEnd\SetAlgoNoLine%
\setlength{\algomargin}{1em}
\begin{algorithm}
\caption{QLSS given norm estimate }\label{algo:main_algo}
\DontPrintSemicolon
\SetAlgoLined
\LinesNumbered
\Input{$(A, \vec{b}, \kappa, \eta, t)$}
\Output{$\ket{\vec{\tilde{x}}}$ with probability $p_{\rm succ}$, or ``fail'' with probability $1-p_{\rm succ}$}
\BlankLine
Prepare $\ket{\vec{e_n}}$ \;
Apply KR to approximately reflect about the kernel of $G_t$ (defined in Eq.~\eqref{eq:G_t}), with parameters $(\kappa,\eta)$. If KR fails, output ``fail.''\;
Project onto $\spn\{\ket{\vec{e_j}}\}_{j=0}^{n-1}$ to produce $\ket{\vec{\tilde{x}}}$, by measuring the operator $I-\ketbra{\vec{e_n}}$. If projection fails, output ``fail.''\;
\end{algorithm}
We now analyze each step of this algorithm. The state prepared in step 1 can be decomposed as
\begin{equation}
    \ket{\vec{e_n}} = \cos(\theta_t)\ket{\vec{x_t}} + \sin(\theta_t)\ket{\vec{y_t}}
\end{equation}
where 
\begin{align}
    \ket{\vec{x_t}} &= \phantom{-}\sin(\theta_t)\ket{\vec{x}} + \cos(\theta_t)\ket{\vec{e_n}}\,, \\
    \ket{\vec{y_t}} &= -\cos(\theta_t)\ket{\vec{x}} + \sin(\theta_t)\ket{\vec{e_n}}
\end{align}
are orthogonal states. Since $\ket{\vec{x_t}}$ is in the kernel of $G_t$ and $\ket{\vec{y_t}}$ is orthogonal to $\ket{\vec{x_t}}$ and to the kernel of $A_t$ (to see this, observe that both $\ket{\vec{x}}$ and $\ket{\vec{e_n}}$ are orthogonal to the kernel of $A_t$), we can assert that $\ket{\vec{y_t}}$ is orthogonal to the kernel of $G_t$. Thus, we may invoke Eq.~\eqref{eq:impact_KR}, and we find that step 2 sends this state to
\begin{equation}\label{eq:post-step2}
    \cos(\theta_t)\ket{\vec{x_t}} -(1-\delta'_1)\sin(\theta_t)\ket{\vec{y_t}} +  \delta'_2 \sin(\theta_t)\ket{\vec{z}}
\end{equation}
where $\delta_1'$ and $\delta_2'$ obey Eq.~\eqref{eq:delta_deltaprime_relations}. The state $\ket{\vec{z}}$ is a unit vector orthogonal to $\ket{\vec{x_t}}$ and $\ket{\vec{y_t}}$, and therefore orthogonal to $\ket{\vec{x}}$ and $\ket{\vec{e_n}}$. 
Step 3 filters out $\ket{\vec{e_n}}$, bringing the state to
\begin{equation}\label{eq:post-step3}
    \ket{\vec{\tilde{x}}} \propto \cos(\theta_t)\sin(\theta_t)(2-\delta'_1)\ket{\vec{x}} + \delta'_2 \sin(\theta_t)\ket{\vec{z}}\,.
\end{equation}

The probability that steps 2 and 3 both succeed, denoted by $p_{\rm succ}$, is given by the squared norm of the subnormalized state on the right-hand side of Eq.~\eqref{eq:post-step3}. Using Eq.~\eqref{eq:delta_deltaprime_relations}, it is seen to satisfy the bounds
\begin{equation}\label{eq:p_succ_bounds_main}
\sin^2(2\theta_t)\frac{\left(1-\eta\right)^2}{\left(1+\eta\right)^2} \leq p_{\rm succ} \leq \sin^2(2\theta_t) + \frac{4\eta^2}{(1+\eta)^2}
\end{equation}
where the lower (upper) bound is generated by replacing $\delta'_1$ with its maximum (minimum) value and $\delta'_2$ with its minimum (maximum) value. The success probability is plotted in \autoref{fig:succ_prob} for $\eta \rightarrow 0$. Examining Eq.~\eqref{eq:post-step3}, the normalized state $\ket{\vec{\tilde{x}}}$ is seen to satisfy
\begin{align}
    \frac{1}{2}\nrm{\ketbra{\vec{x}}-\ketbra{\tilde{\vec{x}}}}_1 &= \sqrt{1-|\braket{\vec{x}}{\vec{\tilde{x}}}|^2} \\
    &= |\delta_2'|\sin(\theta_t)p_{\rm succ}^{-1/2}\\
    &\leq \frac{\eta+ O(\eta^2)}{\cos(\theta_t)}
\end{align}
Actually, a more careful analysis, performed in \autoref{thm:QLSS_known_norm} of \autoref{app:QLSS_with_norm_estimate}, shows that the unspecified $O(\eta^2)$ term in the final line above is not necessary. 

The cost of step 2 is $\ell = \lceil \kappa\ln(2/\eta)/2\rceil$ controlled queries to each of $U_A$ and $U_A^\dagger$, and $2\ell$ controlled queries to each of $U_{\vec{b}}$ and $U_{\vec{b}}^\dagger$, as well as $O(\ell)$ other gates. The cost of step 3 is a single multi-controlled Toffoli gate, which can be used to set an ancilla to 1 only if the state is $\ket{\vec{e_n}}$. See \autoref{fig:main_algo_circuit} in 
\autoref{app:QLSS_with_norm_estimate} for more information on the implementation of the three steps. 

If the norm estimate $t$ is a constant-factor approximation of $\nrm{\vec{x}}$---that is, $t \in [\beta^{-1}\nrm{\vec{x}}, \beta \nrm{\vec{x}}]$ for some $\kappa$-independent constant $\beta = O(1)$,  then $\sin^2(2\theta_t) =4\nrm{\vec{x}}^2t^2/(\nrm{\vec{x}}^2+t^2)^2 = \Omega(1)$ and $\cos(\theta_t) = t/\sqrt{\nrm{\vec{x}}^2+t^2} = \Omega(1)$. Setting $\eta = \Theta(\varepsilon)$, we conclude that the algorithm need only be repeated an expected $O(1)$ number of times to successfully produce an output, and once it does, the output state $\ket{\vec{\tilde{x}}}$ solves the QLSP to error $\varepsilon$. The total expected query complexity is
\begin{equation}
    Q = O(\kappa\log(1/\varepsilon))\,.
\end{equation}
Furthermore, if $t$ is exactly equal to the norm $\nrm{\vec{x}}$, then $\sin(2\theta_t) = 1$ and $\cos(\theta_t) = 1/\sqrt{2}$, so we may choose $\eta = \varepsilon/\sqrt{2}$ and find that the expected total query complexity to $U_A$ and $U_A^\dagger$ is given by 
\begin{equation}
    Q = (1+O(\varepsilon))\kappa \ln(2\sqrt{2}/\varepsilon)\,.
\end{equation}
See \autoref{cor:known_norm_expected_complexity} of \autoref{app:QLSS_with_norm_estimate} for a more precise statement of the complexity in terms of the approximation ratio $\beta$. 

\begin{figure}[h!]
    \centering
    \includegraphics[width=\columnwidth]{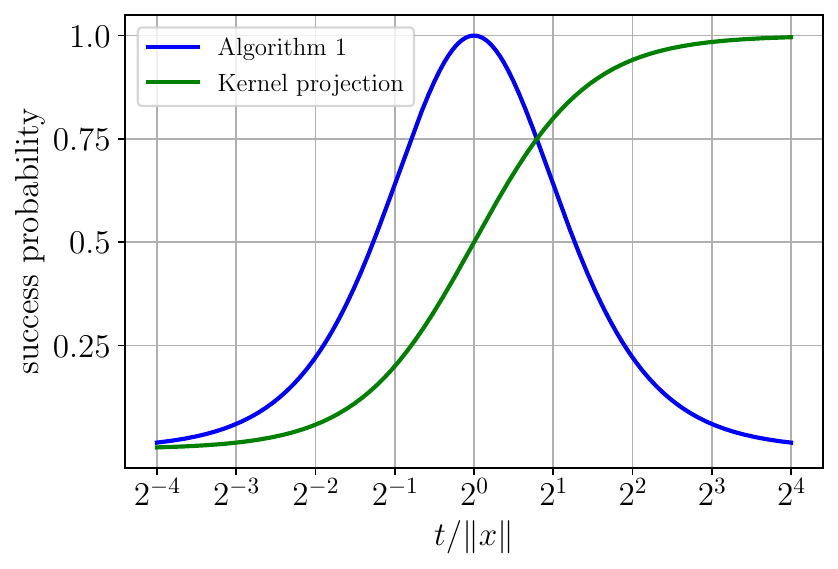}
    \caption{Success probability of \autoref{algo:main_algo} and of the kernel projection protocol of Eq.~\eqref{eq:post-step2-KP}, as a function of the approximation ratio $t/\nrm{\vec{x}}$, in the limit that the precision parameter $\eta \rightarrow 0$.}
    \label{fig:succ_prob}
\end{figure}

\section{Estimating the norm}\label{sec:estimating_the_norm}

To achieve $O(\kappa \log(1/\varepsilon))$  query complexity, \autoref{algo:main_algo} requires that the input parameter $t$ is chosen to approximate $\nrm{\vec{x}}$ up to a constant multiplicative factor. However, $\nrm{\vec{x}}$ is generally not known, other than that it lies in the interval $[1,\kappa]$. Note that the optimal and near-optimal QLSSs based on adiabatic path-following do not offer a straightforward way to also estimate $\nrm{\vec{x}}$---the ability to produce $\ket{\vec{x}}$ in optimal $O(\kappa)$ complexity is not alone sufficient for optimal norm estimation. For example, one could try estimating $\nrm{\vec{x}}$ by applying the block-encoding $U_A$ to the vector $\ket{0}^{\otimes a}\ket{\vec{x}}$, producing the state $\nrm{\vec{x}}^{-1}\ket{0}^{\otimes a}\ket{\vec{b}} +\ket{{\perp}}$, where $\ket{{\perp}}$ is a state for which $(\bra{0}^{\otimes a} \otimes I_{2^s})\ket{{\perp}} = 0$. The norm $\nrm{\vec{x}}$ could then be read out via overlap estimation \cite{knill2007ObservableMeasurement} with the state $\ket{0}^{\otimes a} \ket{\vec{b}}$. However, gaining a constant factor approximation to the overlap in this fashion requires $O(\nrm{\vec{x}})$ queries to the procedure that prepares $\ket{\vec{x}}$, for total complexity $O(\kappa\nrm{\vec{x}})$, which can be as large as $O(\kappa^2)$. 

The QLSS based on the variable-time amplitude estimation technique provides a method to generate a $(1+\varepsilon)$-factor approximation for $\nrm{\vec{x}}$  in near-optimal $O(\kappa \varepsilon^{-1} \log^2(\kappa\varepsilon^{-1})\log^3(\kappa) \log\log(\kappa) )$ complexity \cite[Corollary 32]{chakraborty2018BlockMatrixPowers}. However, this still leaves considerable room for improvement, both in performance and in conceptual transparency. In contrast, our approach offers several immediate near-optimal options for estimating the norm. 

\subsection{Exhaustive search in log space}\label{sec:estimate_norm_exhaustive_search}

The simplest method for norm estimation is an exhaustive search in log space: we can recognize when $t$ is a good approximation of $\nrm{\vec{x}}$ by the fact that \autoref{algo:main_algo} has high success probability---thus, we can simply run \autoref{algo:main_algo} on a geometrically increasing sequence of choices of $t$ until we have found one that leads to success. 

We now provide a concrete implementation and compact proof of correctness for this approach. Suppose we seek a multiplicative $2$-approximation to $\nrm{\vec{x}}$. Equivalently, we may find an additive $\ln(2)$-approximation to $\chi = \ln(\nrm{\vec{x}})$. We consider the set 
\begin{equation}\label{eq:T_candidates}
    \mathcal{T} = \{0, \ln(2), 2\ln(2), \ldots, \lceil \log_2(\kappa) \rceil \ln(2)\}
\end{equation}
Let $\tau^*$ denote the element of $\mathcal{T}$ that is nearest to $\chi$, which is guaranteed to be an additive $\ln(2)/2$-approximation for $\chi$. 
Since $\tau^*$ is an additive $\ln(2)/2$-approximation to $\chi$, we have $\sin^2(2\theta_{e^{\tau^*}}) = \frac{4e^{2\tau^*+2\chi}}{(e^{2\tau^*} + e^{2\chi})^2} \geq 8/9$ (see \autoref{fig:succ_prob}). Now, fix the precision parameter for \autoref{algo:main_algo} to $\eta = 0.025$, ensuring that $(1-\eta)^2/(1+\eta)^2 \geq 9/10$ and $4\eta^2/(1+\eta)^2 \leq 1/100$. 
By Eq.~\eqref{eq:p_succ_bounds_main}, the success probability of \autoref{algo:main_algo} when $t = e^{\tau^*}$ satisfies $p_{\rm succ} \geq (8/9)(9/10) = 0.8$. Meanwhile, for any value of $\tau$ that is more than $\ln(2)$-far from $\chi$, we have $\sin^2(2\theta_{e^{\tau}}) < \frac{16}{25}$, and by Eq.~\eqref{eq:p_succ_bounds_main},  $p_{\rm succ} < \frac{16}{25}+\frac{1}{100} = 0.65$. Having established these bounds, we run \autoref{algo:main_algo} 
\begin{equation}
    k = \lceil 100 \ln(20|\mathcal{T}|) \rceil = O(\log\log(\kappa))
\end{equation}
times for each of the $|\mathcal{T}|$ candidate values of $\tau$, and by observing the fraction of times it succeeds, we compute an estimate $\tilde{p}_{\tau}$ of the success probability for each candidate. By Hoeffding's inequality, for all $\tau \in \mathcal{T}$, it holds that $\Pr[\tilde{p}_\tau - p_{\rm succ} \geq 0.075] \leq e^{-2(0.075)^2k} \leq 1/(20 |\mathcal{T}|)$, and similarly $\Pr[p_{\rm succ}-\tilde{p}_\tau \geq 0.075]$ satisfies the same bound. By the union bound, there is at most $0.05$ probability that either $p_{\tau^*} \leq 0.725$ or that $p_\tau \geq 0.725$ for some candidate $\tau$ that is not an additive $\ln(2)$-approximation to $\chi$. Thus, with probability 0.95, at least one of our estimates will satisfy $\tilde{p}_\tau > 0.725$ and if we output $t=e^\tau$ for the first $\tau$ we find for which this is the case, that value of $t$ will be a multiplicative $2$-approximation to $\nrm{\vec{x}}$. This proves that the output of the algorithm is a multiplicative $2$-approximation to $\nrm{\vec{x}}$ with high probability. 

The total query complexity to controlled-$U_A$ and controlled-$U_A^\dagger$ is 
\begin{equation}\label{eq:Q_exhaustive_nonoptimized}
    Q = 2k|\mathcal{T}|\lceil \kappa\ln(2/\eta)/2\rceil = O(\kappa \log(\kappa) \log\log(\kappa))\,.
\end{equation}
In \autoref{lem:estimate_norm_random_t} of \autoref{app:estimating_norm}, we provide a more formal analysis of a slightly different implementation that achieves the same complexity, with better constant prefactors on the complexity. 
This query complexity can be improved to $O( \kappa \sqrt{\log(\kappa)}\log\log(\kappa))$ by straightforwardly replacing the exhaustive search above with a Grover search \cite{grover1996QSearch,yoder2014FixedPointSearch}, which we analyze in \autoref{lem:estimate_norm_random_t_FPAA} of \autoref{app:estimating_norm}. 

\subsection{Binary search in log space}\label{sec:estimate_norm_binary_search}

Furthermore, with a bit more thought, the query complexity can be made even closer to linear-in-$\kappa$ using a binary search for $\tau = \ln(t)$. Here, rather than running \autoref{algo:main_algo}, we use KP to detect if a candidate value of $\tau$ is too large or too small. Specifically, suppose we prepare the state $\ket{\vec{e_n}}$, as in step 1 of \autoref{algo:main_algo}, and then we run step 2 using KP in place of KR. Invoking Eq.~\eqref{eq:impact_KP}, we arrive at the state (cf.~Eq.~\eqref{eq:post-step2})
\begin{equation}\label{eq:post-step2-KP}
    \cos(\theta_t)\ket{\vec{x_t}} + \sin(\theta_t)\delta_1 \ket{\vec{y_t}} + \sin(\theta_t)\delta_2 \ket{\vec{z}}
\end{equation}
where $\sqrt{\delta_1^2+\delta_2^2} \leq \eta$. The success probability $q_{\rm succ}$ of KP is given by the squared norm of the above state, which satisfies
\begin{equation}\label{eq:succ_prob_KP}
    \cos^2(\theta_t) \leq q_{\rm succ} \leq \cos^2(\theta_t)  + \eta^2\sin^2(\theta_t)\,.
\end{equation}
The function $\cos^2(\theta_t)$ is plotted in \autoref{fig:succ_prob}; it is monotonically increasing on the relevant interval and equal to 1/2 exactly at $t=\nrm{\vec{x}}$. Thus, estimating the success probability $q_{\rm succ}$ of KP for a certain choice of $\tau = \ln(t)$ offers a mechanism for determining whether $\tau \geq \chi$ or $\tau < \chi$, enabling a (noisy) binary search.

We now provide a concrete implementation for this noisy binary search. Fix $\eta = \sqrt{1/8}$. We maintain a set $\mathcal{S}$ of ``active'' candidates, and initially set $\mathcal{S} = \mathcal{T}$. We choose $\tau$ to be the median of $\mathcal{S}$, and we round the median to the closest element of $\mathcal{S}$, breaking ties arbitrarily. We run KP 
\begin{equation}
    k = \lceil 72 \ln(40\lceil\log_{3/2}(|\mathcal{T}|)\rceil )\rceil  = O(\log\log\log(\kappa))
\end{equation}
times with $t=e^{\tau}$, and we compute an estimate $\tilde{q}$ for $q_{\rm succ}$ based on the fraction of these $k$ trials that succeed. If $\tilde{q} > 1/2$, we eliminate all elements of $\mathcal{S}$ less than $\tau$; otherwise we eliminate all elements greater than $\tau$. 
As long as $|\mathcal{S}|>2$, we are guaranteed to eliminate at least $1/3$ fraction of the elements. We repeat the process of choosing $\tau$ to be the median of $\mathcal{S}$, estimating $\tilde{q}$, and eliminating part of $\mathcal{S}$, a total of $\lceil \log_{3/2}|\mathcal{T}|\rceil$ times, which reduces the size of $\mathcal{S}$ to at most 2. If $\mathcal{S} = \{\tau\}$, we output $t = e^{\tau}$ as the estimate for $\nrm{\vec{x}}$. If $\mathcal{S} = \{\tau,\tau+\ln(2)\}$, we output $t = e^{\tau+\ln(2)/2}$. 

To show correctness, we argue that with high probability, at every step of the algorithm there is at least one element of the active set $\mathcal{S}$ that is an additive $\ln(2)/2$-approximation for $\chi$. Hence, in either case ($\mathcal{S} = \{\tau\}$ or $\mathcal{S} = \{\tau, \tau+\ln(2)\}$) the output is a multiplicative $2$-approximation for $\nrm{\vec{x}}$. 
Specifically, by Hoeffding's inequality, $\Pr[|\tilde{q} - q_{\rm succ}| \geq 1/12] \leq 2e^{-2k/144}\leq 1/(20\lceil\log_{3/2}(|\mathcal{T}|)\rceil)$, and by the union bound, there is at most 0.05 probability that any of the estimates across all $\lceil\log_{3/2}(|\mathcal{T}|)\rceil$ steps have error more than $1/12$. If $\tau$ is not an additive $\ln(2)/2$ approximation of $\nrm{\vec{x}}$, then by Eq.~\eqref{eq:succ_prob_KP}, $q_{\rm succ}$ lies outside the interval $[10/24, 2/3]$. Thus, in the 0.95 fraction of cases where $|\tilde{q} - q_{\rm succ}| < 1/12$ holds in all steps, the search always correctly decides whether $\tau > \chi$ or $\tau \leq \chi$, unless $\tau$ is an additive $\ln(2)/2$-approximation to $\chi$. However, since we always keep $\tau$ in the set (eliminating only elements that are larger than or smaller than $\tau$), we conclude that, whether or not $\tau$ is an additive $\ln(2)/2$-approximation, the active set $\mathcal{S}$ is guaranteed to retain at least one element that is an additive $\ln(2)/2$-approximation of $\chi$. This proves that the output of the algorithm is a multiplicative $2$-approximation to $\nrm{\vec{x}}$ with high probability. 

The overall query complexity is
\begin{align}
    Q &= 2k\lceil\log_{3/2}(|\mathcal{T}|)\rceil \lceil \kappa \ln(2/\eta)/2\rceil\\
    &=O(\kappa \log\log(\kappa)\log\log\log(\kappa))
\end{align}

More sophisticated algorithms for noisy binary search may be able to eliminate the $\log\log\log(\kappa)$ factor. This large body of work (e.g., \cite{burnashev1974interval,ben-or2008BayesianLearner,wang2022noisySortingcapacity,gretta2023sharpNoisyBinarySearch}) has established that the optimal approach in settings like this one is to maintain a belief distribution over the possible estimates $\tau$, initially uniform. Each step of the algorithm chooses $\tau$ to be the median of the distribution and queries whether $\tau > \chi$. This query returns a noisy answer, and the algorithm responds by updating its belief distribution according to Bayes' rule. This approach may allow us to exploit the additional information we have in our setting, namely, that the probability that KP succeeds follows (up to $O(\eta^2)$ precision) a known functional form, depicted in \autoref{fig:succ_prob}.

\subsection{Achieving linear-in-$\kappa$ complexity}\label{sec:estimate_by_adiabatic}
The methods sketched above do not depend on sophisticated techniques or analytical methods, other than QSVT with filtering polynomials, yet they achieve \textit{nearly} linear-in-$\kappa$ asymptotic scaling, superior to previous known methods for estimating the norm. 

We can also combine the observations above with some of the ideas behind adiabatic path-following to achieve \textit{strictly} linear-in-$\kappa$ complexity. Specifically, we propose to estimate the norm of a sequence of linear systems of increasing condition number, where the norm of the solution to the final linear system in the sequence equals $\nrm{\vec{x}}$. The linear systems are related in such a way that the norm cannot change by more than a constant factor from one step to the next. Thus, if we have an estimate of the norm for one linear system, we need only search over a constant-sized interval for the norm of the next linear system---this eliminates the $\mathrm{polylog}(\kappa)$ factors from the norm estimation methods proposed in \autoref{sec:estimate_norm_exhaustive_search} and \autoref{sec:estimate_norm_binary_search} (specifically, the same analysis goes through with $|\mathcal{T}| = O(1)$ rather than $|\mathcal{T}| = O(\log(\kappa))$).

The sequence of linear systems is parameterized by $\sigma \in [\kappa^{-1},1]$ and inspired by the adiabatic path followed by the QLSS in Ref.~\cite{subasi2019QAlgSysLinEqsAdiabatic} and successors. 
To define it, let
\begin{equation}
    f(\sigma) = \sqrt{\frac{\sigma^2\kappa^2-1}{\kappa^2-1}}\,,
\end{equation}
chosen so that $f(\sigma)^2+ (1-f(\sigma)^2)\kappa^{-2} = \sigma^2$. Observe that $f(\sigma)$ is monotonically increasing for $\sigma \in [\kappa^{-1},1]$ with $f(\kappa^{-1}) = 0$ and $f(1) =1$. 
Then, we construct a linear system 
\begin{equation}
    \bar{A}_\sigma \vec{\bar{x}_\sigma} = \vec{b}\,.
\end{equation}
Here for clarity we interpret $\vec{b}$ as a length-$2^{s}$ vector by padding it with $2^s-m$ zeros. Let $\bar{n}=2^s-n$ and $\bar{m} = 2^s-m$ and define the $2^{s} \times 2^{s+1}$ matrix $\bar{A}_\sigma$ as follows
\NiceMatrixOptions{code-for-first-row = \scriptstyle,code-for-first-col = \scriptstyle }
\begin{equation}\label{eq:barA_sigma}
\bar{A}_\sigma =\;\; \begin{bNiceArray}{cc|cc}[first-row, first-col]
    & n                     & \bar{n} & m & \bar{m} \\
m   & \sqrt{1-f(\sigma)^2}A & 0     &  f(\sigma) I_m &  0  
\\
\bar{m}   & 0 & 0     &  0 &  0  
\end{bNiceArray}
\end{equation}
where the sizes of each block are included for convenience. In essence, $\bar{A}_\sigma$ is really an $m \times (n+m)$ matrix, but the explicit zero padding in Eq.~\eqref{eq:barA_sigma} ensures that the two parts of the matrix separated by the vertical line are each square $2^s \times 2^s$ blocks. This convention makes it more natural to construct a block-encoding for $\bar{A}_\sigma$ on an $(s+1)$-qubit system. In particular, if we further zero-pad $\bar{A}_\sigma$ to be a square $2^{s+1} \times 2^{s+1}$ matrix, we may identify it with the $(s+1)$-qubit operator equivalent to
\begin{equation}
    \sqrt{1-f(\sigma)^2}\ketbra{0} \otimes A + f(\sigma) \ketAbraB{0}{1} \otimes I_m
\end{equation}
where $I_m$ is the operator for which $I_m\ket{\vec{e_j}} = \ket{\vec{e_j}}$ if $j < m$ and $I_m\ket{\vec{e_j}} = 0$ otherwise (although the following results would still hold if $I_m$ were replaced with $I_{2^s}$). 
Under this identification, we give a block-encoding for $\bar{A}_\sigma$ using one controlled-query to $U_A$ in \autoref{app:block_encodings}. Unlike previous methods \cite{jennings2023QLSS}, we do not use the standard LCU technique to combine the two terms, which would have given a normalization factor $\sqrt{1-f(\sigma)^2} + f(\sigma) >1$ when $\sigma \in (\kappa^{-1},1)$. Note also that here we avoid relying upon the Hermitianized version $\left[\begin{smallmatrix}0 & A \\ A^\dagger & 0\end{smallmatrix}\right]$ of the matrix $A$, as in prior work, which saves a factor of two on the query complexity.

Intuitively, the purpose of $\bar{A}_\sigma$ is that by tuning the value of $\sigma$, we can introduce information about the matrix $A$ and its spectrum gradually: when $\sigma = 1$, $\bar{A}_\sigma$ has no dependence on $A$, and when $\sigma = \kappa^{-1}$, it has full dependence. Crucially, the condition number of $\bar{A}_\sigma$ grows as $\sigma$ decreases; it can be upper bounded by $\sigma^{-1}$ (as justified later). Thus, as $\sigma$ is reduced, we can begin to extract some information about $\nrm{\vec{x}}$ without paying the full $O(\kappa)$ cost associated with the matrix $A$. 

We can derive an explicit expression for the solution $\vec{\bar{x}_\sigma}$ to the equation $\bar{A}_\sigma \vec{\bar{x}_\sigma} = \vec{b}$ in terms of the singular value decomposition (SVD) of $A$. Specifically, let $A = \sum_j \varsigma_j \vec{u_j}\vec{v_j}^\dagger$ be a SVD of $A$, with $\varsigma_j \neq 0$ for all $j$, and interpreting $\vec{u_j}$ and $\vec{v_j}$ as length-$2^s$ normalized vectors. Let $\vec{b} = \sum_j w_j\vec{u_j}$ be the decomposition of $\vec{b}$ into left singular vectors of $A$, which is guaranteed to exist since we have assumed that $A\vec{x} = \vec{b}$ has a solution. Define the length-$2^{s+1}$ normalized vector
\begin{equation}
    \vec{\bar{v}_{\sigma,j}} = \frac{1}{\sqrt{(1-f(\sigma)^2)\varsigma_j^2 + f(\sigma)^2}}\begin{bmatrix}
        \sqrt{1-f(\sigma)^2}\varsigma_j\vec{v_j}\\
        f(\sigma) \vec{u_j} 
    \end{bmatrix}
\end{equation}
Since the sets $\{\vec{u_j}\}$ and $\{\vec{v_j}\}$ are each orthonormal, the set $\{\vec{\bar{v}_{\sigma,j}}\}$ is also orthonormal. In fact, it can be verified that $\{\vec{\bar{v}_{\sigma,j}}\}$ are the right singular vectors of $\bar{A}_\sigma$, and the corresponding left singular vector is $\vec{u_j}$ and correpsonding singular value is $\sqrt{f(\sigma)^2 + (1-f(\sigma)^2)\varsigma_j^2}$. 
Thus, by applying the inverted singular values to the decomposition of $\vec{b}$, we may assert that
\begin{equation}\label{eq:barx_sigma}
    \vec{\bar{x}_\sigma} = \sum_j \frac{w_j}{\sqrt{f(\sigma)^2+(1-f(\sigma)^2)\varsigma_j^2}} \vec{\bar{v}_{\sigma,j}}\,.
\end{equation}
Matrix-vector block multiplication of $\bar{A}_\sigma \vec{\bar{x}_\sigma}$ yields $\vec{b}$, verifying that $\vec{\bar{x}_\sigma}$ is a solution to the system.  The fact that $\vec{\bar{x}_\sigma}$ is the solution of minimal Euclidean norm follows from the fact that it is orthogonal to the kernel of $\bar{A}_\sigma$ (it is a linear combination of right singular vectors of $\bar{A}_\sigma$ associated with nonzero singular values). 

To apply the framework described in this paper to the matrix $\bar{A}_\sigma$, we will need to augment the linear system, as described in \autoref{sec:QLSS_given_norm_main}. In this instance, we augment to form the matrix $\bar{A}_{\sigma,t} = \bar{A}_{\sigma} + t^{-1}\vec{e_m}\vec{e_n}^\dagger$ (cf.~Eq.~\eqref{eq:A_t})
\begin{equation}
\bar{A}_{\sigma,t}= \hspace{-8pt} \begin{bNiceArray}{ccc|cc}[first-row, first-col]
    & n                     & 1 & \bar{n}-1 & m & \bar{m} \\
m   & \sqrt{1-f(\sigma)^2}A & 0 & 0     &  f(\sigma) I_m &  0  \\
1   & 0    &  t^{-1} & 0 & 0 & 0 \\
\bar{m}-1& 0 & 0 &     0 & 0& 0
\end{bNiceArray}
\end{equation}
for which we have the linear system $\bar{A}_{\sigma,t}\vec{\bar{x}_{\sigma,t}} = \vec{b'}$. This choice of augmentation avoids the need to expand the Hilbert space and use more than $s+1$ qubits. The relationship between $\vec{\bar{x}_{\sigma, t}}$ and $\vec{\bar{x}_\sigma}$ is the same as the relationship between $\vec{x}_t$ and $\vec{x}$, enabling the methods from prior sections to be directly applied.

Before proceeding, we must establish a few additional properties of the linear system. 
\begin{enumerate}
    \item All nonzero singular values of $\bar{A}_\sigma$ lie in the interval $[\sigma,1]$.
    \item $\nrm{\vec{\bar{x}}_1} = 1$ and $\nrm{\vec{\bar{x}_{1/\kappa}}} = \nrm{\vec{x}}$.
    \item For all $\kappa^{-1} \leq \sigma \leq  \sigma'\leq 1$,  it holds that $1 \leq \nrm{\vec{\bar{x}_\sigma}}/\nrm{\vec{\bar{x}_{\sigma'}}} \leq \sigma'/\sigma$.
\end{enumerate}
Property 1 follows from the fact that $\bar{A}_\sigma\bar{A}_\sigma^\dagger= f(\sigma)^2 I_m + (1-f(\sigma)^2)AA^\dagger$ has all nonzero eigenvalues in $[\sigma^2,1]$ by construction. 
Property 2 follows by inspection of Eq.~\eqref{eq:barx_sigma}, and the fact that $f(1)=1$ and $f(1/\kappa) = 0$. In particular, we see that $\nrm{\vec{\bar{x}_{1/\kappa}}}^2 = \sum_j |w_j|^2\varsigma_j^{-2} =\nrm{\vec{x}}^2$.
To verify property 3, first note 
that due to the monotonicity of $f$ and the fact that $\kappa^{-1} \leq \varsigma_j \leq 1$, we have
\begin{align}
    1 &\leq \frac{(f(\sigma')^2+(1-f(\sigma')^2)\varsigma_j^2)^{1/2}}{(f(\sigma)^2+(1-f(\sigma)^2)\varsigma_j^2)^{1/2}} \\
    &= \frac{(\sigma'^2 + (1-f(\sigma')^2)(\varsigma_j^2-\kappa^{-2}))^{1/2}}{(\sigma^2 + (1-f(\sigma)^2)(\varsigma_j^2-\kappa^{-2}))^{1/2}} \leq \frac{\sigma'}{\sigma}
\end{align} 
for any $\sigma\leq \sigma'$. Thus, working from Eq.~\eqref{eq:barx_sigma}, we have
\begin{align}
    \frac{\nrm{\vec{\bar{x}_\sigma}}}{\nrm{\vec{\bar{x}_{\sigma'}}}} &= \frac{\sqrt{\sum_j |w_j|^2 (f(\sigma)^2+(1-f(\sigma)^2)\varsigma_j^2)^{-1}}}{\sqrt{\sum_j |w_j|^2 (f(\sigma')^2+(1-f(\sigma')^2)\varsigma_j^2)^{-1}}} \nonumber \\
    &\leq \frac{\sqrt{\sum_j |w_j|^2 (f(\sigma)^2+(1-f(\sigma)^2)\varsigma_j^2)^{-1}}}{\sqrt{\sum_j |w_j|^2\frac{\sigma^2}{\sigma'^2} (f(\sigma)^2+(1-f(\sigma)^2)\varsigma_j^2)^{-1}}} \nonumber\\
    &= \frac{\sigma'}{\sigma}  
\end{align}

Now we are ready to describe the algorithm. We propose to learn $\nrm{\vec{x}}  = \nrm{\vec{\bar{x}_{1/\kappa}}}$ by sequentially estimating the norm of $\nrm{\vec{\bar{x}_{2^{-j}}}}$ for $j=0,1,\ldots,\log_2(\kappa)$ (for simplicity, here we round $\kappa$ up to the closest exact power of 2). For $j=0$, we use the estimate $t_0 = 1$, which is exact, by property 2 above. We then aim to generate estimates $t_1,t_2,\ldots, t_{\log_2(\kappa)}$ such that each $t_j$ is a multiplicative $2$-approximation of $\nrm{\vec{\bar{x}_{2^{-j}}}}$. 

The key insight is that, from property 3, if $t_{j-1}$ is a valid 2-approximation, i.e.~$\nrm{\vec{\bar{x}_{2^{-j+1}}}} \in [t_{j-1}/2,2t_{j-1}]$, then it must hold that $\nrm{\vec{\bar{x}_{2^{-j}}}} \in [t_{j-1}/2,4t_{j-1}]$. We can find a 2-approximation $t_j$ to 
$\nrm{\vec{\bar{x}_{2^{-j}}}}$ by exhaustively (or binarily) searching $O(1)$ candidates in this interval---that is, whereas the analysis in prior subsections required the size of the set $\mathcal{T}$ of candidates to grow as $O(\log(\kappa))$ (see Eq.~\eqref{eq:T_candidates}), here we have $|\mathcal{T}| = O(1)$ instead.  
Furthermore, the cost of running KR or KP for a particular candidate is reduced due to the fact that the condition number of $\bar{A}_{2^{-j}}$ is only $2^{j}$, as seen in property 1. 

Ultimately, from Eq.~\eqref{eq:Q_exhaustive_nonoptimized}, the expected query complexity of producing the estimate $t_j$ using the exhaustive search method is $O(2^j|\mathcal{T}| \log(|\mathcal{T}|)) = O(2^j)$, and the probability that $t_j$ is not a valid $2$-approximation---conditioned on $t_{j-1}$ being a valid 2-approximation---is at most $0.05$. To suppress this conditional failure probability from $0.05$ to $\delta$, we perform median amplification \cite[Lemma 1]{nagaj2009FastAmpQMA} at multiplicative overhead $O(\log(1/\delta))$ to the query complexity. Naively, we might choose $\delta=O(1/\log(\kappa))$ so that there is high probability that all $\log_2(\kappa)$ steps succeed; however, if we do so then the final step, for which the condition number is $\kappa$, will have complexity $O(\kappa \log(1/\delta)) = O(\kappa \log\log(\kappa))$ complexity.  

Instead, we employ a version of the ``log log trick'' \cite{kothari2023meanEstimationSourceCode}, choosing $\delta$ to be smaller when $j$ is smaller---we can afford to do more median amplification at the beginning of the sequence because the condition number is smaller and the norm estimation is cheaper. Specifically, we use median amplification to suppress the failure probability of step $j$ to $e^{-\Omega(1+\log_2(\kappa)-j)}$, at multiplicative overhead $O(1+\log_2(\kappa)-j)$, and we conclude that the overall expected query cost of step $j$ is $O( 2^{j}(1+\log_2(\kappa)-j))$. By the union bound, the probability that at least one of the steps fails to generate a $2$-approximation is upper bounded by a geometrically decaying series $\sum_{j=1}^{\log_2(\kappa)} e^{-\Omega(1+\log_2(\kappa)-j)}$, which is at most a small constant. Thus, with high probability all of the steps succeed, and the output is a multiplicative 2-approximation for $\nrm{\vec{x_{1/\kappa}}} = \nrm{\vec{x}}$. 

The overall total query complexity is upper bounded by the sum of the query complexity of the $\log_2(\kappa)$ steps, given by
\begin{align}
 Q&=\sum_{j=1}^{\log_2(\kappa)} O(2^j(1+\log_2(\kappa)-j)) \\
 &\leq  O(\kappa) \sum_{j'=0}^{\infty } 2^{-j'}(1+j')
 \leq O(\kappa)\label{eq:Q_adiabatic_norm_optimal}
\end{align}
where we have substituted $j' = \log_2(\kappa)-j$ and upper bounded the finite sum by its infinite extension. 

In summary, the method solves the norm estimation problem up to a constant factor with optimal $O(\kappa)$ query complexity. It does so by leveraging an idea that is inspired by adiabatic path-following; in fact, it follows essentially the same family of linear systems as prior adiabatic solvers. However, the analysis is fairly simple, and crucially it avoids the need for rigorous analysis of the adiabatic theorem.  

By first estimating the norm with this method and then running \autoref{algo:main_algo} to produce $\ket{\vec{\tilde{x}}}$, we obtain a method for solving the QLSP with overall query complexity $O(\kappa\log(1/\varepsilon))$. In \autoref{app:explicit_optimal_QLSS} of the appendix, we give a more detailed specification of an optimal QLSS with $O(\kappa\log(1/\varepsilon))$ complexity, which follows the same general approach sketched above, and we analyze its complexity, including constant prefactors, in \autoref{thm:optimal_QLSS}. 

\paragraph{Remark on relationship to Zeno method.} The QLSS in Ref.~\cite{lin2019OptimalQEigenstateFiltering} also gave nearly optimal-in-$\kappa$ complexity of $O(\kappa \log(\kappa) \log\log(\kappa))$ by leveraging the \textit{quantum Zeno effect} to traverse an adiabatic eigenpath. This bears resemblance to our method, and in particular, its analysis also does not require the adiabatic theorem. Whereas our method varies $\sigma$ and tracks changes in the scalar value of the solution norm, the Zeno method can be understood (in our notation) as tracking the state $\ket{\vec{\bar{x}_{\sigma}}}$ itself. The key tool is the ability to project onto $\ket{\vec{\bar{x}_{\sigma}}}$ up to accuracy $\eta$ with KP (i.e., eigenstate filtering), at cost $O(\sigma^{-1} \log(1/\eta))$. The original analysis in Ref.~\cite{lin2019OptimalQEigenstateFiltering} takes short steps so that each projection succeeds with probability close to 1. Amplitude amplification is not required, but this leads the final complexity to be off from optimal by at least a $\log(\kappa)$ factor. As remarked in Section 5 of Ref.~\cite{lin2019OptimalQEigenstateFiltering}, one could modify the Zeno algorithm to take a fewer number of larger jumps---each jump then succeeds with constant success probability, and one can use fixed-point amplitude amplification to boost this probability to 1 before moving on to the next jump. The (approximate) reflection operators needed within the amplification procedure could be supplied by KR.  In \autoref{app:Zeno}, we show that the idea to use fixed-point amplitude amplification works, and together with the log log trick leads to optimal $O(\kappa)$ complexity for the Zeno approach when following the same adiabatic path $f(\sigma)$ defined above. In other words, all of the essential technical tools and conceptual ingredients needed for strictly optimal-in-$\kappa$ complexity have been in place since at least 2019. 

While we have not done a rigorous analysis of the constant prefactors of the optimal Zeno method, we expect that they would be strictly worse than the optimal QLSS proposed in this paper. The main reason is that an optimized version of our norm estimation method can take \textit{very} large (constant-sized) jumps $\sigma' \mapsto \sigma$ along the adiabatic path, since (using the binary search method) the norm estimation subroutine has cost $O(\log\log(\sigma'/\sigma))$ calls to KP. In contrast, the Zeno method requires a number of calls to KR going as $\lvert\braket{\vec{\bar{x}_{\sigma'}}}{\vec{\bar{x}_{\sigma}}}\rvert^{-1}$ for fixed-point amplitude amplification, and the (potentially suboptimal) bound we use in \autoref{app:Zeno} suggests this quantity goes roughly as $\sigma' / \sigma$. That is, the complexity of a jump is doubly exponentially worse in the Zeno method compared to the norm estimation method, as a function of the length of the jump $\sigma' / \sigma$. Consequently, we expect that an optimized Zeno method would either take smaller steps and/or spend larger complexity per step compared to our method, leading to larger constant prefactors.  

\subsection{Improving the approximation ratio }\label{sec:improving_approx_ratio}
We have now seen several methods for obtaining a constant-factor approximation to the norm $\nrm{\vec{x}}$. However, in many instances we may wish to improve the approximation ratio of our estimate. For one, the success probability of \autoref{algo:main_algo} increases (approaching 1) as the accuracy of the norm estimate improves, motivating us to seek the best possible estimate. Additionally, some applications require not only access to the normalized state $\ket{\vec{x}}$, but also a precise estimate for $\nrm{\vec{x}}$ that is correct up to small relative error---see \autoref{sec:conclusion}. 

Given a multiplicative 2-approximation $t \in [\nrm{\vec{x}}/2, 2\nrm{\vec{x}}]$, one can improve the estimate to a multiplicative $(1+\varepsilon)$-approximation 
\begin{equation}
    t_{\rm out} \in [(1+\varepsilon)^{-1}\nrm{\vec{x}}, (1+\varepsilon)\nrm{\vec{x}}]
\end{equation}
using amplitude estimation \cite{brassard2002AmpAndEst} (the method is easily generalized to handle input approximations with constant-factors worse than 2), as we now explain. Amplitude estimation allows one to estimate the quantity $\nrm{\Pi \ket{\psi}}^2$ to additive precision $\nu$, using $O(1/\nu)$ calls to a procedure for preparing $\ket{\psi}$ and the ability to recognize when a state is in the image of a projector $\Pi$ \cite{brassard2002AmpAndEst}.
Let $\eta = \sqrt{\varepsilon/100}$ and, as in \autoref{sec:estimate_norm_binary_search}, consider the procedure that prepares $\ket{\vec{e_n}}$ and applies KP for the matrix $G_t$. The success probability $q$ of this procedure can be cast as a measurable quantity $\nrm{\Pi \ket{\psi}}^2$, and by Eq.~\eqref{eq:succ_prob_KP}, it lies within $\eta^2 = \varepsilon/100$ of the quantity 
\begin{equation}
    q_t = \cos^2(\theta_t) =  t^2/(t^2+\nrm{\vec{x}}^2) \in [0.2,0.8]\,.
\end{equation}
Using $O(1/\varepsilon)$ queries to this KP procedure, amplitude estimation can produce an estimate $\tilde{q}$ for $q$ that is correct up to additive  error $\varepsilon/100$, with probability 0.95. In this 0.95 fraction of cases, by the triangle inequality, we have $|\tilde{q}-q_t| \leq \varepsilon/50$. Furthermore, defining the relative error as $\Delta = \tilde{q}/q_t - 1$, we have $|\Delta| \leq \varepsilon/10$.  Then, to produce an estimate for $\nrm{\vec{x}}$, one can take the result $\tilde{q}$ of amplitude estimation and output the quantity
\begin{align}
    t_{\rm out} &= \frac{t\sqrt{1-\tilde{q}}}{\sqrt{\tilde{q}}} = \frac{t\sqrt{1-q_t(1+\Delta)}}{\sqrt{q_t(1+\Delta)}} \nonumber \\
    &= \frac{t\sqrt{\frac{\nrm{\vec{x}}^2}{t^2}-\Delta}}{\sqrt{1+\Delta}} 
    = \nrm{\vec{x}}\frac{\sqrt{1-\frac{t^2\Delta}{\nrm{\vec{x}}^2}}}{\sqrt{1+\Delta}}\,.
\end{align}
Since $t^2/\nrm{\vec{x}}^2 \leq 4$ and $|\Delta| \leq \varepsilon/10$, the output is guaranteed to be multiplicative $(1+\varepsilon)$-approximation for $\nrm{\vec{x}}$. 
Accounting for $O(1/\varepsilon)$ calls to KP, each costing $O(\kappa\log(1/\eta))$, the total complexity is $O(\kappa\log(1/\varepsilon)/\varepsilon)$. 

Since the query complexity of obtaining the $2$-approximation in the first place is $O(\kappa)$ using the method from \autoref{sec:estimate_by_adiabatic}, we have thus established that the total query complexity to obtain a multiplicative $(1+\varepsilon)$-approximation of $\nrm{\vec{x}}$ is 
\begin{equation}
    Q = O(\kappa\log(1/\varepsilon)/\varepsilon)
\end{equation}
In \autoref{thm:norm_query_lower_bound} of \autoref{app:norm_query_lower_bound}, building on the method of Ref.~\cite{Orsucci2021solvingclassesof}, we derive a lower bound on the query complexity of this task when $\varepsilon < 1/4$ of 
\begin{equation}
    Q = \Omega(\min(\kappa/\varepsilon,N))\,,
\end{equation}
where $N$ is the size of the matrix $A$. Thus, we have established optimal-in-$\kappa$ query complexity for norm estimation, and our $\varepsilon$ dependence is off by a single logarithmic factor. It is an interesting question for future work to tighten the $\varepsilon$ dependence of the complexity of norm estimation. 

\section{Constant prefactor analysis}\label{sec:constant_factors}

\begin{table*}[t]
    \centering
    \begin{adjustbox}{max width=\textwidth}

\begin{tabular}{|c|c|c|c|c|}
    \hline
    & \textbf{QLSS method} & \makecell{\textbf{Asymptotic complexity} \\ $ + \;O(\kappa \log(1/\varepsilon)) $} & \makecell{\textbf{Explicit upper bound} \\ $(\kappa \in [3,10^6])$} & \makecell{\textbf{Value at} \\ $(\kappa,\varepsilon) =(10^5,10^{-10})$} \\
    \hline
    \hline
    \multirow{3}{*}{\rotatebox[origin=c]{90}{other work}} & Quantum walk method \cite{costa2021OptimalLinearSystem} & $O(\kappa)$ & $234470\kappa + 4\kappa \ln(2/\varepsilon)$ & $234562\kappa$\\
    \cline{2-5}
    & Randomization method \cite{subasi2019QAlgSysLinEqsAdiabatic} & $O(\kappa\log(\kappa))$ & $162\kappa\ln(\kappa) + 188\kappa + 5.2\kappa\ln(1/\varepsilon)$ \cite[arXiv v1]{jennings2023QLSS} & $2173\kappa$ \\
    \cline{2-5}
    & \makecell{Randomization method \\ with Poissonization \cite{cunningham2024eigenpathTraversal, jennings2023QLSS}} & $O(\kappa)$ & $1671\kappa + 4.2\kappa + 2.0\kappa\ln(1/\varepsilon)$ \cite{jennings2023QLSS} & $1728\kappa$ \\
    \hline
    \hline
    \multirow{4}{*}{\rotatebox[origin=c]{90}{\makebox[2mm][c]{this work \;\;\;}}} & \makecell{Augmented linear system KR \\ with exhaustive norm search } & $O(\kappa \log(\kappa)\log\log(\kappa))$ & $6 \kappa \ln(\kappa) +6\kappa + 1.1\kappa\ln(1/\varepsilon)$ & $83\kappa$\\
    \cline{2-5}
    & \makecell{Augmented linear system KR \\ with Grover norm search} & $O(\kappa \sqrt{\log(\kappa)}\log\log(\kappa))$ & $10 \kappa \sqrt{\ln(\kappa)+1} +12\kappa + 1.1\kappa\ln(1/\varepsilon)$ & $66\kappa$\\
    \cline{2-5}
    & \makecell{Augmented linear system KR \\ with binary norm search } & $O(\kappa\log\log(\kappa)\log\log\log(\kappa))$ & \text{not analyzed} & not analyzed \\
    \cline{2-5}
    & \makecell{Augmented linear system KR \\ with adiabatic norm search } & $O(\kappa)$ & $56\kappa + 1.05\kappa\log(1/\varepsilon)$ & $80\kappa$ \\
    \hline
\end{tabular}
    \end{adjustbox}
    \caption{A comparison of our proposed QLSS methods with previous optimal and near-optimal methods for which explicit query complexity bounds have been derived. The explicit upper bounds omit subleading terms, and the rows for Grover and exhaustive search replace doubly and triply logarithmic factors with their maximum values on the domain $\kappa \in [3,10^6]$ for compactness---see \autoref{app:explicit_near_optimal_QLSS} for more detailed expressions. In all cases, the evaluation of the bound at $\kappa=10^5$ and $\varepsilon =10^{-10}$ is calculated using the exact expression. For the quantum walk method, a numerical evaluation of the upper bound expression on the domain $\kappa \in [1,50]$ suggested the bound is about 20 times smaller, but here we report the fully rigorous bound \cite{costa2021OptimalLinearSystem}. Additionally, the estimates for the quantum walk and randomization method include an additional factor of 2 compared to some of the numerical values reported elsewhere (e.g., in Ref.~\cite{costa2021OptimalLinearSystem}) because when $A$ is non-Hermitian they require two queries to $U_A$ to block-encode the Hermitianized $\left[\begin{smallmatrix} 0 & A \\ A^\dagger & 0 \end{smallmatrix}\right]$. We expect the practical performance of all of  these methods to be significantly better than their bounds. Finally, note that the result from the second row comes from arXiv v1 of Ref.~\cite{jennings2023QLSS}, which appeared prior to the first version of this paper. The third row comes from the published version of Ref.~\cite{jennings2023QLSS}, which appeared after. }\label{tab:comparison}
\end{table*}

The constant prefactors and explicit complexity expressions for the QLSS are important for determining whether applications based on the QLSS are viable. 
We derived the asymptotic complexities of several distinct norm estimation methods in \autoref{sec:estimating_the_norm}; the ability to estimate the norm up to a constant multiplicative factor implies a full-fledged QLSS at the cost of $O(\kappa\log(1/\varepsilon))$ additional complexity.  However, the implementations and proofs provided for these algorithms were optimized for simplicity and ease of presentation, rather than for their constant prefactors. In \autoref{app:explicit_near_optimal_QLSS} of the appendix, we provide detailed analyses of two QLSS methods, one that uses a variant of the exhaustive search method for learning the norm, and one that additionally applies amplitude amplification. We arrive at explicit complexity expressions in \autoref{thm:QLSS_random_t} and \autoref{thm:QLSS_random_t_FPAA}. In deriving these methods, we have put a bit more thought into reducing the constant factors, although there remains room for improvement. For example, in \autoref{algo:random_t_full_QLSS} (analyzed in \autoref{thm:QLSS_random_t}), rather than exhaustively searching over a discrete set of candidate estimates $\tau = \ln(t)$ for the norm $\nrm{\vec{x}}$, we opt to choose $\tau$ (nearly) uniformly at random from the continuum, repeating this process until we find a $t$ for which \autoref{algo:main_algo} succeeds. We do not need to spend queries repeating \autoref{algo:main_algo} many times to become confident in this value of $t$, and we need not re-run KR using the value of $t$ we find; we can proceed directly to the KP step that refines the state to error $\varepsilon$ to solve the QLSP. Ultimately, we achieve complexity that scales as $O(\kappa\log(\kappa)\log\log(\kappa))$, but this is upper bounded by the easier-to-use expression reported in \autoref{tab:comparison} and Eq.~\eqref{eq:Q_exhaustive_search_practical} when $\kappa \in [3,10^6]$, which covers the vast majority of practical cases. See Eq.~\eqref{eq:Q_FPAA_practical} for the analogous expression for the amplitude-amplified method. 

In \autoref{app:explicit_optimal_QLSS}, we give a detailed analysis of an optimal QLSS (\autoref{algo:full_QLSS_constant_prefactors}). which resembles the QLSS described in \autoref{sec:estimate_by_adiabatic} and also has optimal $O(\kappa\log(1/\varepsilon))$ complexity. One difference is that with each step in the sequence, \autoref{algo:full_QLSS_constant_prefactors} increases the condition number of the linear system by a factor of about 20, rather than a factor of 2. Furthermore, rather than learning a multiplicative 2-approximation of the norm at each step, it targets an approximation ratio that begins large and gradually decreases with each step. The decreasing approximation ratio plays the role of median amplification and represents a different manifestation of the log log trick. Some aspects of the analysis are fairly tedious (perhaps reflecting the fact that this method is more complicated than the others considered), but crucially there is no need for the adiabatic theorem. Ultimately, we prove in \autoref{thm:optimal_QLSS} that the method achieves expected complexity upper bounded by
\begin{equation}
    Q \leq 56.0 \kappa + 1.05 \kappa \ln(\frac{\sqrt{1-\varepsilon^2}}{\varepsilon})  + 2.78 \ln(\kappa)^3 + 3.17\,,
\end{equation}
which holds for any $\kappa$ and $\varepsilon$. 

\autoref{tab:comparison} compares these complexity statements to the rigorous bounds provided for the quantum walk method \cite{costa2021OptimalLinearSystem} and the randomization method \cite{subasi2019QAlgSysLinEqsAdiabatic}. We see that all of our methods yield considerable savings of more than an order of magnitude. The rigorous guarantee provided by the asymptotically near-optimal method using ``Grover norm search'' is actually better than the asymptotically optimal ``adiabatic norm search'' method for $\kappa \approx 10^6$. Indeed, it is quite possible that an asymptotically suboptimal method could ultimately be the the best option practically, since factors like $\sqrt{\ln(\kappa)}\ln\ln(\kappa)$ evaluate to less than 10 at the scale of $\kappa =10^6$, and grow extremely slowly in $\kappa$. 

We believe the rigorous bounds we report could be improved with more work. However, some aspects of the analysis will always be loose; for example, we see no way around using worst case error bounds on the QSVT polynomials used to do KP and KR. We expect that the empirical performance of our algorithms will have lower constant prefactors than those we have reported. Indeed, numerical experiments conducted in Ref.~\cite{costa2023discreteConstantFactors} suggest the actual constant prefactors of both the randomization method and the quantum walk method are significantly smaller than their rigorous bounds, and that the quantum walk method is better than the randomization method in practice. Concretely, experiments on $16 \times 16$ random matrices of increasing condition number (up to $\kappa = 50$) suggest that for those instances the quantum walk method can prepare an ansatz state $\ket{\vec{x_{\rm ans}}}$ with constant overlap $\gamma \approx 0.98$ in complexity roughly $14\kappa$, yielding overall QLSS complexity roughly $15\kappa + 2.1\kappa \ln(1/\varepsilon)$.\footnote{This uses $\Delta = 0.2$ and assumes $\alpha \kappa/\Delta$ walk steps, with $\alpha = 1.37$, as reported in \cite[Section 4]{costa2023discreteConstantFactors} for non-Hermitian matrices.  We also account for the factor of 2 to Hermitianize the matrix---see caption of \autoref{tab:comparison}.}  However, the numerical simulations were on a specific random ensemble of small matrices (no larger than $16 \times 16$), and it remains untested if this performance would be matched for larger instances, and for instances appearing in actual applications. Empirical analyses are valuable, but it is difficult to replace a broad worst-case guarantee. The question of which approach yields the best practical performance on application-relevant instances remains an interesting open problem. 

\section{Conclusions and summary}\label{sec:conclusion}

Here we have developed a new framework for the QLSP, one where the Euclidean norm $\nrm{\vec{x}}$ plays a central role. We have shown that if the value of $\nrm{\vec{x}}$ is known up to a constant factor, the QLSP is simple to solve in optimal complexity; the only algorithmic tools required are well-established QSVT-based methods for reflecting about the kernel of an operator. 

Conceptually, then, an important aspect of our QLSS is that it requires the ability to first estimate $\nrm{\vec{x}}$, a single real number in the interval $[1,\kappa]$. Norm estimation replaces the cost-intensive and technically challenging task of preparing a high-overlap $n$-dimensional ansatz state, which has been the strategy in prior work. 
One reason to prefer norm estimation over ansatz-state preparation is that the norm need only be learned once, whereas ansatz states may need to be re-prepared each time the linear system needs to be solved. Another reason is that in certain cases we may have information about the norm that assists us to obtain a constant-factor approximation; for example, if the norm concentrates for a certain class of instances, the ensemble average may act as a decent guess for the norm for typical instances, or at least it can act as a starting point that enables a faster search. 

Having established the importance of norm estimation for solving the QLSP, we have initiated a systematic study of methods for accomplishing this task. We have proved a lower bound of $\Omega(\kappa \varepsilon^{-1})$ query compleixty to estimate the norm to within a factor of $1+\varepsilon$, and we have given several methods that nearly achieve this bound. For one, any constant factor approximation can be improved to a $1+\varepsilon$ approximation at cost $O(\kappa\varepsilon^{-1}\log(1/\varepsilon))$; matching the lower bound to within a $\log(1/\varepsilon)$ factor. To obtain the initial constant-factor approximation,  our simple framework offers a straightforward method that is nearly optimal---its $\kappa$ dependence is off by a doubly logarithmic factor of $O(\log\log(\kappa)\log\log\log(\kappa))$. Additionally, we have shown that this factor can be eliminated, yielding a norm estimation algorithm with optimal $O(\kappa)$ complexity, by combining our framework with ideas from the adiabatic path-following methods employed in prior work for ansatz preparation. Here, the relative simplicity of the norm estimation problem keeps the analysis manageable and avoids the need to analyze the adiabatic theorem. Our final complexity of $O(\kappa\log(1/\varepsilon)/\varepsilon)$ shaves off several logarithmic factors in $\kappa$ and $1/\varepsilon$ from the complexity of the best norm-estimation method known previously \cite{chakraborty2018BlockMatrixPowers}.  

Furthermore, estimating the norm is an important task in its own right, independent of how it is used within our QLSS. By improving the state of the art for norm estimation, we reduce the asymptotic complexity for several end-to-end applications \cite{dalzell2023quantumAlgorithmsSurvey}. For example, knowing the norm is essential when the computational task being solved is to estimate an inner product $W = \vec{w}^\dagger \vec{x}$, where $\vec{w}$ is a fixed vector and $\vec{x}$ is the solution to a linear system. This can be solved by rewriting $W$ as $\nrm{\vec{w}}\nrm{\vec{x}}\braket{\vec{w}}{\vec{x}}$: we may use a QLSS to prepare $\ket{\vec{x}}$ and then apply overlap estimation \cite{knill2007ObservableMeasurement} to estimate $\braket{\vec{w}}{\vec{x}}$, but notably we still need to know $\nrm{\vec{x}}$ to compute the output $W$. This situation is the prototypical problem solved by quantum algorithms for differential equations \cite{berry2014highOrderQuantumAlgorithmDiffEQ,montanaro2016quantum}, where $A\vec{x} =\vec{b}$ is a discretization of the differential equation and $W$ is some physical quantity, such as an electromagnetic scattering cross section \cite{clader2013preconditioned}. A similar situation arises for quantum algorithms for certain problems in machine learning, such as Gaussian process regression \cite{zhao2015QAssisstedGaussProcRegr, zhao2019TrainingGaussianProcess} and support vector machines \cite{rebentrost2014QSVM}.

Although the existence of QLSSs with asymptotically optimal $O(\kappa \log(1/\varepsilon))$ query complexity has been previously established in Ref.~\cite{costa2021OptimalLinearSystem}, our approach matches their complexity using a distinct approach with a (we believe) simpler analysis, and offers the potential for better rigorous performance guarantees. Indeed, we have given explicit complexity statements for our methods that are significantly better than those available previously, although we expect that all of these methods under consideration have substantial room for further optimization. 
Moreover, our framework is also slightly more flexible than much of the previous work on QLSSs. We do not require any assumptions on our linear system $A\vec{x}=\vec{b}$, other than that it has at least one solution. In contrast, prior work often assumes $A$ is invertible, and some prior work, such as Refs.~\cite{lin2019OptimalQEigenstateFiltering,costa2021OptimalLinearSystem}, handles non-Hermitian matrices $A$ indirectly by doubling the dimension, and examining the Hermitian matrix $\left[\begin{smallmatrix}0 & A \\
A^\dagger & 0 \end{smallmatrix}\right]$ instead. In general,  this incurs a factor of 2 in the query complexity to $U_A$ and $U_A^\dagger$. 

On the other hand, our framework is not able to handle the situation where there is no $\vec{x}$ satisfying $A\vec{x}=\vec{b}$, and one instead seeks to find the least-squares solution that minimizes $\nrm{A\vec{x}-\vec{b}}$. This task is handled by the QLSS based on variable-time amplitude amplification of Ref.~\cite{chakraborty2018BlockMatrixPowers}, and fitting it into our framework is an interesting open problem.

\subsection*{Acknowledgments}

We thank Fernando Brandao and Grant Salton for discussions, Mario Berta and Sam McArdle for extensive comments on a draft of this paper, and Lin Lin for motivating a closer look at the relationship to the Zeno method.  Special thanks to Andr\'as Gily\'en for the suggestion to generalize eigenstate filtering to kernel projection and use $Q_{\vec{b}}A$ directly, rather than its doubled Hermitianized version.  We are grateful to the AWS Center for Quantum Computing for its support. 

\balance
\printbibliography

@preamble{ "\newcommand{\lName}{0}" }

@preamble{ "\newcommand{\arxiv}[1]{arXiv:\href{https://arxiv.org/abs/#1}{\ttfamily{#1}}\removefirstdot}" }

@preamble{ "\newcommand{\arXiv}[1]{arXiv:\href{https://arxiv.org/abs/#1}{\ttfamily{#1}}\removefirstdot}" }

@preamble{ "\newcommand{\iacr}[1]{ePrint:\href{https://eprint.iacr.org/#1}{\ttfamily{#1}}\removefirstdot}" }

@preamble{ "\def\removefirstdot#1{\if.#1{}\else#1\fi}" }

@preamble{ "\providecommand{\multiletter}[1]{#1}\renewcommand{\multiletter}[1]{#1}" }

@preamble{ "\DeclareRobustCommand{\dutchPrefix}[2]{#2}" }

@preamble{ "\providecommand{\dutchPrefix}[2]{#2}\renewcommand{\dutchPrefix}[2]{#2}" }

@preamble{ "\newcommand{\skp}[3]{#2}" }

@preamble{"\newcommand{\focs       }[1]{\if\lName1\skp{  }{Proceedings of the #1 {IEEE} Symposium on Foundations of Computer Science ({FOCS})}{                            }\else{FOCS}\fi}"}

@preamble{"\newcommand{\stoc       }[1]{\if\lName1\skp{  }{Proceedings of the #1 {ACM} Symposium on the Theory of Computing ({STOC})}{                                     }\else{STOC}\fi}"}

@preamble{"\newcommand{\soda       }[1]{\if\lName1\skp{  }{Proceedings of the #1 {ACM-SIAM} Symposium on Discrete Algorithms ({SODA})}{                                    }\else{SODA}\fi}"}

@preamble{"\newcommand{\stacs      }[1]{\if\lName1\skp{  }{Proceedings of the #1 Symposium on Theoretical Aspects of Computer Science ({STACS})}{                          }\else{STACS}\fi}"}

@preamble{"\newcommand{\itcs       }[1]{\if\lName1\skp{  }{Proceedings of the #1 Innovations in Theoretical Computer Science Conference ({ITCS})}{                         }\else{ITCS}\fi}"}

@preamble{"\newcommand{\fsttcs     }[1]{\if\lName1\skp{  }{Proceedings of the #1 International Conference on Foundations of Software Technology and Theoretical Computer Science ({FSTTCS})}{ }\else{FSTTCS}\fi}"}

@preamble{"\newcommand{\mfcs       }[1]{\if\lName1\skp{  }{Proceedings of the #1 International Symposium on Mathematical Foundations of Computer Science ({MFCS})}{        }\else{MFCS}\fi}"}

@preamble{"\newcommand{\ccc        }[1]{\if\lName1\skp{  }{Proceedings of the #1 {IEEE} Conference on Computational Complexity ({CCC})}{                                   }\else{CCC}\fi}"}

@preamble{"\newcommand{\isit       }[1]{\if\lName1\skp{  }{Proceedings of the #1 {IEEE} International Symposium on Information Theory ({ISIT})}{                           }\else{ISIT}\fi}"}

@preamble{"\newcommand{\colt       }[1]{\if\lName1\skp{  }{Proceedings of the #1 Conference On Learning Theory ({COLT})}{                                                  }\else{COLT}\fi}"}

@preamble{"\newcommand{\nips       }[1]{\if\lName1\skp{  }{Advances in Neural Information Processing Systems #1 ({NIPS})}{                                                 }\else{NIPS}\fi}"}

@preamble{"\newcommand{\aistats    }[1]{\if\lName1\skp{  }{Proceedings of the #1 International Conference on Artificial Intelligence and Statistics ({AISTATS})}{          }\else{AISTATS}\fi}"}

@preamble{"\newcommand{\icml       }[1]{\if\lName1\skp{  }{Proceedings of the #1 International Conference on Machine Learning ({ICML})}{                                   }\else{ICML}\fi}"}

@preamble{"\newcommand{\icalp      }[1]{\if\lName1\skp{  }{Proceedings of the #1 International Colloquium on Automata, Languages, and Programming ({ICALP})}{              }\else{ICALP}\fi}"}

@preamble{"\newcommand{\esa        }[1]{\if\lName1\skp{  }{Proceedings of the #1 Annual European Symposium on Algorithms ({ESA})}{                                         }\else{ESA}\fi}"}

@preamble{"\newcommand{\tqc        }[1]{\if\lName1\skp{  }{Proceedings of the #1 Conference on the Theory of Quantum Computation, Communication, and Cryptography ({TQC})}{}\else{TQC}\fi}"}

@preamble{"\newcommand{\isca        }[1]{\if\lName1\skp{  }{Proceedings of the #1 International Symposium on
      Computer Architecture ({ISCA})}{}\else{ISCA}\fi}"}

@preamble{"\newcommand{\isaac      }[1]{\if\lName1\skp{  }{Proceedings of the #1 International Symposium on Algorithms and Computation ({ISAAC})}{                         }\else{ISAAC}\fi}"}

@preamble{"\newcommand{\aaai       }[1]{\if\lName1\skp{  }{Proceedings of the #1 AAAI Conference on Artificial Intelligence}{                                              }\else{AAAI}\fi}"}

@preamble{"\newcommand{\socg       }[1]{\if\lName1\skp{  }{Proceedings of the #1 Annual Symposium on Computational geometry}{                                              }\else{SoCG}\fi}"}

@preamble{"\newcommand{\sofsem     }[1]{\if\lName1\skp{  }{SOFSEM #1: Theory and Practice of Computer Science}{                                                            }\else{SOFSEM}\fi}"}

@preamble{"\newcommand{\ecc        }[1]{\if\lName1\skp{  }{#1 European Control Conference (ECC)}{                                                                          }\else{ECC}\fi}"}

@preamble{"\newcommand{\crypto     }[1]{\if\lName1\skp{  }{Advances in Cryptology -- CRYPTO #1}{                                                                           }\else{CRYPTO}\fi}"}

@preamble{"\newcommand{\asiacrypt  }[1]{\if\lName1\skp{  }{Advances in Cryptology -- ASIACRYPT #1}{                                                                        }\else{ASIACRYPT}\fi}"}

@preamble{"\newcommand{\eurocrypt  }[1]{\if\lName1\skp{  }{Advances in Cryptology -- EUROCRYPT #1}{                                                                        }\else{EUROCRYPT}\fi}"}

@preamble{"\newcommand{\pqcrypto   }[1]{\if\lName1\skp{  }{Post-Quantum Cryptography}{                                                                                     }\else{PQCrypto}\fi}"}

@preamble{"\newcommand{\scConference}[1]{\if\lName1\skp{  }{Proceedings of the International Conference for High Performance Computing, Networking, Storage and Analysis}{  }\else{SC}\fi}"}

@preamble{"\newcommand{\fccm}[1]{\if\lName1\skp{  }{#1 {IEEE} Annual International Symposium on Field-Programmable Custom Computing Machines}{  }\else{FCCM}\fi}"}

@preamble{"\newcommand{\lattice       }[1]{\if\lName1\skp{  }{Proceedings of the #1 International Symposium on Lattice Field Theory}{                                              }\else{Lattice}\fi}"}

@preamble{"\newcommand{\aft       }[1]{\if\lName1\skp{  }{Proceedings of the #1 ACM Conference on Advances in Financial Technologies}{                                              }\else{AFT}\fi}"}

@preamble{"\newcommand{\secConference       }[1]{\if\lName1\skp{  }{Proceedings of the IEEE/ACM #1 Symposium on Edge Computing}{                                              }\else{SEC}\fi}"}

@preamble{"\newcommand{\ims       }[1]{\if\lName1\skp{  }{#1 {IEEE} {MTT}-{S} International Microwave Symposium (IMS)}{                                              }\else{IMS}\fi}"}

@preamble{"\newcommand{\jacm          }{\if\lName1\skp{  }{Journal of the ACM}{                             }\else{J. ACM}\fi}"}

@preamble{"\newcommand{\acmta         }{\if\lName1\skp{  }{ACM Transactions on Algorithms}{                 }\else{{ACM} Trans. Algorithms}\fi}"}

@preamble{"\newcommand{\acmtct        }{\if\lName1\skp{  }{ACM Transactions on Computation Theory}{         }\else{ACM Trans. Comput. Theory}\fi}"}

@preamble{"\newcommand{\acmtqc        }{\if\lName1\skp{  }{ACM Transactions on Quantum Computing}{          }\else{ACM Trans. Quantum Comput.}\fi}"}

@preamble{"\newcommand{\acmjetcs        }{\if\lName1\skp{  }{   ACM Journal on Emerging Technologies in Computing Systems }{          }\else{ACM J. Emerg. Technol. Comput. Syst.}\fi}"}

@preamble{"\newcommand{\canadianjmath }{\if\lName1\skp{  }{Canadian Journal of Mathematics      }{          }\else{Can. J. Math.}\fi}"}

@preamble{"\newcommand{\jams          }{\if\lName1\skp{  }{Journal of the AMS}{                             }\else{J. AMS}\fi}"}

@preamble{"\newcommand{\bullAMS        }{\if\lName1\skp{  }{Bulletin of the American Mathematical Society}{                                }\else{Bull. AMS}\fi}"}

@preamble{"\newcommand{\pams          }{\if\lName1\skp{  }{Proceedings of the AMS}{                         }\else{Proc. AMS}\fi}"}

@preamble{"\newcommand{\linalgappl    }{\if\lName1\skp{  }{Linear Algebra and its Applications}{            }\else{Linear Algebra App.}\fi}"}

@preamble{"\newcommand{\jalgo         }{\if\lName1\skp{  }{Journal of Algorithms}{                          }\else{J. Algorithms}\fi}"}

@preamble{"\newcommand{\jcss          }{\if\lName1\skp{  }{Journal of Computer and System Sciences}{        }\else{J. Comput. Syst. Sci.}\fi}"}

@preamble{"\newcommand{\jcomputapplmath}{\if\lName1\skp{  }{Journal of Computational and Applied Mathematics}{        }\else{J. Comput. Appl. Math.}\fi}"}

@preamble{"\newcommand{\cc            }{\if\lName1\skp{  }{Computational Complexity}{                       }\else{Comput. Complex.}\fi}"}

@preamble{"\newcommand{\algor         }{\if\lName1\skp{  }{Algorithmica}{                                   }\else{Algorithmica}\fi}"}

@preamble{"\newcommand{\comb          }{\if\lName1\skp{  }{Combinatorica}{                                  }\else{Combinatorica}\fi}"}

@preamble{"\newcommand{\cacm          }{\if\lName1\skp{  }{Communications of the ACM}{                      }\else{Commun. ACM}\fi}"}

@preamble{"\newcommand{\sigart        }{\if\lName1\skp{  }{SIGART Bulletin}{                                }\else{SIGART Bull.}\fi}"}

@preamble{"\newcommand{\sigactn       }{\if\lName1\skp{  }{SIGACT News}{                                    }\else{SIGACT News}\fi}"}

@preamble{"\newcommand{\eatcsbul      }{\if\lName1\skp{  }{Bulletin of the {EATCS}}{                        }\else{Bull. {EATCS}}\fi}"}

@preamble{"\newcommand{\siamrev       }{\if\lName1\skp{  }{SIAM Review}{                                    }\else{SIAM Rev.}\fi}"}

@preamble{"\newcommand{\siamjc        }{\if\lName1\skp{  }{SIAM Journal on Computing}{                      }\else{SIAM J. Comp.}\fi}"}

@preamble{"\newcommand{\siamjo        }{\if\lName1\skp{  }{SIAM Journal on Optimization}{                   }\else{SIAM J. Opt.}\fi}"}

@preamble{"\newcommand{\siamjdm       }{\if\lName1\skp{  }{SIAM Journal on Discrete Mathematics}{           }\else{SIAM J. Disc. Math.}\fi}"}

@preamble{"\newcommand{\siamjnum      }{\if\lName1\skp{  }{SIAM Journal on Numerical Analysis}{             }\else{SIAM J. Num. Anal.}\fi}"}

@preamble{"\newcommand{\siamjmathanal }{\if\lName1\skp{  }{SIAM Journal on Mathematical Analysis}{          }\else{SIAM J. Math. Anal.}\fi}"}

@preamble{"\newcommand{\discmath      }{\if\lName1\skp{  }{Discrete Mathematics}{                           }\else{Disc. Math.}\fi}"}

@preamble{"\newcommand{\das           }{\if\lName1\skp{  }{Discrete Applied Mathematics}{                   }\else{Disc. App. Math.}\fi}"}

@preamble{"\newcommand{\annmath      }{\if\lName1\skp{  }{Annals of Mathematics}{              }\else{Ann. Math.}\fi}"}

@preamble{"\newcommand{\amatstat      }{\if\lName1\skp{  }{Annals of Mathematical Statistics}{              }\else{Ann. Math. Stat.}\fi}"}

@preamble{"\newcommand{\rms           }{\if\lName1\skp{  }{Russian Mathematical Surveys}{                   }\else{Russ. Math. Surv.}\fi}"}

@preamble{"\newcommand{\invmath       }{\if\lName1\skp{  }{Inventiones Mathematicae}{                       }\else{Inv. Math.}\fi}"}

@preamble{"\newcommand{\jnumber       }{\if\lName1\skp{  }{Journal of Number Theory}{                       }\else{J. Num. Th.}\fi}"}

@preamble{"\newcommand{\tcs           }{\if\lName1\skp{  }{Theoretical Computer Science}{                   }\else{Theor. Comput. Sci.}\fi}"}

@preamble{"\newcommand{\numeralgorithms}{\if\lName1\skp{  }{Numerical Algorithms}{                          }\else{Numer. Algorithms}\fi}"}

@preamble{"\newcommand{\toc           }{\if\lName1\skp{  }{Theory of Computing}{                            }\else{Theory Comput.}\fi}"}

@preamble{"\newcommand{\cjtcs         }{\if\lName1\skp{  }{Chicago Journal of Theoretical Computer Science}{}\else{Chicago J. Theoret. Comput. Sci.}\fi}"}

@preamble{"\newcommand{\mathprogram   }{\if\lName1\skp{  }{Mathematical Programming}{}\else{Math. Program.}\fi}"}

@preamble{"\newcommand{\mathcomput    }{\if\lName1\skp{  }{Mathematics of Computation}{}\else{Math. Comput.}\fi}"}

@preamble{"\newcommand{\jfourieranalappl   }{\if\lName1\skp{  }{Journal of Fourier Analysis and Applications}{}\else{J. Fourier Anal. Appl.}\fi}"}

@preamble{"\newcommand{\dcg           }{\if\lName1\skp{  }{Discrete \& Computational Geometry}{}\else{Discrete Comput. Geom.}\fi}"}

@preamble{"\newcommand{\randstructalg }{\if\lName1\skp{  }{Random Structures \& Algorithms}{}\else{Rand. Struct. Algorithms}\fi}"}

@preamble{"\newcommand{\intjunconvent }{\if\lName1\skp{  }{International Journal of Unconventional Computing}{}\else{Int. J. Unconv. Comput.}\fi}"}

@preamble{"\newcommand{\machlearnscitech }{\if\lName1\skp{  }{Machine Learning: Science and Technology}{}\else{Mach. Learn.: Sci. Technol.}\fi}"}

@preamble{"\newcommand{\neuralprocesslett }{\if\lName1\skp{  }{Neural Processing Letters}{}\else{Neural Process. Lett.}\fi}"}

@preamble{"\newcommand{\neuralcomput }{\if\lName1\skp{  }{Neural Computation}{}\else{Neural Comput.}\fi}"}

@preamble{"\newcommand{\frontartifintell}{\if\lName1\skp{  }{Frontiers in Artificial Intelligence}{}\else{Front. Artif. Intell.}\fi}"}

@preamble{"\newcommand{\jgloboptim}{\if\lName1\skp{  }{Journal of Global Optimization}{}\else{J. Glob. Optim.}\fi}"}

@preamble{"\newcommand{\epjdatasci}{\if\lName1\skp{  }{EPJ Data Science}{}\else{EPJ Data Sci.}\fi}"}

@preamble{"\newcommand{\epjquantumtechnol}{\if\lName1\skp{  }{EPJ Quantum Technology}{}\else{EPJ Quantum Technol.}\fi}"}

@preamble{"\newcommand{\applsci       }{\if\lName1\skp{  }{Applied Sciences}{    }\else{Appl. Sci.}\fi}"}

@preamble{"\newcommand{\applphysrev   }{\if\lName1\skp{  }{Applied Physics Reviews}{    }\else{Appl. Phys. Rev.}\fi}"}

@preamble{"\newcommand{\advquantumtechnol          }{\if\lName1\skp{  }{Advanced Quantum Technologies}{    }\else{Adv. Quantum Technol.}\fi}"}

@preamble{"\newcommand{\annphys       }{\if\lName1\skp{  }{Annals of Physics}{    }\else{Ann. Phys.}\fi}"}

@preamble{"\newcommand{\annualrevCMP          }{\if\lName1\skp{  }{Annual Review of Condensed Matter Physics}{    }\else{Annu. Rev. Condens. Matter Phys.}\fi}"}

@preamble{"\newcommand{\annualrevNPS          }{\if\lName1\skp{  }{Annual Review of Nuclear and Particle Science}{    }\else{Annu. Rev. Nucl. Part. Sci.}\fi}"}

@preamble{"\newcommand{\quantum       }{\if\lName1\skp{  }{Quantum}{                                          }\else{Quantum}\fi}"}

@preamble{"\newcommand{\quantumscitechnol       }{\if\lName1\skp{  }{Quantum Science and Technology}{                                          }\else{Quantum Sci. Technol.}\fi}"}

@preamble{"\newcommand{\cmp           }{\if\lName1\skp{  }{Communications in Mathematical Physics}{             }\else{Commun. Math. Phys.}\fi}"}

@preamble{"\newcommand{\frontphys     }{\if\lName1\skp{  }{Frontiers in Physics}{                               }\else{Front. Phys.}\fi}"}

@preamble{"\newcommand{\jmp           }{\if\lName1\skp{  }{Journal of Mathematical Physics}{                    }\else{J. Math. Phys.}\fi}"}

@preamble{"\newcommand{\rspa          }{\if\lName1\skp{  }{Proceedings of the Royal Society A}{                 }\else{Proc. R. Soc. A}\fi}"}

@preamble{"\newcommand{\philostransroyal}{\if\lName1\skp{  }{Philosophical Transactions of the Royal Society A}{                 }\else{Philos. Trans. R. Soc. A}\fi}"}

@preamble{"\newcommand{\qic           }{\if\lName1\skp{  }{Quantum Information and Computation}{                }\else{Quantum Inf. Comput.}\fi}"}

@preamble{"\newcommand{\qip           }{\if\lName1\skp{  }{Quantum Information Processing}{                }\else{Quantum Inf. Process.}\fi}"}

@preamble{"\newcommand{\physrev       }{\if\lName1\skp{  }{Physical Review}{                                    }\else{Phys. Rev.}\fi}"}

@preamble{"\newcommand{\pra           }{\if\lName1\skp{  }{Physical Review A}{                                  }\else{Phys. Rev. A}\fi}"}

@preamble{"\newcommand{\prb           }{\if\lName1\skp{  }{Physical Review B}{                                  }\else{Phys. Rev. B}\fi}"}

@preamble{"\newcommand{\prc           }{\if\lName1\skp{  }{Physical Review C}{                                  }\else{Phys. Rev. C}\fi}"}

@preamble{"\newcommand{\prd           }{\if\lName1\skp{  }{Physical Review D}{                                  }\else{Phys. Rev. D}\fi}"}

@preamble{"\newcommand{\pre           }{\if\lName1\skp{  }{Physical Review E}{                                  }\else{Phys. Rev. E}\fi}"}

@preamble{"\newcommand{\prr           }{\if\lName1\skp{  }{Physical Review Research}{                           }\else{Phys. Rev. Res.}\fi}"}

@preamble{"\newcommand{\prx           }{\if\lName1\skp{  }{Physical Review X}{                                  }\else{Phys. Rev. X}\fi}"}

@preamble{"\newcommand{\prxq          }{\if\lName1\skp{  }{PRX Quantum}{                                        }\else{PRX Quantum}\fi}" }

@preamble{"\newcommand{\prl           }{\if\lName1\skp{  }{Physical Review Letters}{                            }\else{Phys. Rev. Lett.}\fi}"}

@preamble{"\newcommand{\europhyslett  }{\if\lName1\skp{  }{Europhysics Letters}{                                }\else{Europhys. Lett.}\fi}"}

@preamble{"\newcommand{\epja          }{\if\lName1\skp{  }{The European Physical Journal A}{                    }\else{Euro. Phys. J. A}\fi}"}

@preamble{"\newcommand{\epjd          }{\if\lName1\skp{  }{The European Physical Journal D}{                    }\else{Euro. Phys. J. D}\fi}"}

@preamble{"\newcommand{\njp           }{\if\lName1\skp{  }{New Journal of Physics}{                             }\else{New J. Phys.}\fi}"}

@preamble{"\newcommand{\prapp         }{\if\lName1\skp{  }{Physical Review Applied}{                            }\else{Phys. Rev. Appl.}\fi}"}

@preamble{"\newcommand{\physrep       }{\if\lName1\skp{  }{Physics Reports}{                                    }\else{Phys. Rep.}\fi}"}

@preamble{"\newcommand{\rmp           }{\if\lName1\skp{  }{Reviews of Modern Physics}{                          }\else{Rev. Mod. Phys.}\fi}"}

@preamble{"\newcommand{\repprogphys   }{\if\lName1\skp{  }{Reports on Progress in Physics}{                     }\else{Rep. Prog. Phys.}\fi}"}

@preamble{"\newcommand{\physplasmas   }{\if\lName1\skp{  }{Physics of Plasmas}{                                 }\else{Phys. Plasmas}\fi}"}

@preamble{"\newcommand{\phystoday     }{\if\lName1\skp{  }{Physics Today}{                                      }\else{Phys. Today}\fi}"}

@preamble{"\newcommand{\physics       }{\if\lName1\skp{  }{Physics}{                                            }\else{Phys.}\fi}"}

@preamble{"\newcommand{\nature        }{\if\lName1\skp{  }{Nature}{                                             }\else{Nature}\fi}"}

@preamble{"\newcommand{\natcomm       }{\if\lName1\skp{  }{Nature Communications}{                              }\else{Nat. Commun.}\fi}"}

@preamble{"\newcommand{\natcomputsci  }{\if\lName1\skp{  }{Nature Computational Science}{                       }\else{Nat. Comput. Sci.}\fi}"}

@preamble{"\newcommand{\natphys       }{\if\lName1\skp{  }{Nature Physics}{                                     }\else{Nat. Phys.}\fi}"}

@preamble{"\newcommand{\natphotonics  }{\if\lName1\skp{  }{Nature Photonics}{                                     }\else{Nat. Photonics}\fi}"}

@preamble{"\newcommand{\natrevphys    }{\if\lName1\skp{  }{Nature Reviews Physics}{                                     }\else{Nat. Rev. Phys.}\fi}"}

@preamble{"\newcommand{\natrevmater   }{\if\lName1\skp{  }{Nature Reviews Materials}{                                     }\else{Nat. Rev. Mater.}\fi}"}

@preamble{"\newcommand{\natrevmethodsprimers}{\if\lName1\skp{}{Nature Reviews Methods Primers}{                               }\else{Nat. Rev. Methods Primers}\fi}"}

@preamble{"\newcommand{\npjqi         }{\if\lName1\skp{  }{npj Quantum Information}{                            }\else{npj Quant. Inf.}\fi}"}

@preamble{"\newcommand{\scirep        }{\if\lName1\skp{  }{Scientific Reports}{                                 }\else{Sci. Rep.}\fi}"}

@preamble{"\newcommand{\science       }{\if\lName1\skp{  }{Science}{                                            }\else{Science}\fi}"}

@preamble{"\newcommand{\sciadv      }{\if\lName1\skp{  }{Science Advances}{                                   }\else{Sci. Adv.}\fi}"}

@preamble{"\newcommand{\scibull      }{\if\lName1\skp{  }{Science Bulletin}{                                   }\else{Sci. Bull.}\fi}"}

@preamble{"\newcommand{\jchemphys           }{\if\lName1\skp{  }{The Journal of Chemical Physics}{ }\else{J. Chem. Phys.}\fi}"}

@preamble{"\newcommand{\jphyschemlett      }{\if\lName1\skp{  }{The Journal of Physical Chemistry Letters}{ }\else{J. Phys. Chem. Lett.}\fi}"}

@preamble{"\newcommand{\jpa           }{\if\lName1\skp{  }{Journal of Physics A: Mathematical and Theoretical}{ }\else{J. Phys. A}\fi}"}

@preamble{"\newcommand{\jpg           }{\if\lName1\skp{  }{Journal of Physics G: Nuclear and Particle Physics}{ }\else{J. Phys. G}\fi}"}

@preamble{"\newcommand{\ijtp          }{\if\lName1\skp{  }{International Journal of Theoretical Physics}{       }\else{Int. J. Th. Phys.}\fi}"}

@preamble{"\newcommand{\jmo           }{\if\lName1\skp{  }{Journal of Modern Optics}{                           }\else{J. Mod. Opt.}\fi}"}

@preamble{"\newcommand{\jhep           }{\if\lName1\skp{  }{Journal of High Energy Physics}{                           }\else{J. High Energy Phys.}\fi}"}

@preamble{"\newcommand{\jstatph       }{\if\lName1\skp{  }{Journal of Statistical Physics}{                     }\else{J. Stat. Phys.}\fi}"}

@preamble{"\newcommand{\jcompphys       }{\if\lName1\skp{  }{Journal of Computational Physics}{                 }\else{J. Comput. Phys.}\fi}"}

@preamble{"\newcommand{\computphyscommun}{\if\lName1\skp{  }{Computer Physics Communications}{               }\else{Comput. Phys. Commun.}\fi}"}

@preamble{"\newcommand{\jstatmech     }{\if\lName1\skp{  }{Journal of Statistical Mechanics: Theory and Experiment}{ }\else{J. Stat. Mech. Theory Exp.}\fi}"}

@preamble{"\newcommand{\pnas          }{\if\lName1\skp{  }{Proceedings of the National Academy of Sciences}{    }\else{Proc. Natl. Acad. Sci.}\fi}"}

@preamble{"\newcommand{\avsquantsci          }{\if\lName1\skp{  }{AVS Quantum Science}{    }\else{AVS Quantum Sci.}\fi}"}

@preamble{"\newcommand{\quantummachintell    }{\if\lName1\skp{  }{Quantum Machine Intelligence}{    }\else{Quantum Mach. Intell.}\fi}"}

@preamble{"\newcommand{\natmachintell    }{\if\lName1\skp{  }{Nature Machine Intelligence}{    }\else{Nat. Mach. Intell.}\fi}"}

@preamble{"\newcommand{\scichinainfsci   }{\if\lName1\skp{  }{
Science China Information Sciences}{    }\else{Sci. China Inf. Sci.}\fi}"}

@preamble{"\newcommand{\jmagnreson   }{\if\lName1\skp{  }{Journal of Magnetic Resonance}{    }\else{J. Magn. Reson.}\fi}"}

@preamble{"\newcommand{\lncs          }{\if\lName1\skp{  }{Lecture Notes in Computer Science}{                  }\else{L. Notes Comp. Sci.}\fi}"}

@preamble{"\newcommand{\lnai          }{\if\lName1\skp{  }{Lecture Notes in Artificial Intelligence}{           }\else{L. Notes Art. Int.}\fi}"}

@preamble{"\newcommand{\lnm           }{\if\lName1\skp{  }{Lecture Notes in Mathematics}{                       }\else{L. Notes Math.}\fi}"}

@preamble{"\newcommand{\tams          }{\if\lName1\skp{  }{Transactions of the American Mathematical Society}{  }\else{Trans. AMS}\fi}"}

@preamble{"\newcommand{\ieeetit       }{\if\lName1\skp{  }{{IEEE} Transactions on Information Theory}{          }\else{{IEEE} Trans. Inf. Theory}\fi}"}

@preamble{"\newcommand{\ieeetnnls       }{\if\lName1\skp{  }{{IEEE} Transactions on Neural Networks and Learning Systems}{          }\else{{IEEE} Trans. Neural Netw. Learn. Syst.}\fi}"}

@preamble{"\newcommand{\ieeetcad       }{\if\lName1\skp{  }{{IEEE} Transactions on Computer-Aided Design of Integrated Circuits and Systems}{          }\else{{IEEE} Trans. Comput.-Aided Des. Integr. Circuits Syst.}\fi}"}

@preamble{"\newcommand{\ieeetqe       }{\if\lName1\skp{  }{{IEEE} Transactions on Quantum Engineering}{          }\else{{IEEE} Trans. Quantum Eng.}\fi}"}

@preamble{"\newcommand{\ieeejsac       }{\if\lName1\skp{  }{{IEEE} Journal on Selected Areas in Communications}{          }\else{{IEEE} J. Sel. Areas Commun.}\fi}"}

@preamble{"\newcommand{\ieeebits       }{\if\lName1\skp{  }{{IEEE} {BITS} the Information Theory Magazine}{          }\else{{IEEE} {BITS} Inf. Theory Mag.}\fi}"}

@preamble{"\newcommand{\quantumeng       }{\if\lName1\skp{  }{Quantum Engineering}{          }\else{Quantum Eng.}\fi}"}

@preamble{"\newcommand{\iscs          }{\if\lName1\skp{  }{International Series in Computer Science}{           }\else{Int. Ser. Comp. Sci.}\fi}"}

@preamble{"\newcommand{\tocl          }{\if\lName1\skp{  }{Theory of Computing Library}{                        }\else{Th. Comp. Lib.}\fi}"}

@preamble{"\newcommand{\actanumer     }{\if\lName1\skp{  }{Acta Numerica}{                        }\else{Acta Numer.}\fi}"}

@preamble{"\newcommand{\jderiv     }{\if\lName1\skp{  }{The Journal of Derivatives}{                        }\else{J. Deriv.}\fi}"}

@preamble{"\newcommand{\chemrev     }{\if\lName1\skp{  }{Chemical Reviews}{                        }\else{Chem. Rev.}\fi}"}

@preamble{"\newcommand{\WIREsCMS     }{\if\lName1\skp{  }{WIREs Molecular Computational Science}{                        }\else{WIREs Comput. Mol. Sci.}\fi}"}

@preamble{"\newcommand{\accchemres     }{\if\lName1\skp{  }{Accounts of Chemical Research}{                        }\else{Acc. Chem. Res.}\fi}"}

@preamble{"\newcommand{\chemphyslett     }{\if\lName1\skp{  }{Chemical Physics Letters}{                        }\else{Chem. Phys. Lett.}\fi}"}

@preamble{"\newcommand{\molphys     }{\if\lName1\skp{  }{Molecular Physics}{                        }\else{Mol. Phys.}\fi}"}

@preamble{"\newcommand{\commchem     }{\if\lName1\skp{  }{Communications Chemistry}{                        }\else{Commun. Chem.}\fi}"}

@preamble{"\newcommand{\acscentsci     }{\if\lName1\skp{  }{ACS Central Science}{                        }\else{ACS Cent. Sci.}\fi}"}

@preamble{"\newcommand{\jchemtheorycomput}{\if\lName1\skp{  }{Journal of Chemical Theory and Computation}{                        }\else{J. Chem. Theory Comput.}\fi}"}

@preamble{"\newcommand{\chemsci        }{\if\lName1\skp{  }{Chemical Science}{                        }\else{Chem. Sci.}\fi}"}

@preamble{"\newcommand{\intjqchem        }{\if\lName1\skp{  }{International journal of quantum chemistry}{                        }\else{Int. J. Quantum Chem.}\fi}"}

@preamble{"\newcommand{\jfqa        }{\if\lName1\skp{  }{Journal of Financial and Quantitative Analysis}{                        }\else{J. Financial Quant. Anal.}\fi}"}

@Article{	  costa2021OptimalLinearSystem,
  title		= {Optimal Scaling Quantum Linear-Systems Solver via Discrete
		  Adiabatic Theorem},
  author	= {Costa, Pedro C.S. and An, Dong and Sanders, Yuval R. and
		  Su, Yuan and Babbush, Ryan and Berry, Dominic W.},
  journal	= {\prxq},
  volume	= {3},
  issue		= {4},
  pages		= {040303},
  numpages	= {54},
  year		= {2022},
  month		= {10},
  publisher	= {American Physical Society},
  doi		= {10.1103/PRXQuantum.3.040303},
  url		= {https://link.aps.org/doi/10.1103/PRXQuantum.3.040303},
  note		= {\arXiv{2111.08152}}
}

@InProceedings{	  chakraborty2018BlockMatrixPowers,
  author	= {Shantanav Chakraborty and Andr{\'a}s Gily{\'e}n and Stacey
		  Jeffery},
  booktitle	= {\icalp{46th}},
  doi		= {10.4230/LIPIcs.ICALP.2019.33},
  note		= {\arxiv{1804.01973}},
  pages		= {33:1--33:14},
  sseries	= {Leibniz International Proceedings in Informatics
		  (LIPIcs)},
  title		= {The power of block-encoded matrix powers: {I}mproved
		  regression techniques via faster {H}amiltonian simulation},
  vvolume	= {132},
  year		= {2019}
}

@Article{	  harrow2009QLinSysSolver,
  author	= {Harrow, Aram W. and Hassidim, Avinatan and Lloyd, Seth},
  doi		= {10.1103/PhysRevLett.103.150502},
  journal	= {\prl},
  note		= {\arxiv{0811.3171}},
  number	= {15},
  pages		= {150502},
  title		= {Quantum algorithm for linear systems of equations},
  volume	= {103},
  year		= {2009}
}

@article{jennings2023QLSS,
  title = {Randomized Adiabatic Quantum Linear Solver Algorithm with Optimal Complexity Scaling and Detailed Running Costs},
  author = {Jennings, David and Lostaglio, Matteo and Pallister, Sam and Sornborger, Andrew T. and Suba\c{s}\i, Yi\u{g}it},
  journal = {\prxq},
  volume = {6},
  issue = {4},
  pages = {040373},
  numpages = {23},
  year = {2025},
  month = {12},
  publisher = {American Physical Society},
  doi = {10.1103/1xkb-22cc},
  url = {https://link.aps.org/doi/10.1103/1xkb-22cc},
  note		= {\arxiv{2305.11352}}
}

@Article{	  lin2019OptimalQEigenstateFiltering,
  author	= {Lin, Lin and Tong, Yu},
  doi		= {10.22331/q-2020-11-11-361},
  journal	= {\quantum},
  note		= {\arxiv{1910.14596}},
  pages		= {361},
  title		= {Optimal polynomial based quantum eigenstate filtering with
		  application to solving quantum linear systems},
  volume	= {4},
  year		= {2020}
}

@Article{	  yoder2014FixedPointSearch,
  author	= {Yoder, Theodore J. and Low, Guang Hao and Chuang, Isaac
		  L.},
  doi		= {10.1103/PhysRevLett.113.210501},
  journal	= {\prl},
  note		= {\arxiv{1409.3305}},
  number	= {21},
  numpages	= {5},
  pages		= {210501},
  title		= {Fixed-Point Quantum Search with an Optimal Number of
		  Queries},
  volume	= {113},
  year		= {2014}
}

@book{dalzell2023quantumAlgorithmsSurvey,
      title={Quantum Algorithms: A Survey of Applications and End-to-end Complexities}, 
      author={Dalzell,  Alexander M. and Sam McArdle and Mario Berta and Przemyslaw Bienias and Chi-Fang Chen and Andr\'{a}s Gily\'{e}n and Connor T. Hann and Michael J. Kastoryano and Emil T. Khabiboulline and Aleksander Kubica and Grant Salton and Samson Wang and Fernando G. S. L. Brandão},
        doi = {10.1017/9781009639651},
        publisher = {Cambridge University Press},
      year={2025},
    month={4},
      archivePrefix = {arXiv},
      eprint = {2310.03011},
      note = {\arxiv{2310.03011}}
}

@Article{	  albash2018AQCreview,
  author	= {Albash, Tameem and Lidar, Daniel A.},
  doi		= {10.1103/RevModPhys.90.015002},
  issue		= {1},
  journal	= {\rmp},
  month		= {1},
  note		= {\arxiv{1611.04471}},
  numpages	= {64},
  pages		= {015002},
  publisher	= {American Physical Society},
  title		= {Adiabatic quantum computation},
  volume	= {90},
  year		= {2018}
}

@InProceedings{	  ambainis2010VTAA,
  author	= {Andris Ambainis},
  booktitle	= {\stacs{29th}},
  doi		= {10.4230/LIPIcs.STACS.2012.636},
  note		= {\arxiv{1010.4458}},
  pages		= {636--647},
  title		= {Variable time amplitude amplification and quantum
		  algorithms for linear algebra problems},
  year		= {2012}
}

@Article{	  childs2015QLinSysExpPrec,
  author	= {Andrew M. Childs and Robin Kothari and Rolando D. Somma},
  doi		= {10.1137/16M1087072},
  journal	= {\siamjc},
  note		= {\arxiv{1511.02306}},
  number	= {6},
  pages		= {1920--1950},
  title		= {Quantum Algorithm for Systems of Linear Equations with
		  Exponentially Improved Dependence on Precision},
  volume	= {46},
  year		= {2017}
}

@Article{	  subasi2019QAlgSysLinEqsAdiabatic,
  author	= {Suba\c{s}\i, Yi\u{g}it and Somma, Rolando D. and Orsucci,
		  Davide},
  doi		= {10.1103/PhysRevLett.122.060504},
  journal	= {\prl},
  note		= {\arxiv{1805.10549}},
  number	= {6},
  numpages	= {5},
  pages		= {060504},
  title		= {Quantum Algorithms for Systems of Linear Equations
		  Inspired by Adiabatic Quantum Computing},
  volume	= {122},
  year		= {2019}
}

@Article{	  an2022QLSStimeDepAdiabatic,
  address	= {New York, NY, USA},
  author	= {An, Dong and Lin, Lin},
  doi		= {10.1145/3498331},
  journal	= {\acmtqc},
  number	= {2},
  publisher	= {Association for Computing Machinery},
  title		= {Quantum Linear System Solver Based on Time-Optimal
		  Adiabatic Quantum Computing and Quantum Approximate
		  Optimization Algorithm},
  volume	= {3},
  year		= {2022},
  note		= {\arxiv{1909.05500}}
}

@Article{	  dalzell2022socp,
  title = {End-To-End Resource Analysis for Quantum Interior-Point Methods and Portfolio Optimization},
  author = {Dalzell, Alexander M. and Clader, B. David and Salton, Grant and Berta, Mario and Lin, Cedric Yen-Yu and Bader, David A. and Stamatopoulos, Nikitas and Schuetz, Martin J. A. and Brand\~ao, Fernando G. S. L. and Katzgraber, Helmut G. and Zeng, William J.},
  journal = {\prxq},
  volume = {4},
  issue = {4},
  pages = {040325},
  numpages = {40},
  year = {2023},
  month = {11},
  publisher = {American Physical Society},
  doi = {10.1103/PRXQuantum.4.040325},
  url = {https://link.aps.org/doi/10.1103/PRXQuantum.4.040325},
  note		= {\arXiv{2211.12489}},
}

@article{jennings2023costDiffEQ,
  doi = {10.22331/q-2024-12-10-1553},
  url = {https://doi.org/10.22331/q-2024-12-10-1553},
  title = {The cost of solving linear differential equations on a quantum computer: fast-forwarding to explicit resource counts},
  author = {Jennings, David and Lostaglio, Matteo and Lowrie, Robert B. and Pallister, Sam and Sornborger, Andrew T.},
  journal = {\quantum},
  issn = {2521-327X},
  publisher = {{Verein zur F{\"{o}}rderung des Open Access Publizierens in den Quantenwissenschaften}},
  volume = {8},
  pages = {1553},
  month = {12},
  year = {2024},
  note={\arxiv{2309.07881}}
}

@Book{		  nielsen2002QCQI,
  author	= {Nielsen, Michael A. and Chuang, Isaac L.},
  doi		= {10.1017/CBO9780511976667},
  publisher	= {Cambridge University Press},
  title		= {Quantum computation and quantum information},
  year		= {2000}
}

@Article{	  nagaj2009FastAmpQMA,
  author	= {Nagaj, Daniel and Wocjan, Pawel and Zhang, Yong},
  doi		= {10.26421/QIC9.11-12},
  journal	= {\qic},
  note		= {\arxiv{0904.1549}},
  number	= {11\&12},
  numpages	= {16},
  pages		= {1053--1068},
  title		= {Fast Amplification of {QMA}},
  volume	= {9},
  year		= {2009}
}

@Unpublished{gretta2023sharpNoisyBinarySearch,
      title={Sharp Noisy Binary Search with Monotonic Probabilities}, 
      author={Lucas Gretta and Eric Price},
      year={2023},
      note={\arxiv{2311.00840}}
}

@InCollection{	  brassard2002AmpAndEst,
  author	= {Gilles Brassard and Peter H{\o}yer and Michele Mosca and
		  Alain Tapp},
  booktitle	= {Quantum Computation and Quantum Information: A Millennium
		  Volume},
  doi		= {10.1090/conm/305/05215},
  note		= {\arxiv{quant-ph/0005055}},
  pages		= {53--74},
  publisher	= {AMS},
  series	= {Contemporary Mathematics Series},
  title		= {Quantum Amplitude Amplification and Estimation},
  volume	= {305},
  year		= {2002}
}

@InProceedings{	  grover1996QSearch,
  author	= {Lov K. Grover},
  booktitle	= {\stoc{28th}},
  doi		= {10.1145/237814.237866},
  note		= {\arxiv{quant-ph/9605043}},
  pages		= {212--219},
  title		= {A Fast Quantum Mechanical Algorithm for Database Search},
  year		= 1996
}

@article{burnashev1974interval,
  title={An interval estimation problem for controlled observations},
  author={Burnashev, Marat Valievich and Zigangirov, Kamil'Shamil'evich},
  journal={Problemy Peredachi Informatsii},
  volume={10},
  number={3},
  pages={51--61},
  year={1974},
  publisher={Russian Academy of Sciences, Branch of Informatics, Computer Equipment and~…},
  note = {In Russian, see \cite{wang2022noisySortingcapacity} for key proof in English}
}

@INPROCEEDINGS{wang2022noisySortingcapacity,
  author={Wang, Ziao and Ghaddar, Nadim and Wang, Lele},
  booktitle={\isit{2022}}, 
  title={Noisy Sorting Capacity}, 
  year={2022},
  volume={},
  number={},
  pages={2541-2546},
  doi={10.1109/ISIT50566.2022.9834370},
  note = {\arxiv{2202.01446}}
}

@INPROCEEDINGS{ben-or2008BayesianLearner,

  author={Ben-Or, Michael and Hassidim, Avinatan},

  booktitle={\focs{49th}}, 

  title={The {B}ayesian Learner is Optimal for Noisy Binary Search  (and Pretty Good for Quantum as Well)}, 

  year={2008},

  volume={},

  number={},

  pages={221-230},

  doi={10.1109/FOCS.2008.58},
  note = {\arxiv{quant-ph/0703231}}
}

@InProceedings{	  gilyen2018QSingValTransf,
  author	= {Andr{\'a}s Gily{\'e}n and Yuan Su and Guang Hao Low and
		  Nathan Wiebe},
  booktitle	= {\stoc{51st}},
  doi		= {10.1145/3313276.3316366},
  note		= {\arxiv{1806.01838}},
  numpages	= {12},
  pages		= {193--204},
  title		= {Quantum singular value transformation and beyond:
		  {E}xponential improvements for quantum matrix arithmetics},
  year		= {2019}
}

@article{Orsucci2021solvingclassesof,
  doi = {10.22331/q-2021-11-08-573},
  url = {https://doi.org/10.22331/q-2021-11-08-573},
  title = {On solving classes of positive-definite quantum linear systems with quadratically improved runtime in the condition number},
  author = {Orsucci, Davide and Dunjko, Vedran},
  journal = {\quantum},
  issn = {2521-327X},
  publisher = {{Verein zur F{\"{o}}rderung des Open Access Publizierens in den Quantenwissenschaften}},
  volume = {5},
  pages = {573},
  month = {11},
  year = {2021},
  note = {\arxiv{2101.11868}}
}

@Article{	  knill2007ObservableMeasurement,
  title		= {Optimal quantum measurements of expectation values of
		  observables},
  author	= {Knill, Emanuel and Ortiz, Gerardo and Somma, Rolando D.},
  journal	= {\pra},
  volume	= {75},
  issue		= {1},
  pages		= {012328},
  numpages	= {13},
  year		= {2007},
  month		= {1},
  publisher	= {American Physical Society},
  doi		= {10.1103/PhysRevA.75.012328},
  url		= {https://link.aps.org/doi/10.1103/PhysRevA.75.012328},
  note		= {\arxiv{quant-ph/0607019}}
}

@Article{	  montanaro2016quantum,
  title		= {Quantum algorithms and the finite element method},
  author	= {Montanaro, Ashley and Pallister, Sam},
  journal	= {\pra},
  volume	= {93},
  number	= {3},
  pages		= {032324},
  year		= {2016},
  publisher	= {APS},
  doi		= {10.1103/PhysRevA.93.032324},
  note		= {\arXiv{1512.05903}}
}

@Article{	  zhao2015QAssisstedGaussProcRegr,
  author	= {Zhao, Zhikuan and Fitzsimons, Jack K. and Fitzsimons,
		  Joseph F.},
  doi		= {10.1103/PhysRevA.99.052331},
  journal	= {\pra},
  note		= {\arxiv{1512.03929}},
  number	= {5},
  numpages	= {6},
  pages		= {052331},
  title		= {Quantum-assisted {G}aussian process regression},
  volume	= {99},
  year		= {2019}
}

@Article{	  rebentrost2014QSVM,
  author	= {Rebentrost, Patrick and Mohseni, Masoud and Lloyd, Seth},
  doi		= {10.1103/PhysRevLett.113.130503},
  journal	= {\prl},
  note		= {\arxiv{1307.0471}},
  number	= {13},
  pages		= {130503},
  title		= {Quantum support vector machine for big data
		  classification},
  volume	= {113},
  year		= {2014}
}

@Article{	  zhao2019TrainingGaussianProcess,
  title		= {Quantum algorithms for training Gaussian processes},
  author	= {Zhao, Zhikuan and Fitzsimons, Jack K. and Osborne, Michael
		  A. and Roberts, Stephen J. and Fitzsimons, Joseph F.},
  journal	= {\pra},
  volume	= {100},
  issue		= {1},
  pages		= {012304},
  numpages	= {5},
  year		= {2019},
  month		= {7},
  publisher	= {American Physical Society},
  doi		= {10.1103/PhysRevA.100.012304},
  url		= {https://link.aps.org/doi/10.1103/PhysRevA.100.012304},
  note		= {\arxiv{1803.10520}}
}

@Article{	  clader2013preconditioned,
  author	= {Clader, B David and Jacobs, Bryan C and Sprouse, Chad R},
  journal	= {\prl},
  number	= {25},
  pages		= {250504},
  publisher	= {APS},
  title		= {Preconditioned quantum linear system algorithm},
  volume	= {110},
  doi		= {10.1103/PhysRevLett.110.250504},
  year		= {2013},
  note		= {\arXiv{1301.2340}}
}

@Article{	  martyn2021GrandUnificationQAlgs,
  author	= {Martyn, John M. and Rossi, Zane M. and Tan, Andrew K. and
		  Chuang, Isaac L.},
  doi		= {10.1103/PRXQuantum.2.040203},
  journal	= {\prx},
  note		= {\arXiv{2105.02859}},
  number	= {4},
  numpages	= {40},
  pages		= {040203},
  title		= {Grand Unification of Quantum Algorithms},
  volume	= {2},
  year		= {2021}
}

@article{costa2023discreteConstantFactors,
  doi = {10.22331/q-2025-10-20-1887},
  url = {https://doi.org/10.22331/q-2025-10-20-1887},
  title = {The discrete adiabatic quantum linear system solver has lower constant factors than the randomized adiabatic solver},
  author = {Costa, Pedro C.S. and An, Dong and Babbush, Ryan and Berry, Dominic},
  journal = {\quantum},
  issn = {2521-327X},
  publisher = {{Verein zur F{\"{o}}rderung des Open Access Publizierens in den Quantenwissenschaften}},
  volume = {9},
  pages = {1887},
  month = {10},
  year = {2025},
  note={\arxiv{2312.07690}}
}

@unpublished{mori2026sparsityLowerBound,
      title={Sparsity-dependent Complexity Lower Bound of Quantum Linear System Solvers}, 
      author={Hitomi Mori and Yuta Kikuchi and Marcello Benedetti and Matthias Rosenkranz},
      year={2026},
      note={\arxiv{2601.16697}}
}

@Article{	  babbush2021FocusBeyondQuadratic,
  title		= {Focus beyond Quadratic Speedups for Error-Corrected
		  Quantum Advantage},
  author	= {Babbush, Ryan and McClean, Jarrod R. and Newman, Michael
		  and Gidney, Craig and Boixo, Sergio and Neven, Hartmut},
  journal	= {\prxq},
  volume	= {2},
  number	= {1},
  pages		= {010103},
  numpages	= {11},
  year		= {2021},
  month		= {3},
  publisher	= {American Physical Society},
  doi		= {10.1103/PRXQuantum.2.010103},
  note		= {\arxiv{2011.04149}}
}

@inbook{kothari2023meanEstimationSourceCode,
	author = {Robin Kothari and Ryan O'Donnell},
	booktitle = {\soda{2023}},
	doi = {10.1137/1.9781611977554.ch44},
	eprint = {https://epubs.siam.org/doi/pdf/10.1137/1.9781611977554.ch44},
	pages = {1186-1215},
    year = {2023},
	title = {Mean estimation when you have the source code; or, quantum Monte Carlo methods},
	url = {https://epubs.siam.org/doi/abs/10.1137/1.9781611977554.ch44},
    note = {\arxiv{2208.07544}}}

@inproceedings{nayak1999queryComplexityMedian,
author = {Nayak, Ashwin and Wu, Felix},
title = {The quantum query complexity of approximating the median and related statistics},
year = {1999},
isbn = {1581130678},
publisher = {Association for Computing Machinery},
address = {New York, NY, USA},
url = {https://doi.org/10.1145/301250.301349},
doi = {10.1145/301250.301349},
booktitle = {\stoc{31st}},
pages = {384–393},
numpages = {10},
location = {Atlanta, Georgia, USA},
series = {STOC '99},
note={\arxiv{quant-ph/9804066}}
}

@article{berry2014highOrderQuantumAlgorithmDiffEQ,
doi = {10.1088/1751-8113/47/10/105301},
url = {https://dx.doi.org/10.1088/1751-8113/47/10/105301},
year = {2014},
month = {2},
publisher = {IOP Publishing},
volume = {47},
number = {10},
pages = {105301},
author = {Dominic W Berry},
title = {High-order quantum algorithm for solving linear differential equations},
journal = {\jpa},
note = {\arxiv{1010.2745}}
}

@InProceedings{cunningham2024eigenpathTraversal,
  author =	{Cunningham, Joseph and Roland, J\'{e}r\'{e}mie},
  title =	{Eigenpath Traversal by {P}oisson-Distributed Phase Randomisation},
  booktitle =	{\tqc{19th}},
  pages =	{7:1--7:20},
  series =	{Leibniz International Proceedings in Informatics (LIPIcs)},
  ISBN =	{978-3-95977-328-7},
  ISSN =	{1868-8969},
  year =	{2024},
  volume =	{310},
  editor =	{Magniez, Fr\'{e}d\'{e}ric and Grilo, Alex Bredariol},
  publisher =	{Schloss Dagstuhl -- Leibniz-Zentrum f{\"u}r Informatik},
  address =	{Dagstuhl, Germany},
  URL =		{https://drops.dagstuhl.de/entities/document/10.4230/LIPIcs.TQC.2024.7},
  URN =		{urn:nbn:de:0030-drops-206779},
  doi =		{10.4230/LIPIcs.TQC.2024.7},
  note = {\arxiv{2406.03972}}
}

@article{low2026vtaa,
  doi = {10.22331/q-2026-03-23-2041},
  url = {https://doi.org/10.22331/q-2026-03-23-2041},
  title = {Quantum linear system algorithm with optimal queries to initial state preparation},
  author = {Low, Guang Hao and Su, Yuan},
  journal = {\quantum},
  issn = {2521-327X},
  publisher = {{Verein zur F{\"{o}}rderung des Open Access Publizierens in den Quantenwissenschaften}},
  volume = {10},
  pages = {2041},
  month = {3},
  year = {2026},
  note = {\arxiv{2410.18178}}
}

@article{chakraborty2023quantumRegularized,
  doi = {10.22331/q-2023-04-27-988},
  url = {https://doi.org/10.22331/q-2023-04-27-988},
  title = {Quantum regularized least squares},
  author = {Chakraborty, Shantanav and Morolia, Aditya and Peduri, Anurudh},
  journal = {\quantum},
  issn = {2521-327X},
  publisher = {{Verein zur F{\"{o}}rderung des Open Access Publizierens in den Quantenwissenschaften}},
  volume = {7},
  pages = {988},
  month = {4},
  year = {2023},
  note = {\arxiv{2206.13143}}
}

\onecolumn 
\appendix 

\section{Block-encoding and state-preparation constructions}\label{app:block_encodings}

Here we provide circuits for block-encoding and state preparation of various matrices and vectors used in our analysis. We assume access to the $(\alpha=1,a)$-block-encoding $U_A$ of $A$ and the state preparation unitary $U_{\vec{b}}$, as well as their inverses and controlled versions. Our constructions also involve other gates such as single-qubit rotation gates and multi-controlled Toffoli gates. For a projector $\Pi$ we denote $\CPiNOT$ as the operation that flips an ancilla bit (by applying a Pauli-$X$ gate) conditioned on being in the image of $\Pi$, and otherwise applying the identity. 
\begin{equation}
    \CPiNOT = \Pi \otimes X + (I-\Pi) \otimes I_2
\end{equation}
In quantum circuits, we draw this with a $\Pi$ box on the first register controlling a target $\oplus$ on the ancilla. When $\Pi = \ketbra{\vec{e_j}}$, this is a multi-controlled Toffoli gate, with the control bits set to $\ketbra{0}$ or $\ketbra{1}$ depending on the binary representation of $j$. 

We denote single qubit rotations by $e^{i\theta P} = \cos(\theta) I + i \sin(\theta) P$ where $P \in \{X,Y,Z\}$ is a Pauli matrix. Note that $e^{i\theta Y}\ket{0} = \cos(\theta) \ket{0} + \sin(\theta) \ket{1}$. 

Furthermore, we recall the identification of the identity matrix of dimension $d$ with $s$-qubit operators, via $I_{d} = \sum_{j=0}^{d-1} \ketbra{\vec{e_j}}$. 

\subsection{Block-encoding of $G$}

The $(1,a+1)$-block-encoding of $G = Q_{\vec{b}}A$ is depicted in \autoref{fig:block_encoding_G}. It has normalization $1$ and involves $a+1$ ancilla qubits. It costs one query to $U_A$, $U_{\vec{b}}$ and $U_{\vec{b}}^\dagger$, as well as one multi-controlled Toffoli gate.  
\begin{figure}[h!]
\centering
\scalebox{1.0}{
\begin{quantikz}[row sep={2em,between origins}, column sep=0.75em, align equals at=2]
    &   & \gate[3]{U_G} &  \\
    & \qwbundle{a} &  &  \\
    & \qwbundle{s} &  &  
\end{quantikz}
 = 
\begin{quantikz}[row sep={2em,between origins}, column sep=0.75em, align equals at=2]
    &            &               &                              & \targ{}\wire[d][2]{q}       &                    &     \\
    &\qwbundle{a}& \gate[2]{U_A} &                              &                             &                    &     \\
    &\qwbundle{s}&               & \gate{U_{\vec{b}}^{\dagger}} & \gate[1]{\ketbra{\vec{e_0}}}& \gate{U_{\vec{b}}} & 
\end{quantikz}
}
\caption{Quantum circuit implementing $(1, a+1)$-block-encoding of $G$. }\label{fig:block_encoding_G}
\end{figure}
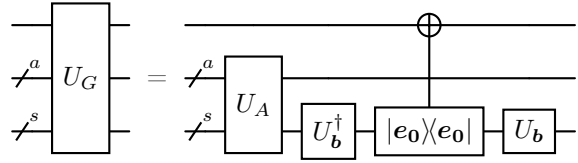

To verify the correctness of the circuit, note that the final three gates implement $I_2 \otimes I_{2^a} \otimes (I_{2^s}-\ketbra{\vec{b}}) + X \otimes I_{2^a} \otimes \ketbra{\vec{b}}$. By sandwiching this operation with $\bra{0} \cdot \ket{0}$ on the first qubit, the second term vanishes, and we see it constitutes a block-encoding of $I_{2^s}-\vec{b}\vec{b}^\dagger$ with block-encoding factor 1. The whole circuit then is a simple product of block-encodings for $A$ and for $I_{2^s}-\vec{b}\vec{b}^\dagger$. We then conclude by noting that since $\vec{e_j}^\dagger A = 0$ for all $j \geq m$, we have $G = Q_{\vec{b}}A = (I_m-\vec{b}\vec{b}^\dagger)A = (I_{2^{s}}-\vec{b}\vec{b}^\dagger)A$. 

\subsection{Block-encoding of $A_t$}

The $(1,a+1)$-block-encoding of $A_t = A + t^{-1}\ketAbraB{\vec{e_m}}{\vec{e_n}}$ is shown in \autoref{fig:block_encoding_At}. It uses one controlled query to $U_A$, two multi-controlled Toffoli gates, one controlled-rotation gate, and up to $s$ CNOT gates. It requires precomputing the angle $\arccos(1/t)$, based on the choice of $t$.
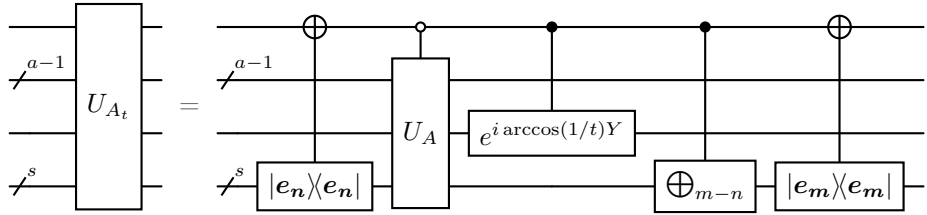
\begin{figure}[h!]
\centering
\scalebox{1.0}{
\begin{quantikz}[row sep={2em,between origins}, column sep=0.75em, align equals at=2.5]
    &                & [1em]\gate[4]{U_{A_t}} &  \\
    & \qwbundle{a-1} &                   &  \\
    &                &                   &  \\
    & \qwbundle{s}   &                   &  
\end{quantikz}
 = 
\begin{quantikz}[row sep={2em,between origins}, column sep=0.75em, align equals at=2.5]
    &              & \targ{}\wire[d][3]{q}      & \octrl{1}     & \ctrl{2}  & \ctrl{3}                         & \targ{}\wire[d][3]{q}     & \\
    &\qwbundle{a-1}&                            & \gate[3]{U_A} &                                    &               &            &\\
    &              &                            &               & \gate{e^{i\arccos(1/t)Y}}  &  &                       &\\
    &\qwbundle{s}  & \gate{\ketbra{\vec{e_n}}}  &               & & \gate{\bigoplus_{m-n}}                                   & \gate{\ketbra{\vec{e_m}}} & 
\end{quantikz}
}
\caption{Quantum circuit implementing $(1, a+1)$-block-encoding of $A_t$. The controlled-$\oplus_{m-n}$ gate denotes a series of at most $s$ CNOT gates controlled by the first register and acting on different bits of the final $s$-qubit register, which transform $\ket{\vec{e_n}}$ into $\ket{\vec{e_m}}$ if the control bit is set to 1. }\label{fig:block_encoding_At}
\end{figure}
We can verify the circuit by setting the input and output of the ancillas to $\ket{0}$ and considering possible input and output states on the $s$-qubit sytem. If the input state is not $\ket{\vec{e_n}}$, then $U_A$ is applied and $A$ is enacted on the $s$-qubit system. Since $\bra{\vec{e_m}}A \ket{\vec{e_j}} = 0$ for all indices $j$, the final gate does not flip the first ancilla. Meanwhile, if the input state is $\ket{\vec{e_n}}$ and the output state is $\ket{\vec{e_m}}$, then the first ancilla is turned to $\ket{1}$ by the first gate and returned to $\ket{0}$ by the last gate. Prior to being returned to $\ket{0}$, it controls a single-qubit rotation gate that sends $\ket{0} \mapsto t^{-1} \ket{0} + \sqrt{1-t^{-2}}\ket{1}$, and it also controls a set of NOT gates that transform the input $\ket{\vec{e_n}}$ into the output $\ket{\vec{e_m}}$. Postselecting this ancilla on $\ket{0}$ verifies the correct matrix entry $t^{-1}$ for $\ketAbraB{\vec{e_m}}{\vec{e_n}}$. Finally, if the input is $\ket{\vec{e_n}}$ and the output is not $\ket{\vec{e_m}}$, then the matrix element is seen to be zero, as expected. 

\subsection{State preparation unitary $U_{\vec{b'}}$}

The state preparation unitary that prepares $\ket{\vec{b'}} = (\ket{\vec{b}} + \ket{\vec{e_m}})/\sqrt{2}$ uses one ancilla qubit that begins and ends in $\ket{0}$. It also requires one controlled query to $U_{\vec{b}}$ and a controlled-$U_{\vec{e_m}}$ gate that prepares $\ket{\vec{e_m}}$ controlled on an ancilla---this can be accomplished with up to $s$ CNOT gates (depending on the binary representation of $m$). Finally, it utilizes one multi-controlled Toffoli gate to disentangle the ancilla after the controlled operations---to verify this action, recall that $\braket{\vec{b}}{\vec{e_m}}=0$.

\begin{figure}[h!]
\centering
\scalebox{1.0}{
\begin{quantikz}[row sep={2em,between origins}, column sep=0.75em, align equals at=1.5]
    \lstick{$\ket{0}$}         & [1em]\gate[2]{U_{\vec{b'}}} & \rstick{$\ket{0}$}\\
    \lstick{$\ket{\vec{e_0}}$} & \qwbundle{s}           & \rstick{$\ket{\vec{b'}}$}  
\end{quantikz}
 = 
\begin{quantikz}[row sep={2em,between origins}, column sep=0.75em, align equals at=1.5]
    \lstick{$\ket{0}$}&              & \gate{H} & \octrl{1}           & \ctrl{1}            & \targ{} \wire[d][1]{q} & \rstick{$\ket{0}$} \\
                      \lstick{$\ket{\vec{e_0}}$}& \qwbundle{s} &          & \gate{U_{\vec{b}}} & \gate{U_{\vec{e_m}}} & \gate{\ketbra{\vec{e_m}}} & \rstick{$\ket{\vec{b'}}$} 
\end{quantikz}
}
\caption{Quantum circuit implementing a state-preparation unitary $U_{\vec{b'}}$ with the assistance of one ancilla qubit beginning and ending in $\ket{0}$. }\label{fig:state_prep_bprime}
\end{figure}
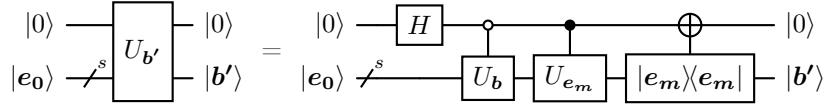

\subsection{Block-encoding of $G_t$}

The $(1,a+2)$-block-encoding of $G_t = Q_{\vec{b'}}A_t$ is depicted in \autoref{fig:block_encoding_Gt}. It is the same as the block-encoding of $G$ from \autoref{fig:block_encoding_G}, except with $U_A$ replaced by $U_{A_t}$ and $U_{\vec{b'}}$ replaced by $U_{\vec{b'}}$.

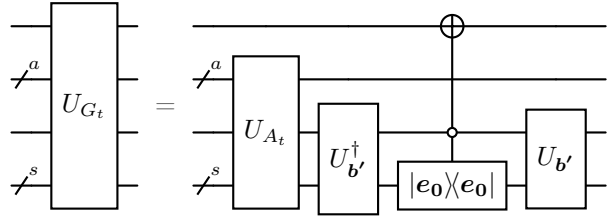
\begin{figure}[h!]
\centering
\scalebox{1.0}{
\begin{quantikz}[row sep={2em,between origins}, column sep=0.75em, align equals at=2.5]
    &              & \gate[4]{U_{G_t}} &  \\
    & \qwbundle{a} &                   &  \\
    &              &                   &  \\
    & \qwbundle{s} &                   &  
\end{quantikz}
 = 
\begin{quantikz}[row sep={2em,between origins}, column sep=0.75em, align equals at=2.5]
    &            &                   &                                  & \targ{}\wire[d][2]{q}        &                        &     \\
    &\qwbundle{a}& \gate[3]{U_{A_t}} &                                  &                              &                        &     \\
    &            &                   & \gate[2]{U_{\vec{b'}}^{\dagger}} & \octrl{1}                    & \gate[2]{U_{\vec{b'}}} &     \\
    &\qwbundle{s}&                   &                                  & \gate[1]{\ketbra{\vec{e_0}}} &                        & 
\end{quantikz}
}
\caption{Quantum circuit implementing $(1, a+2)$-block-encoding of $G_t$. }\label{fig:block_encoding_Gt}
\end{figure}

\subsection{Block-encoding of $\bar{A}_\sigma$}

The $(1,a+1)$-block-encoding of the matrix $\bar{A}_\sigma$ is the most complex, and depicted in \autoref{fig:block_encoding_barAsigma}. It uses one controlled query to $U_A$, one CNOT gate, one single-qubit rotation gate, and one $\CPiNOT$ gate with $\Pi = I_{2^s}-I_m$. The $\CPiNOT$ gate could be further decomposed into $\mathrm{poly}(s)$ Toffoli and CNOT gates, for example, by comparing the binary representation of the input state $\ket{\vec{e_j}}$ to that of $m$, starting with the most significant digit, and computing whether $m \geq j$ into the third register depicted in the figure. This strategy would require at least 1 additional (unshown) ancilla qubit which would then need to be uncomputed. 

To enact the single-qubit rotation, we precompute the angle $\varphi = \arccos(\sqrt{1-f(\sigma)^2})$ and note that $e^{-i\varphi Y}\ket{0} = \sqrt{1-f(\sigma)^2}\ket{0} - f(\sigma)\ket{1}$, and $e^{-i\varphi Y}\ket{1}= f(\sigma)\ket{0} + \sqrt{1-f(\sigma)^2}\ket{1}$.

To verify the correctness of the circuit, recall that here we view the matrix $\bar{A}_\sigma$ as an $(s+1)$-qubit operator equivalent to $\sqrt{1-f(\sigma)^2}\ketbra{0} \otimes A + f(\sigma)\ketAbraB{0}{1}\otimes I_m$, which acts on the final two registers of the circuit.
The first two gates are seen to enact a block-encoding of the operator $\ketbra{0} \otimes A + \ketbra{1} \otimes I_m$, using the second and third registers as block-encoding ancillas. Meanwhile, the final two gates are seen to enact a block-encoding of the operator $\ket{0}(\sqrt{1-f(\sigma)^2}\bra{0} + f(\sigma)\bra{1}) \otimes I_{2^s}$, using the first register as a block-encoding ancilla---this is because sandwiching the first qubit between $\bra{0} \cdot \ket{0}$ postselects on both the input and output of the CNOT gate being $\ket{0}$.  

The correctness of the circuit is then verified by multiplying the operators block-encoded by these two parts of the circuit.

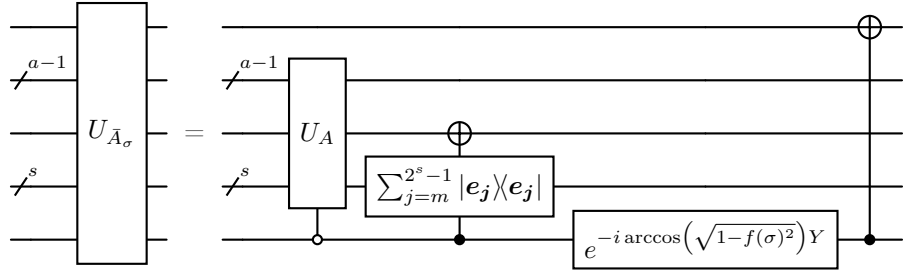
\begin{figure}[h!]
\centering
\scalebox{1.0}{
\begin{quantikz}[row sep={2em,between origins}, column sep=0.75em, align equals at=3]
    &              &[1em] \gate[5]{U_{\bar{A}_\sigma}} &  \\
    & \qwbundle{a-1}              &                       &  \\
    &   &                       &  \\
    & \qwbundle{s} &                       &  \\
    &              &                       &  
\end{quantikz}
 = 
\begin{quantikz}[row sep={2em,between origins}, column sep=0.75em, align equals at=3]
    &              & [1em]              &           &                                                                         & \targ{}   & \\
    & \qwbundle{a-1} &  \gate[3]{U_A}              &           &                                                                         &           & \\
    &              &  & \targ{}           &                                                                         &           & \\
    & \qwbundle{s} &               & \gate{\sum_{j=m}^{2^s-1}\ketbra{\vec{e_j}}}\wire[u][1]{q}          &                                                                         &           & \\
    &              & \octrl{-1}     & \ctrl{-1} & \gate{e^{-i\arccos\left(\sqrt{1-f(\sigma)^2}\right)Y}} & \ctrl{-4} &\\
\end{quantikz}
}
\caption{Quantum circuit implementing $(1, a+1)$-block-encoding of $\bar{A}_\sigma$, which acts on a $s+1$-qubit system. }\label{fig:block_encoding_barAsigma}
\end{figure}

\subsection{How to pad $A$ in the first place}\label{sec:how_to_pad}

In \autoref{sec:QLSP}, we assumed we are given an instance $(A,\vec{b})$ of the QLSP of size $m \times n$, and a block-encoding $U_A$ that views $A$ as an operator on $s$ qubits with $\max(m,n) < 2^s$. This means there is at least one row and column of zero padding built in. Here we drop this assumption, and assume that $\max(m,n) = 2^s$. This means that in order to augment the linear system we need to introduce a new basis vector $\ket{\vec{e_m}} = \ket{\vec{e_n}}$, and this requires introducing a whole new qubit and expanding the total vector space to size $2^{s+1}$. To fit this situation into our framework, we now give an alternative block-encoding of $A$ that simply pads it with zeros as $\left[\begin{smallmatrix}A & 0 \\ 0 & 0\end{smallmatrix}\right]$ when viewed as an $(s+1)$-qubit operator rather than as an $s$-qubit operator, so that $\max(m,n) < 2^{s+1}$.  This construction uses $a+1$ total ancillas, a multi-controlled Toffoli and a CNOT gate. It uses one call to $U_{A}$ but no calls to controlled $U_{A}$. This construction could then be used in place of $U_A$ everwhere else that $U_A$ appears, causing the substitution $a \mapsto a+1$, $s \mapsto s+1$. 
\begin{figure}[h!]
\centering
\scalebox{1.0}{
\begin{quantikz}[row sep={2em,between origins}, column sep=0.75em, align equals at=2.5]
    &             & \targ{}       &                & \targ{}    &  &     \\
    &\qwbundle{a}& \octrl{-1}    & \gate[2]{U_{A}} &            &  &     \\
    &\qwbundle{s}&               &                &            &  &     \\
    &             &               &                & \octrl{-3} &  &     
\end{quantikz}
}
\caption{Quantum circuit implementing $(1,a) = (1, a+1)$-block-encoding of $A$ as an operator on $s+1$ qubits (padded with zeros), given as input a $(1,a)$-block-encoding of $A$ as an operator on $s$ qubits. }\label{fig:block_encoding_padded}
\end{figure}
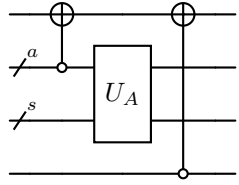
To see that this construction is correct, note the following sequence of logic. First, sandwiching the $a$-qubit register with $\bra{0} \cdot \ket{0}$ triggers the multi-controlled Toffoli that flips the first qubit. Then, sandwhiching the first qubit with $\bra{0} \cdot \ket{0}$ requires a second flip to occur, meaning nonzero output is only possible if the final qubit is in $\ket{0}$.

\section{Implementation of kernel projection and kernel reflection}\label{app:KP_KR}

\subsection{Quantum singular value tranformation}

The formalism of quantum singular value transformation (QSVT) \cite{gilyen2018QSingValTransf} allows for polynomial transformations of the singular values of a block-encoded matrix. 

Let $U_B$ be an $(1,a_B)$-block-encoding of a matrix $B = \sum_{j=0}^{2^s-1}\sum_{k=0}^{2^s-1} b_{jk} \ketAbraB{\vec{e_j}}{\vec{e_k}}$ (viewed as an operator on $s$ qubits with appropriate zero-padding) and let $B = W\Sigma V^\dagger$
be a singular value decomposition (SVD) of $B$, where $\Sigma$ is the diagonal matrix of singular values. Let $P$ be a real, even, degree-$d$ polynomial for which $|P(x)|\leq 1$ whenever $|x|\leq 1$ (QSVT also works for odd polynomials, but we only need even polynomials in our application). Corollary 18 of Ref.~\cite[arXiv version]{gilyen2018QSingValTransf} states that we can find a quantum circuit implementing $U_{P}^{(B)}$, which is a $(1,a_B+1)$-block-encoding of $VP(\Sigma)V^\dagger$:
\begin{equation}
  V P(\Sigma) V^\dagger = (\bra{0}^{\otimes(a_B+1)} \otimes I_{2^s}) U_{P}^{(B)}(\ket{0}^{\otimes(a_B+1)} \otimes I_{2^s})\,,
\end{equation}
The unitary $U^{(B)}_P$ uses $d/2$ calls to each of $U_B$ and $U_B^\dagger$, $2d$ multi-controlled Toffoli gates, and $d$ single-qubit rotation gates. The circuit for $U_P^{(B)}$ is given in \autoref{fig:QSVT_unitary}. 

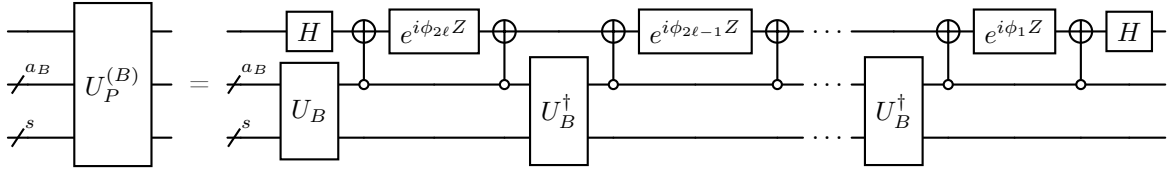
\begin{figure}[h!]
\centering
\scalebox{1.0}{
\begin{quantikz}[row sep={2em,between origins}, column sep=0.75em, align equals at=2]
    &                & [1em]\gate[3]{U_P^{(B)}} &  \\
    & \qwbundle{a_B} &                       &  \\
    & \qwbundle{s}   &                       &  
\end{quantikz}
 = 
\begin{quantikz}[row sep={2em,between origins}, column sep=0.5em, align equals at=2]
    &                & [1em]\gate{H}          & \targ{}    & \gate{e^{i\phi_{2\ell} Z}} & \targ{}    &                       & \targ{}    & \gate{e^{i\phi_{2\ell-1} Z}} & \targ{}    & \ \ldots\ &                       & \targ{}    & \gate{e^{i\phi_{1} Z}} & \targ{}    & \gate{H} & \\
    & \qwbundle{a_B} & \gate[2]{U_B}          & \octrl{-1} &                            & \octrl{-1} & \gate[2]{U_B^\dagger} & \octrl{-1} &                               & \octrl{-1} & \ \ldots\ & \gate[2]{U_B^\dagger} & \octrl{-1} &                        & \octrl{-1} &          & \\
    & \qwbundle{s}   &                        &            &                            &            &                       &            &                               &            & \ \ldots\ &                       &            &                        &            &          &          
\end{quantikz}
}
\caption{Quantum circuit implementing the QSVT unitary $U_{P}^{(B)}$. The phases $\phi_{j}$ must be chosen to implement the polynomial $P$. Conditioned on the input and output of the first $a_B+1$ qubits being $\ket{0}^{\otimes(a_B+1)}$, the operation $VP(\Sigma)V^\dagger$ is applied to the state on the third register, where $B = W \Sigma V^\dagger$ is a SVD of $B$. The gate $H$ denotes Hadamard.  }\label{fig:QSVT_unitary}
\end{figure}

\subsection{Kernel projection}
We can apply QSVT to perform eigenstate filtering (EF) \cite{lin2019OptimalQEigenstateFiltering}, or more precisely in our application, kernel projection (KP). This technique gives a procedure to project onto the kernel of a matrix by constructing an appropriate polynomial filter function and applying QSVT. Specifically, given target parameters $\Delta \in (0,1)$ and $\eta \in (0,1]$, we take
\begin{equation} \label{eq:ell_KP}
    \ell = \left\lceil \frac{\arccosh(\eta^{-1})}{\arccosh\left(\frac{1+\Delta^2}{1-\Delta^2}\right)} \right\rceil \leq \left\lceil \frac{1}{2\Delta}\ln\left(\frac{2}{\eta}\right) \right\rceil
\end{equation}
and we define the even, degree-$2\ell$ polynomial $F_{\Delta,\ell}(x)$ given in Ref.~\cite{lin2019OptimalQEigenstateFiltering}
\begin{equation}
    F_{\Delta, \ell}(x) = \frac{T_\ell\left(\frac{1+\Delta^2-2x^2}{1-\Delta^2}\right)}{T_\ell\left(\frac{1+\Delta^2}{1-\Delta^2}\right)}
\end{equation}
where $T_\ell$ is the $\ell$th Cheybshev polynomial of the first kind, given by 
\begin{equation}
    T_\ell(z) = \begin{cases}
        \cos(\ell \arccos(z)) & \text{if } |z|\leq 1 \\
        \cosh(\ell \arccosh(z)) & \text{if } z >1 \\
        (-1)^\ell \cosh(\ell \arccosh(z)) & \text{if } z < -1
    \end{cases}\,.
\end{equation}
This polynomial is plotted in \autoref{fig:polynomials}. 
\begin{figure}[ht!]
    \centering
    \includegraphics[width=0.6\textwidth]{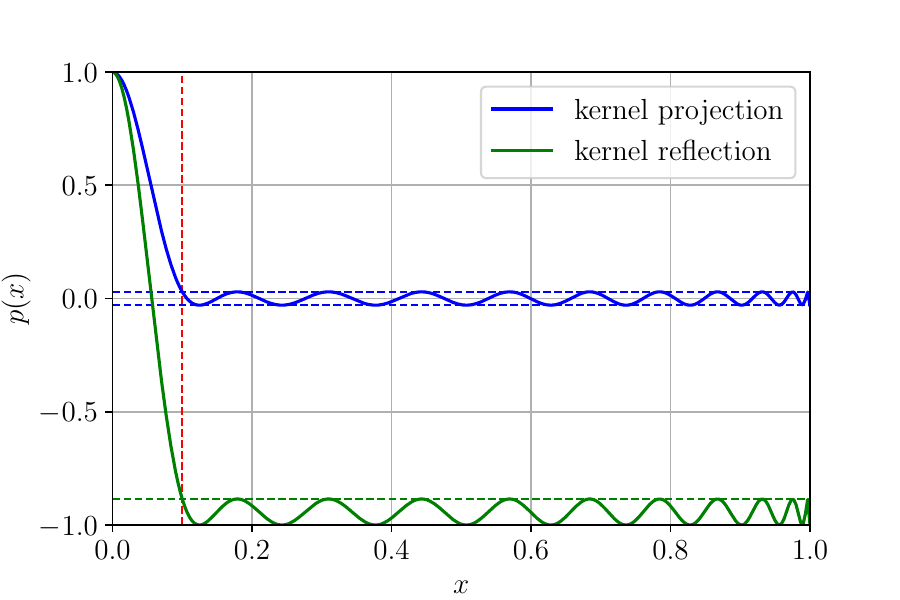}
    \caption{\label{fig:polynomials} Example of $F_{\Delta,\ell}(x)$ (kernel projection) and $R_{\Delta,\ell}(x)$ (kernel reflection) for $\Delta = 0.1$ and $\ell = 21$, implying $\eta \leq 0.03$. }
\end{figure}
\begin{lemma}\label{lem:KP_poly_properties}
    The polynomial $F_{\Delta,\ell}$ is guaranteed to satisfy 
\begin{enumerate}
    \item For all $x \in [-1,1]$, it holds that $|F_{\Delta,\ell}(x)| \leq 1$
    \item For all $x \in [\Delta, 1]$, it holds that $|F_{\Delta,\ell}(x)| \leq T_\ell(\frac{1+\Delta^2}{1-\Delta^2})^{-1} \leq \eta$. 
    \item $F_{\Delta,\ell}(0) = 1$. 
\end{enumerate}
\end{lemma}
\begin{proof}
    This builds on and improves upon Lemma 13 of Ref.~\cite{lin2019OptimalQEigenstateFiltering}.\footnote{The lemma stated here has a constant prefactor for $\ell$ that is better than that of Lemma 13 of Ref.~\cite{lin2019OptimalQEigenstateFiltering} by a factor of $\sqrt{2}$. Additionally, we no longer require the restriction $\Delta \in (0,1/\sqrt{12}]$. We acknowledge communciation with Yu Tong on how the constant prefactor in their bound could be improved, a fact that was also pointed out in Ref.~\cite[Section V]{costa2021OptimalLinearSystem}.} Note that the Chebyshev polynomial $T_\ell(z)$ is bounded on $[-1,1]$ when $z \in [-1,1]$. Moreover, for $z \geq 1$ it is monotonically increasing. Thus, the numerator achieves its maximum on $x \in [-1,1]$ when $x = 0$, where $F_{\Delta,\ell} = 1$ by construction, verifying properties 1 and 3. To verify Property 2, note that for all $x\in[\Delta,1]$, the numerator is bounded in the range $[-1,1]$. Moreover, the absolute value of the denominator evaluates to 
    \begin{align}
        T_\ell\left(\frac{1+\Delta^2}{1-\Delta^2}\right) = \cosh\left(\ell \arccosh\left(\frac{1+\Delta^2}{1-\Delta^2}\right)\right)
    \end{align}
    We immediately see that if $\ell \geq \arccosh(\eta^{-1})/\arccosh((1+\Delta^2)/(1-\Delta^2))$, then the denominator is larger than $\eta^{-1}$. 
\end{proof}

In other words, $F_{\Delta,\ell}(x)$ is a filter of width $\Delta$, which maps the zero input to one and all inputs outside the window $[-\Delta,\Delta]$ to a value close to zero. In fact, the polynomial $F_{\Delta,\ell}(x)$ is the optimal filter polynomial of degree $2\ell$ in the sense of minimizing the maximum of $|F_{\Delta,\ell}(x)|$ outside the window $[-\Delta,\Delta]$ subject to $F_{\Delta,\ell}(0)=1$ \cite[Appendix E]{lin2019OptimalQEigenstateFiltering}.

We can apply QSVT using the polynomial $F_{\Delta,\ell}$. Suppose we have a $(1,a_B)$-block-encoding of $B = W\Sigma V^\dagger$, which has no singular values in the interval $(0,\Delta)$. We apply KP to a state $\ket{\phi} = c_1 \ket{\vec{a}} + c_2\ket{\vec{b}}$, where $\vec{a}$ is in the kernel of $B$ and $\vec{b}$ is orthogonal to the kernel. The vector $\ket{\vec{b}}$ has a decomposition $\sum_j w_j \ket{\vec{u_j}}$ into right singular vectors $\ket{\vec{u_j}}$ of $B$ all of which have singular values $\varsigma_j$ in $(\Delta,1)$. By Property 2 of \autoref{lem:KP_poly_properties}, for every $j$, $ F_{\Delta,\ell}(\varsigma_j)$ lies in the interval $[-\eta,\eta]$. Thus, the application of $VF_{\Delta,\ell}(\Sigma)V^\dagger$ maps
\begin{align}\label{eq:KP_b_map}
    \ket{\vec{b}}\mapsto \sum_j w_j F_{\Delta,\ell}(\varsigma_j) \ket{\vec{u_j}} = \delta_1\ket{\vec{b}} + \delta_2\ket{\vec{b'}}
\end{align}
where $\vec{b'}$ is orthogonal to $\vec{b}$ and to the kernel of $B$, and $\delta_1^2+\delta_2^2 \leq \eta^2$. This allows us to assert
\begin{equation}\label{eq:KP_action}
    U_{F_{\Delta,\ell}}^{(B)}(\ket{0}^{\otimes(a_B+1)} \otimes \ket{\phi}) = \underbrace{\left[\ket{0}^{\otimes(a_B+1)} \otimes (c_1  \ket{\vec{a}} + c_2\delta_1\ket{\vec{b}} + c_2\delta_2 \ket{\vec{b'}})\right]}_{\text{success}}+ \underbrace{|c_2|\sqrt{1-\delta_1^2-\delta_2}\ket{{\perp}}}_{\text{failure}}
\end{equation}
where $\ket{{\perp}}$ is a normalized state for which $(\bra{0}^{\otimes(a_B+1)}\otimes I_{2^s})\ket{{\perp}} = 0$.  Postselecting on the first register being $\ket{0}^{\otimes(a_B+1)}$, this implements an approximate projector onto the kernel of $B$. Accordingly, we can detect the success of KP by measuring the $a_B+1$ ancillas and checking if they are all $\ket{0}$.

In our application, the matrix $B$ is given by $Q_{\vec{b}}A$ or $Q_{\vec{b'}}A_t$. The number of ancillas in these two cases is $a_B = a+1$ and $a_B = a+2$, respectively, where $a$ is the number of ancillas for the block-encoding $U_A$. In both cases, we can give a $(1,a_B)$-block-encoding of $B$ using one (potentially controlled) query to each of $U_A$, $U_{\vec{b}}$, and $U_{\vec{b}}^\dagger$, as illustrated in \autoref{app:block_encodings}. 

\begin{lemma}[Using KP to refine QLSS solution]\label{lem:KP}
Suppose $\vec{b}$ is in the column space of $A$, and let $\vec{x}$ denote the solution of minimum $\nrm{\vec{x}}$ to the equation $A \vec{x} = \vec{b}$. Suppose $A$ has no singular values in the interval $(0,\kappa^{-1})$ and let $\tilde{\rho}$ be a mixed quantum state for which $\bra{\vec{x}}\tilde{\rho} \ket{\vec{x}} = 1-\mu^2$ and $\bra{\vec{y}}\tilde{\rho}\ket{\vec{y}} = 0$ for all $\vec{y}$ that are in the kernel of $A$. Suppose KP is applied to approximately project $\tilde{\rho}$ onto the kernel of $G = Q_{\vec{b}}A$ using parameters $(\kappa,\eta)$. That is, we choose $\ell$ using Eq.~\eqref{eq:ell_KP} with $\Delta = \kappa^{-1}$ and apply the unitary $U_{F_{\kappa^{-1},\ell}}^{(G)}$, and then measure the ancillas to determine success. KP succeeds with probability at least $1-\mu^2$, and when it succeeds it outputs a state $\tilde{\rho}_{\rm out}$ for which
\begin{align}
     \bra{\vec{x}}\tilde{\rho}_{\rm out}\ket{\vec{x}} &\geq 1-\frac{\mu^2\eta^2}{1-\mu^2+\mu^2\eta^2} \geq 1-\frac{\mu^2\eta^2}{1-\mu^2}\\
    \frac{1}{2}\nrm{\ketbra{\vec{x}}-\tilde{\rho}_{\rm out}}_1 &\leq \frac{\mu \eta}{\sqrt{1-\mu^2+\mu^2\eta^2}} \leq \frac{\mu\eta}{\sqrt{1-\mu^2}}
\end{align}
\end{lemma}
\begin{proof}
    Let $\tilde{\rho} = \sum_i p_i \ketbra{\phi_i}$ be a decomposition of $\tilde{\rho}$ as an ensemble of pure states. 
    We have $\bra{\vec{x}}\tilde{\rho}\ket{\vec{x}} = \sum_i p_i |\braket{\vec{x}}{\phi_i}|^2 = 1-\mu^2$. Referencing Eq.~\eqref{eq:KP_action}, here we have $\ket{\vec{a}} = \ket{\vec{x}}$ and $c_1 = |\braket{\vec{x}}{\phi_i}|$. Let $\ket{\psi_i}$ denote the output state when KP acts on $\ket{\phi_i}$ and we postselect on success, and let $q_i$ be the probability of success. We can see that $q_i \geq |c_1|^2 = |\braket{\vec{x}}{\phi_i}|^2$. We can also see that $|\braket{\vec{x}}{\psi_i}|^2 = |c_1|^2/q_i$.  The overall probability of success of KP is given by $\sum_i p_i q_i \geq \sum_i p_i|\braket{\vec{x}}{\phi_i}|^2 = 1-\mu^2$, showing the first claim of the lemma. Again referring to Eq.~\eqref{eq:KP_action}, we have $|\delta_1|^2+|\delta_2|^2 \leq \eta^2$, so we can also say $q_i \leq |\braket{\vec{x}}{\phi_i}|^2 +\eta^2(1-|\braket{\vec{x}}{\phi_i}|^2)$. The output state $\tilde{\rho}_{\rm out}$ is given by
    \begin{equation}
        \tilde{\rho}_{\rm out} = \frac{\sum_i p_i q_i \ketbra{\psi_i}}{\sum_i p_i q_i}
    \end{equation}
    where including the denominator ensures the state is normalized.  We now compute
    \begin{align}
        \bra{\vec{x}}\tilde{\rho}_{\rm out}\ket{\vec{x}} &= \frac{\sum_i p_i q_i |\braket{\vec{x}}{\psi_i}|^2}{\sum_i p_i q_i} = \frac{\sum_i p_i  |\braket{\vec{x}}{\phi_i}|^2}{\sum_i p_i q_i} =\frac{1-\mu^2}{\sum_i p_i q_i} \\
        &\geq \frac{1-\mu^2}{\sum_i p_i (|\braket{\vec{x}}{\phi_i}|^2 +\eta^2(1-|\braket{\vec{x}}{\phi_i}|^2))}= \frac{1-\mu^2}{\eta^2 + (1-\mu^2)(1-\eta^2)} = 1-\frac{\mu^2\eta^2}{1-\mu^2+\mu^2\eta^2}
    \end{align}
    The trace distance bound then follows from the general relationship $\frac{1}{2}\nrm{\ketbra{\vec{x}} - \sigma} \leq \sqrt{1-\bra{\vec{x}}\sigma \ket{\vec{x}}}$ for any state $\sigma$ \cite[Eq.~(9.110)]{nielsen2002QCQI}. 
\end{proof}

\subsection{Kernel reflection}

We can extend the technique of kernel projection to \textit{kernel reflection}. In kernel reflection, rather than preserving the kernel and (approximately) zeroing out its orthogonal complement, we preserve the kernel and (approximately) apply a $-1$ phase to the orthogonal complement.  We do so within the framework of QSVT by constructing a new polynomial $K_{\Delta,\ell}$. Again, we are given parameters $\Delta$ and $\eta$, and we choose $\ell$ as in Eq.~\eqref{eq:ell_KP}. We define the $2\ell$-degree polynomial by scaling and shifting $F_{\Delta,\ell}$ to cover the range $[-1,1]$:
\begin{equation}
    K_{\Delta,\ell}(x) = \frac{2F_{\Delta,\ell}(x)-1+F_{\Delta,\ell}(\Delta)}{1+F_{\Delta,\ell}(\Delta)} = \frac{2T_{\ell}(\frac{1+\Delta^2-2x^2}{1-\Delta^2})+2}{T_{\ell}(\frac{1+\Delta^2}{1-\Delta^2})+1}-1\,.
\end{equation}
$K_{\Delta,\ell}$ is a degree-$2\ell$, even polynomial. 

\begin{lemma}\label{lem:kernel_reflection_properties}
    The polynomial $K_{\Delta,\ell}(0)$ is guaranteed to satisfy 
\begin{enumerate}
    \item For all $x \in [-1,1]$, it holds that $|K_{\Delta,\ell}(x)| \leq 1$
    \item For all $x \in [\Delta, 1]$, it holds that $|K_{\Delta,\ell}(x)| \leq -1+\frac{4}{T_\ell(\frac{1+\Delta^2}{1-\Delta^2})+1} \leq -1 + \frac{4\eta}{1+\eta}$. 
    \item $K_{\Delta,\ell}(0) = 1$. 
\end{enumerate}
\end{lemma}
\begin{proof}
    Property 3 is readily verified by plugging in $x=0$. Property 1 follows from the fact that for $x \in [-1,1]$, the maximum of $T_\ell(\frac{1+\Delta^2-2x^2}{1-\Delta^2})$ is achieved at $x=0$  and the minimum is $-1$ (as noted in proof of \autoref{lem:KP_poly_properties}).  To verify Property 2, first note that for all $x\in[\Delta,1]$, $T_\ell(\frac{1+\Delta^2-2x^2}{1-\Delta^2}) \in [-1,1]$. Moreover, the denominator evaluates to 
    \begin{align}
         T_\ell\left(\frac{1+\Delta^2}{1-\Delta^2}\right) +1  = \cosh\left(\ell \arccosh\left(\frac{1+\Delta^2}{1-\Delta^2}\right)\right) + 1
    \end{align}
    We immediately see that if $\ell \geq \arccosh(\eta^{-1})/\arccosh((1+\Delta^2)/(1-\Delta^2))$, then the denominator is larger than $\eta^{-1}+1 = \eta/(1+\eta)$, which completes the verification. 
\end{proof}

We can apply QSVT using the polynomial $K_{\Delta,\ell}$. Suppose we have a $(1,a_B)$-block-encoding of $B$, which has no singular values in the interval $(0,\Delta)$. We apply kernel reflection to a state $\ket{\phi} = c_1 \ket{\vec{a}} + c_2\ket{\vec{b}}$, where $\vec{a}$ is in the kernel of $B$ and $\vec{b}$ is orthogonal to the kernel. Building on Eq.~\eqref{eq:KP_b_map}, we have
\begin{align}
    \ket{\vec{b}}\mapsto \sum_j w_j K_{\Delta,\ell}(\varsigma_j) \ket{\vec{u_j}} &= \sum_j w_j \frac{2F_{\Delta,\ell}(\varsigma_j)-1+F_{\Delta,\ell}(\Delta)}{1+F_{\Delta,\ell}(\Delta)} \ket{\vec{u_j}}  \\
    &= - \left(1-\frac{2F_{\Delta,\ell}(\Delta)}{1+F_{\Delta,\ell}(\Delta)}\right)\ket{\vec{b}}+ \frac{2}{1+F_{\Delta,\ell}(\Delta)}(\delta_1\ket{\vec{b}} + \delta_2\ket{\vec{b'}}) \\
    &= -(1-\delta_1') \ket{\vec{b}} + \delta_2' \ket{\vec{b'}}
\end{align}
where $\delta_1' = \frac{2F_{\Delta,\ell}(\Delta)+2\delta_1}{1+F_{\Delta,\ell}(\Delta)}$ and $\delta_2' = \frac{2\delta_2}{1+F_{\Delta,\ell}(\Delta)}$. We know that $\eta \geq F_{\Delta,\ell}(\Delta)$; if we take $\eta = F_{\Delta,\ell}(\Delta)$, we recover the statements in the main text: $\delta_1' = \frac{2\eta + 2\delta_1}{1+\eta}$ and $\delta_2' = \frac{2\delta_2}{1+\eta}$.
\begin{equation}\label{eq:KR_impact}
    U_{K_{\Delta,\ell}}^{(B)}(\ket{0}^{\otimes(a_B+1)} \otimes \ket{\phi}) = \underbrace{\left[\ket{0}^{\otimes(a_B+1)} \otimes (c_1  \ket{\vec{a}} - c_2(1-\delta_1')\ket{\vec{b}} + c_2\delta_2'\ket{\vec{b'}})\right]}_{\text{success}} + \underbrace{|c_2|\sqrt{2\delta_1'-\delta_1'^2-\delta_2'^2}\ket{{\perp}}}_{\text{failure}}\,.
\end{equation}
 From the bound $|\delta_1|\leq \eta$, we have $0 \leq \delta_1' \leq \frac{4\eta}{1+\eta}$. From the bound $\sqrt{\delta_1^2+\delta_2^2}$ we have  $\delta'_2 \leq \sqrt{\delta_1'(\frac{4\eta}{1+\eta}-\delta'_1)}$.

\section{QLSS given norm estimate: detailed analysis}\label{app:QLSS_with_norm_estimate}

First, we provide more specific details on how \autoref{algo:main_algo} is implemented by providing a quantum circuit. The full circuit is given in \autoref{fig:main_algo_circuit}. 
\begin{figure}[h!]
\centering
\scalebox{1.0}{
\begin{quantikz}[row sep={2em,between origins}, column sep=0.75em, align equals at=2]
    \lstick{$\ket{0}$}              & [0.5em]\slice[style=red]{}                & [3em]                      &  \slice[style=red]{}         &[3em]\targ{} \wire[d][2]{q}    & \meter{} \\
    \lstick{$\ket{0}^{\otimes(a+2)}$}& \qwbundle{a+2} & \gate[2]{U^{(G_t)}_{K_{\Delta,\ell}}}  & \meter{} \\
    \lstick{$\ket{\vec{e_n}}$}       & \qwbundle{s}   &                            &          &\gate{\ketbra{\vec{e_n}}} &               & \rstick{$\ket{\vec{\tilde{x}}}$ if all measurement outcomes $\ket{0}$}  
\end{quantikz}
}
\caption{Quantum circuit implementing \autoref{algo:main_algo}. The gate $U_{K_{\Delta,\ell}}^{(G_t)}$ is the QSVT circuit depicted in \autoref{fig:QSVT_unitary} that implements the polynomial $K_{\Delta,\ell}$ on the singular values of $G_t$, leveraging the $(1,a+1)$-block-encoding of $G_t$ from \autoref{fig:block_encoding_Gt}.  Vertical lines separate steps 1, 2, and 3 in the algorithm. Step 2 (KR) succeeds if all $a+2$ ancillas in the second register are measured in $\ket{0}$. Step 3 succeeds if the first ancilla is measured to be $\ket{0}$. Conditioned on success, the output is  $\ket{\vec{\tilde{x}}}$. }\label{fig:main_algo_circuit}
\end{figure}
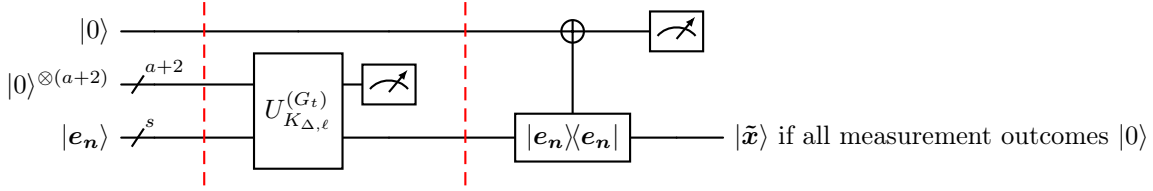

\begin{theorem}\label{thm:QLSS_known_norm}
Suppose that $\nrm{\vec{b}} = 1$, that $\vec{b}$ is in the column space of $A$, and that all nonzero singular values of $A$ lie in the interval $[\kappa^{-1},1]$. Let $U_A$ be a $(1,a)$-block-encoding of $A$, and $U_{\vec{b}}$ be a state-preparation unitary for $\ket{\vec{b}}$.  Let $\vec{x}$ denote the unique vector of minimum norm $\nrm{\vec{x}}$ for which $A\vec{x} = \vec{b}$. 

Fix parameter choices $\eta > 0$, and $t \in [1,\kappa]$, letting $\theta_t = \arctan(\nrm{\vec{x}}/t)$.  Consider \autoref{algo:main_algo}, denoting the overall success probability by $p_{\rm succ}$, and the output state conditioned on success by $\ket{\vec{\tilde{x}}}$. Then, the state $\ket{\vec{\tilde{x}}}$ has no overlap with the kernel of the operator $A$. Furthermore, the following bounds are satisfied (the first three require $\eta/\cos(\theta_t)\leq 1$)
\begin{align}
    \nrm{\ket{\vec{x}} - \ket{\vec{\tilde{x}}}} &\leq \arcsin\left(\frac{\eta}{\cos(\theta_t)}\right) \label{eq:output_dist_bound}
    \\
    \frac{1}{2}\nrm{\ketbra{\vec{x}}-\ketbra{\vec{\tilde{x}}}}_1 &\leq \frac{\eta}{\cos(\theta_t)} \label{eq:trace_distance_bound}
    \\
    \lvert \braket{\vec{\tilde{x}}}{\vec{x}}\rvert &\geq \sqrt{1-\frac{\eta^2}{\cos^2(\theta_t)}} \label{eq:overlap_bound}
    \\
   \sin^2(2\theta_t)\left(\frac{1-\eta}{1+\eta}\right)^2 &\leq p_{\rm succ} \leq \sin^2(2\theta_t)+\frac{4\eta^2}{(1+\eta)^2}\sin^2(\theta_t)\label{eq:p_succ_bounds}
\end{align}
The query cost of the protocol is $\ell$ queries to controlled-$U_A$, $\ell$ queries controlled-$U_A^\dagger$, $2\ell$ queries to controlled-$U_{\vec{b}}$, and $2\ell$ queries to controlled-$U_{\vec{b}}^\dagger$ where
\begin{equation}\label{eq:ell_kappa}
    \ell = \left\lceil \frac{\arccosh(\eta^{-1})}{\arccosh\left(\frac{1+\kappa^{-2}}{1-\kappa^{-2}}\right)} \right\rceil  \leq \left\lceil \frac{\kappa}{2}\ln(\frac{2}{\eta}) \right\rceil 
\end{equation} 
The algorithm requires $a+3$ total ancilla qubits, including the $a$ ancilla qubits for $U_A$. 
\end{theorem}
\begin{proof}
    The stated query cost follows from the observations in \autoref{app:KP_KR}. This cost is incurred entirely during Step 2 of the algorithm (KR). 

    Note that since $\vec{x}$ is defined as the solution to $A\vec{x} = \vec{b}$ with minimum Euclidean norm $\nrm{\vec{x}}$, we have that $\vec{x}$ is orthogonal to the kernel of $A$. If it were not, we could project $\vec{x}$ onto the orthogonal complement of the kernel and obtain a vector $\vec{x'}$ for which $A\vec{x'} = \vec{b}$ and $\nrm{\vec{x'}}\leq \nrm{\vec{x}}$, leading to a contradiction. 
    
    We now justify the stated bounds by tracking the state of throughout the algorithm. Recall that $\ket{\vec{x_t}} = \sin(\theta_t)\ket{\vec{x}} + \cos(\theta_t)\ket{\vec{e_n}}$. Define $\ket{\vec{y_t}}= -\cos(\theta_t)\ket{\vec{x}} + \sin(\theta_t)\ket{\vec{e_n}}$ orthogonal to $\ket{\vec{x_t}}$. Note that $\ket{\vec{y_t}}$ is orthogonal to the kernel of $A$.  In Step 1, we prepare $\ket{\vec{e_n}}$, which  decomposes as
    \begin{equation}
        \ket{\vec{e_n}} = \cos(\theta_t)\ket{\vec{x_t}} + \sin(\theta_t)\ket{\vec{y_t}}
    \end{equation}
In Step 2, we apply KR for the matrix $G_t = Q_{\vec{b'}}A_t$, which preserves $\ket{\vec{x_t}}$ and approximately flips the sign on $\ket{\vec{y_t}}$.  That is, conditioned on success of KR, using Eq.~\eqref{eq:KR_impact}, the resulting state is (up to normalization)
    \begin{equation}\label{eq:post_KP}
        \cos(\theta_t)\ket{\vec{x_t}} - (1-\delta_1')\sin(\theta_t)\ket{\vec{y_t}} +  \delta_2' \sin(\theta_t)\ket{\vec{z}}
    \end{equation}
    where $\ket{\vec{z}}$ is some state orthogonal to both $\ket{\vec{x_t}}$ and $\ket{\vec{y_t}}$, $\delta_1' \leq 4\eta/(1+\eta)$, and $\delta_2' \leq \sqrt{\delta_1'(4\eta/(1+\eta)-\delta_1')}$. Furthermore, since the kernel of $A$ is contained in the kernel of $G_t = Q_{\vec{b'}}A_t$ and $\vec{y_t}$ is orthogonal to the kernel of $A$, the image of $\vec{y_t}$ under any QSVT sequence will remain orthogonal to the kernel of $A$. In particular, this implies that $\ket{\vec{z}}$ is orthogonal to the kernel of $A$.
    
   In Step 3, we project onto the complement of $\ket{\vec{e_n}}$. Note that $\ket{\vec{z}}$ is orthogonal to $\ket{\vec{e_n}}$ and to $\ket{\vec{x}}$, both of which lie in the span of $\{\ket{\vec{x_t}},\ket{\vec{y_t}}\}$. The state that results is (up to normalization)
    \begin{align}
        \ket{\vec{\tilde{x}}} &\propto \Big[\sin(\theta_t)\cos(\theta_t) + (1-\delta_1')\cos(\theta_t)\sin(\theta_t)\Big]\ket{\vec{x}} +  \delta_2'\sin(\theta_t)\ket{\vec{z}} \\
        &=\left(1-\frac{\delta_1'}{2}\right)\sin(2\theta_t)\ket{\vec{x}} + \delta_2' \sin(\theta_t) \ket{\vec{z}}\label{eq:post_step3_KR}
    \end{align}
    This state is orthogonal to the kernel of $A$ since both $\vec{x}$ and $\vec{z}$ are orthogonal to the kernel of $A$. 

    We verify the bounds in reverse order. The overall probability of success is given by the squared norm of the right-hand side of Eq.~\eqref{eq:post_step3_KR}. The lower bound in Eq.~\eqref{eq:p_succ_bounds} follows from the observation that the squared norm is minimized when $\delta_1'=4\eta/(1+\eta)$ and $\delta_2'=0$.
    The upper bound in Eq.~\eqref{eq:p_succ_bounds} is a loose upper bound generated by maximizing the norm of each term in Eq.~\eqref{eq:post_step3_KR} individually, without enforcing $\delta_2'^2 \leq \delta_1'(4\eta/(1+\eta)-\delta_1') \leq 4\eta^2/(1+\eta)^2$. These terms achieve maximum norm at $\delta_1'=0$ and $\delta_2' = 2\eta/(1+\eta)$, yielding the stated upper bound. 
    
    To show Eq.~\eqref{eq:overlap_bound}, to condense some notation, define $M=\eta/(1+\eta)$, $
    \xi_1 = \delta'_1/2-M$,  $\xi_2 = \delta'_2/2$, and $\gamma=1/\cos(\theta_t)$. We have $\xi_1^2 + \xi_2^2 \leq M^2$. After normalizing $\ket{\vec{\tilde{x}}}$ appropriately, and noting that $\delta_2'\sin(\theta_t) = \sin(2\theta_t)\gamma \xi_2$, we see that the overlap is given by
    \begin{align}
        \lvert \braket{\vec{\tilde{x}}}{\vec{x}}\rvert &= \frac{1-M-\xi_1}{\sqrt{(1-M-\xi_1)^2 + \gamma^2\xi_2^2}} \\
         &\geq \frac{1-M-\gamma\xi_1}{\sqrt{(1-M-\gamma\xi_1)^2 + \gamma^2\xi_2^2}}
    \end{align}
    where imposing the inequality puts the expression in a form that easier to minimize.  The pair $(\xi_1,\xi_2)$ where it is minimized will be among the points for which $|\xi_1|^2 + |\xi_2|^2$ achieves its maximum of $M^2$, so we let $\xi_1 = M \sin(\varphi)$ and $\xi_2 = M \cos(\varphi)$ for some $\varphi$ to be specified later. Our lower bound on the overlap simplifies to
    \begin{equation}
        \frac{1-M - M \gamma\sin(\varphi)}{\sqrt{(1-M-M\gamma\sin(\varphi))^2 + \gamma^2 M^2\cos(\varphi))^2}}
    \end{equation}
    We can upper bound the final expression by a geometric argument, aided by \autoref{fig:geometric_argument}. 

\begin{figure}[ht]
    \centering
\scalebox{0.8}{
\begin{tikzpicture}[scale=0.8,thick,baseline=0mm]
\def\radius{4}
\def\arcRadius{\radius*0.35}
\def\arcRadiusL{\radius*0.2}
\def\pointRadius{0.14}
\def\sinPhi{0.34}
\def\cosPhi{0.9397}
\def\phiDegrees{20}
\def\sinPhiOverTwo{0.174}
\def\cosPhiOverTwo{0.9848}
%
\coordinate (D) at (0,0);
\coordinate (C) at (-\radius*\sinPhi, \radius*\cosPhi);
\coordinate (B) at (-\radius*\sinPhi, 0);
\coordinate (A) at (-\radius/\sinPhi, 0);
%
\draw[black,thick] (D)--(C);
\draw[black,thick] (D)--(B);
\draw[black,thick] (B)--(C);
\draw[black,thick] (A)--(B);
\draw[black,thick] (A)--(C);
%
\draw[dashed, thick] (D) circle (\radius);
%
\draw[fill] (A) circle(\pointRadius);
\draw[fill] (B) circle(\pointRadius);
\draw[fill] (C) circle(\pointRadius);
\draw[fill] (D) circle(\pointRadius);
%
\node at (-\radius/\sinPhi-\radius*0.2, 0) {\LARGE $A$};
\node at (0, -\radius*0.15) {\LARGE $
D$};
\node at (-\radius*\sinPhi, -\radius*0.15) {\LARGE $B$};
\node at (-\radius*\sinPhi*1.15, \radius*\cosPhi*1.15) {\LARGE $C$};
%
\draw (-\radius/\sinPhi+\arcRadiusL/\sinPhi, 0) arc(0:\phiDegrees:\arcRadiusL/\sinPhi);
\draw (-\radius*\sinPhi, \radius*\cosPhi-\arcRadius) arc(270:270+\phiDegrees:\arcRadius);
\node at (-\radius/\sinPhi+1.17*\arcRadiusL*\cosPhiOverTwo/\sinPhi,1.17*\arcRadiusL*\sinPhiOverTwo/\sinPhi) {\Large $\varphi$};
\node at (-\radius*\sinPhi + \arcRadius*\sinPhiOverTwo*1.25,\radius*\cosPhi -1.25*\arcRadius*\cosPhiOverTwo) {\Large $\varphi$};
\end{tikzpicture}
}
    \caption{Geometric argument for rigorous lower bound on overlap, see text. The angle $\angle BCD$ is essentially defined as $\varphi$. We then deduce that if $C$ is constrained to lie on the circle, then the maximum value of $\angle BAC$ is $\varphi$ which is achieved when $AC$ is tangent to the circle.  }
    \label{fig:geometric_argument}
\end{figure}
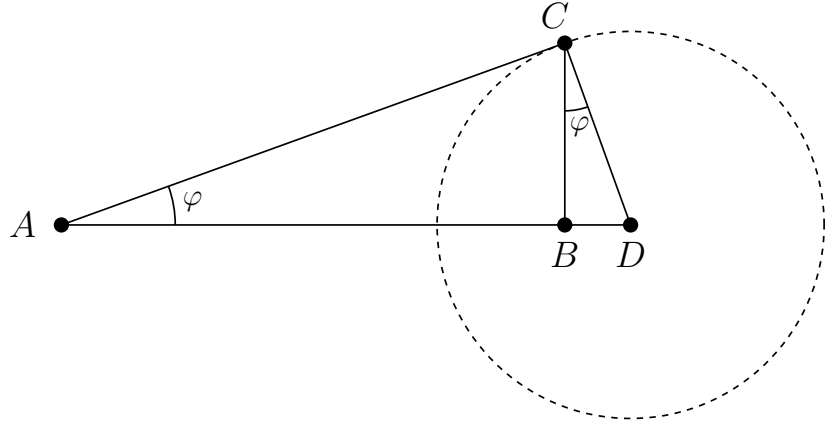
    Let $ABC$ be a triangle, with $AB$ of length $1-M-M\gamma \sin(\varphi)$, $BC$ of length $\gamma M\cos(\varphi)$ and $\angle ABC$ a right angle. Then the lower bound on the overlap is equal to the cosine of the angle $\angle BAC$. Extend the line segment $AB$ to a point $D$ for which the length of $AD$ is $1-M$ and the length of $BD$ is $M\gamma\sin(\varphi)$. Then the length of $DC$ is $M\gamma$, independent of $\varphi$, and the locus of points $C$ as $\varphi$ varies is a circle of radius $M\gamma$ centered at $D$. The angle $\angle BAC$ is maximized when $AC$ is tangent to this circle, in which case $AC$ is perpendicular to $DC$. Noting similar triangles, we see that the maximum of $\angle BAC$ is precisely $\varphi$, and in this case, $\sin(\varphi)$ is $M\gamma/(1-M)$.  This implies
    \begin{align}
        \lvert \braket{\vec{\tilde{x}}}{\vec{x}}\rvert 
        \geq \sqrt{1-\frac{M^2\gamma^2}{(1-M)^2}} = \sqrt{1-\frac{\eta^2}{\cos^2(\theta_t)}}\,.
    \end{align}
    This bound implies Eq.~\eqref{eq:trace_distance_bound} by the relationship $\nrm{\ketbra{\vec{x}}-\ketbra{\tilde{\vec{x}}}}_1 = 2\sqrt{1-|\braket{\vec{x}}{\vec{\tilde{x}}}|^2}$. Finally, to show Eq.~\eqref{eq:output_dist_bound}, we return to the geometric argument. We note that our method guarantees that the overlap $\braket{\vec{\tilde{x}}}{\vec{x}}$ is a positive real number, and thus, the lower bound on the overlap $\braket{\vec{\tilde{x}}}{\vec{x}} = \cos(\varphi)$ implies that the length $\nrm{\ket{\vec{x}} - \ket{\vec{\tilde{x}}}}$ is upper bounded by the angle $\varphi = \arcsin(\eta/\cos(\theta_t))$. 
\end{proof}

\begin{corollary}\label{cor:known_norm_expected_complexity}
    Fix a known constant $\beta \geq 1$ and a value $t \in [1,\kappa]$.  Suppose $\nrm{\vec{x}}$ satisfies the promise $\nrm{\vec{x}} \in [\beta^{-1} t, \beta t]$ (the case where $\nrm{\vec{x}}$ is known exactly corresponds to $\beta=1$). Then, the QLSP can be solved with an expected number $Q$ of combined controlled queries to $U_A$ or $U_A^\dagger$, and $2Q$ combined controlled queries to $U_{\vec{b}}$ or $U_{\vec{b}}^\dagger$, where 
    \begin{equation}
        Q \leq \frac{(\beta^2+1)^2}{\beta^2}\left(\frac{1+\frac{\varepsilon}{2\sqrt{\beta^2+1}}}{1-\frac{\varepsilon}{2\sqrt{\beta^2+1}}}\right)^2\left\lceil\frac{\kappa}{2}\ln\left(\frac{4\sqrt{\beta^2+1}}{\varepsilon}\right)\right\rceil\,.
    \end{equation} 
\end{corollary}
\begin{proof}
As previously, let $\theta_t = \arctan(\nrm{\vec{x}}/t)$, implying $\cos(\theta_t) =t/\sqrt{t^2+\nrm{\vec{x}}^2} \geq 1/\sqrt{\beta^2+1}$ and $\sin(2\theta_t) = 2t\nrm{\vec{x}}/(t^2+\nrm{\vec{x}}^2)\geq 2\beta/(\beta^2+1)$. Choose $\eta = \varepsilon/(2\sqrt{\beta^2+1})$, so that the distance in Eq.~\eqref{eq:trace_distance_bound} is less than $\varepsilon$. Run \autoref{algo:main_algo}. The success probability is lower bounded in \autoref{thm:QLSS_known_norm} by $p_{\rm succ}\geq[2\beta(1-\eta)]^2/[(\beta^2+1)(1+\eta)]^2$. The protocol must be repeated an expected $1/p_{\rm succ}$ times to observe success, and the query complexity of each run of the protocol contributes a factor $2\ell$ to $Q$, where $\ell$ is given in Eq.~\eqref{eq:ell_kappa}. 
\end{proof}

\section{Estimating the norm: detailed analysis}\label{app:estimating_norm}

In \autoref{sec:estimate_norm_binary_search} of the main text, we showed that it was possible to learn a $2$-approximation to the norm with high probability in complexity $O(\kappa \log\log(\kappa)\log\log\log(\kappa))$. The method sketched there requires many repetitions at each step of the binary search, which would contribute unfavorably toward the constant prefactors and the practical performance. We still think the noisy binary search may be a promising practical method, but here we examine an alternative method with $O(\kappa\log(\kappa)\log\log(\kappa))$ complexity that more closely resembles the exhaustive search method from \autoref{sec:estimate_norm_exhaustive_search}. We then show how to reduce it to $O(\kappa\sqrt{\log(\kappa)}\log\log(\kappa))$ with amplitude amplification. In both cases we give explicit expressions for the resource cost. Note that $\sqrt{\log(\kappa)}$ is not much larger than $\ln\ln(\kappa)$ for many practical values of $\kappa$. For example $\ln(\ln(10^6))=2.6$, whereas $\sqrt{\ln(10^6)} = 3.7$.

The idea---implemented in \autoref{algo:random_t_learn_norm}---is to guess random values of $t$, and then run \autoref{algo:main_algo} using that value of $t$. If it succeeds, output $t$, and if it fails, output fail. Choices of $t$ that are closer to $\nrm{\vec{x}}$ will be output more often since they lead to larger values of $p_{\rm succ}$ for \autoref{algo:main_algo} (see \autoref{fig:succ_prob}). When it succeeds, \autoref{algo:random_t_learn_norm} also outputs a state $\ket{\vec{\tilde{x}}}$. This might be regarded as a solution to the QLSP, but its complexity would have a term of the form $O(\kappa\log(\kappa)\log(1/\varepsilon))$. This can be avoided by instead viewing the output of \autoref{algo:random_t_learn_norm} as an ansatz state, and then refining it with KP. We analyze this two-step algorithm later in \autoref{algo:random_t_full_QLSS}.

In \autoref{algo:random_t_learn_norm}, we are given as input an interval $[\mathcal{L},\mathcal{R}]$ that is known to contain $\nrm{\vec{x}}$---equivalently $\ln(\nrm{\vec{x}}) \in [\ln(\mathcal{L}), \ln(\mathcal{R})]$. We choose $\tau$ essentially uniformly at random from the interval $[\ln(\mathcal{L}),\ln(\mathcal{R})]$, but we oversample the edge points $\tau = \ln(\mathcal{L})$ and $\tau = \ln(\mathcal{R})$. The reason for this is related to the fact that when we draw $\tau$ uniformly at random from the interval $[\ln(\mathcal{L}),\ln(\mathcal{R})]$, the average distance between $\tau$ and $\ln(\nrm{\vec{x}})$ is larger when $\nrm{\vec{x}}$ is near a boundary of the interval than when it is not near the boundary. By oversampling the edge points, we compensate for this effect.

\begin{algorithm}
\caption{Estimating the norm by randomly choosing $t$ + postselection}\label{algo:random_t_learn_norm}
\DontPrintSemicolon
\SetAlgoLined
\LinesNumbered
\Input{$(A, \vec{b}, [\mathcal{L},\mathcal{R}], \kappa, \eta)$}
\Output{$(t, \ket{\vec{\tilde{x}}})$ with probability $q_{\rm succ}$, and ``fail'' with probability $1-q_{\rm succ}$}
\BlankLine
Choose $\tau$ uniformly at random from $[\ln(\mathcal{L})-\frac{1}{2},\ln(\mathcal{R})+\frac{1}{2}]$\;\label{line:choose_tau}
\If{$\tau \leq \ln(\mathcal{L})$}{
Set $\tau = \ln(\mathcal{L})$\;
}
\ElseIf{$\tau \geq \ln(\mathcal{R})$}{
Set $\tau = \ln(\mathcal{R})$\;
}
Set $t = e^\tau$\;
Run \autoref{algo:main_algo} with parameters $(A, \vec{b},\kappa, \eta, t)$. If it fails, output ``fail.'' If it succeeds, denote its output by $\ket{\vec{\tilde{x}}}$\;\label{line:run_main_algo}
\Return{$(t , \ket{\vec{\tilde{x}}})$}
\end{algorithm}

\begin{lemma}\label{lem:estimate_norm_random_t}
    Let $\eta > 0$, $\mathcal{L}$, and $\mathcal{R}$ be fixed and known parameters and suppose $\nrm{\vec{x}} \in [\mathcal{L},\mathcal{R}]$. Let $q_{\rm succ}$ denote the success probability of \autoref{algo:random_t_learn_norm} on parameters $(A, \vec{b}, [\mathcal{L},\mathcal{R}], \kappa, \eta)$, and when it succeeds denote its output by $(t,\ket{\vec{\tilde{x}}})$.     The probability of success satisfies
    \begin{equation}\label{eq:q_succ_lowerbound}
        q_{\rm succ} \geq \frac{(1-\eta)^2}{(1+\eta)^2(\ln(\mathcal{R}/\mathcal{L})+1)}
    \end{equation}
    and the query cost of the algorithm is the same as that of \autoref{thm:QLSS_known_norm}.
    Postselecting on success, for any $\beta$ the output $t$ satisfies 
        \begin{align}
            \Pr\left[t \not\in [\beta^{-1}\nrm{\vec{x}}, \beta \nrm{\vec{x}}] \right]
            \leq{} \frac{(1+\eta)^2}{(1-\eta)^2}\frac{4}{\beta^2+1} + \frac{2\eta^2(2\ln(\mathcal{R}/\mathcal{L})+1-4\ln(\beta))}{(1-\eta)^2}\,.\label{eq:prob_deviate_more_than_beta}
        \end{align}
    When it succeeds it also outputs the same $\ket{\vec{\tilde{x}}}$ that appears in \autoref{thm:QLSS_known_norm} for the corresponding value of $t$ and $\eta$. Denoting the ensemble of $\ketbra{\vec{\tilde{x}}}$ by the mixed state $\tilde{\rho}$, we have
    \begin{equation}
        \bra{\vec{x}}\tilde{\rho}\ket{\vec{x}} \geq 1- \frac{2\eta^2}{(1-\eta)^2}\left(\frac{3}{2} + \ln\left(\frac{\mathcal{R}^2+\mathcal{L}^2}{2\mathcal{L}^2}\right) \right)
    \end{equation}
\end{lemma}
\begin{proof}
    Observe that \autoref{algo:random_t_learn_norm} chooses $\tau$ randomly as follows. With probability $1/(2\ln(\mathcal{R}/\mathcal{L})+2)$ it chooses $\tau = \ln(\mathcal{L})$, with probability $1/(2\ln(\mathcal{R}/\mathcal{L})+2)$ it chooses $\tau = \ln(\mathcal{R})$, and with probability $\ln(\mathcal{R}/\mathcal{L})/(\ln(\mathcal{R}/\mathcal{L})+1)$, it chooses $\tau$ uniformly at random from the interval $[\ln(\mathcal{L}),\ln(\mathcal{R})]$. It then runs \autoref{algo:main_algo} (analyzed in \autoref{thm:QLSS_known_norm}) using $t = e^\tau$, and if it succeeds, outputs $t$ and the resulting state $\ket{\vec{\tilde{x}}}$. Since \autoref{algo:main_algo} is called exactly one time, the query cost is identical to that of \autoref{algo:main_algo}. 

    Let $\mathcal{Q}_t$ denote the probability that \autoref{algo:main_algo} succeeds for a certain choice of $t$, which can be lower bounded by Eq.~\eqref{eq:p_succ_bounds}.  The overall probability of success is at least
    \begin{align}
        q_{\rm succ} &= \frac{1}{2\ln(\mathcal{R}/\mathcal{L})+2} \left(\mathcal{Q}_{\mathcal{L}} + \mathcal{Q}_{\mathcal{R}} +2\int_{\ln(\mathcal{L})}^{\ln(\mathcal{R})} d\tau \mathcal{Q}_t\right) \\
        &\geq \frac{(1-\eta)^2}{(1+\eta)^2(2\ln(\mathcal{R}/\mathcal{L})+2)} \left(\sin^2(2\theta_{\mathcal{L}}) + \sin^2(2\theta_{\mathcal{R}}) +2\int_{\ln(\mathcal{L})}^{\ln(\mathcal{R})} d\tau \sin^2(2\theta_t)\right)
    \end{align}
    where $\theta_t = \arctan(\nrm{\vec{x}}/e^\tau)$. Computing $d\theta_t/d\tau = -\sin(2\theta_t)/2$, we make a change of integration variable and express the lower bound as
    \begin{align}
        q_{\rm succ} &\geq \frac{(1-\eta)^2}{(1+\eta)^2(2\ln(\mathcal{R}/\mathcal{L})+2)}\left(\sin^2(2\theta_{\mathcal{L}}) + \sin^2(2\theta_{\mathcal{R}}) +4\int_{\theta_{\mathcal{R}}}^{\theta_{\mathcal{L}}} d\theta_t \sin(2\theta_t) \right)\\
        &=\frac{(1-\eta)^2(\sin^2(2\theta_{\mathcal{L}}) + \sin^2(2\theta_{\mathcal{R}})+2\cos(2\theta_{\mathcal{R}})-2\cos(2\theta_{\mathcal{L}}))}{(1+\eta)^2 (2\ln(\mathcal{R}/\mathcal{L})+2)} \\
        &= \frac{(1-\eta)^2}{(1+\eta)^2(\ln(\mathcal{R}/\mathcal{L})+1)} \nonumber \\
        & \qquad + \frac{(1-\eta)^2}{(1+\eta)^2(2\ln(\mathcal{R}/\mathcal{L})+2)}\left[(2\cos(2\theta_{\mathcal{R}})-\cos^2(2\theta_{\mathcal{R}})) + (- 2\cos(2\theta_{\mathcal{L}}) - \cos^2(2\theta_{\mathcal{L}})\right] \\
        & \geq \frac{(1-\eta)^2}{(1+\eta)^2(\ln(\mathcal{R}/\mathcal{L})+1)} \label{eq:q_succ_ab_final}
    \end{align}
    The last line follows by noting the following two facts. First, since $\mathcal{L} \leq \nrm{\vec{x}}$, we have $-1 \leq \cos(2\theta_{\mathcal{L}}) \leq 0$ and hence $-2\cos(2\theta_{\mathcal{L}}) - \cos^2(2\theta_{\mathcal{L}}) \geq 0$. Second, since $\mathcal{R} \geq \nrm{\vec{x}}$, we have $0 \leq \cos(2\theta_{\mathcal{R}}) \leq 1$ and hence $2\cos(2\theta_{\mathcal{R}}) - \cos^2(2\theta_{\mathcal{R}}) \geq 0$. This confirms Eq.~\eqref{eq:q_succ_lowerbound}.

   We can now compute an upper bound on the probability that the algorithm succeeds and outputs an estimate for $t$ greater than $\beta \nrm{\vec{x}}$, which we denote by $P_{>}$. For this to be the case, first of all, the random choice of $\tau$ must be greater than $\beta\nrm{\vec{x}}$. Then, conditoned on such a choice of $\tau$, the probability of success $\mathcal{Q}_t$ is upper bounded by Eq.~\eqref{eq:p_succ_bounds}. If $\nrm{\vec{x}} < \mathcal{R}/\beta$, we have
   \begin{align}
        P_{>} &= \frac{1}{2\ln(\mathcal{R}/\mathcal{L})+2}\left(\mathcal{Q}_{\mathcal{R}} + 2\int_{\ln(\beta \nrm{\vec{x}})}^{\ln(\mathcal{R})} d\tau \mathcal{Q}_t \right)\\
        &\leq \frac{1}{2\ln(\mathcal{R}/\mathcal{L})+2}\left( 2\int_{\ln(\beta \nrm{\vec{x}})}^{\ln(\mathcal{R})} d\tau \left(\sin^2(2\theta_t)+\frac{4\eta^2\sin^2(\theta_t)}{(1+\eta)^2}\right) + \sin^2(2\theta_{\mathcal{R}}) +\frac{4\eta^2\sin^2(\theta_{\mathcal{R}})}{(1+\eta)^2}\right)
    \end{align}
    and if $\nrm{\vec{x}} \geq b/\beta$, we have $P_{>} = 0$.    Following the method that arrived at Eq.~\eqref{eq:q_succ_ab_final}, we compute
    \begin{equation}
        2\int_{\ln(\beta \nrm{\vec{x}})}^{\ln(\mathcal{R})} d\tau \sin^2(2\theta_t) = 2\cos(2\theta_{\mathcal{R}}) - 2\cos(2\theta_{\beta\nrm{\vec{x}}}) \,.
    \end{equation}
    and we note that $2\cos(2\theta_{\mathcal{R}}) + \sin^2(2\theta_{\mathcal{R}}) = 2-4\sin^4(\theta_{\mathcal{R}}) \leq 2$. We also note that $2-2\cos(\theta_{\beta\nrm{\vec{x}}}) = 4\sin^2(\theta_{\beta\nrm{\vec{x}}}) = \frac{4}{\beta^2+1}$. 
    \begin{align}
        P_{>} &\leq \frac{1}{2\ln(\mathcal{R}/\mathcal{L})+2}
        \left(\frac{4}{\beta^2+1} + \frac{4\eta^2}{(1+\eta)^2}\left(
        \sin^2(\theta_{\mathcal{R}}) + 2\int_{\ln(\beta\nrm{\vec{x}})}^{\ln(\mathcal{R})} d \tau \sin^2(\theta_t)
        \right)
        \right)
    \end{align}
    Similarly, we can compute the probability that the algorithm succeeds and outputs a value of $t < \nrm{\vec{x}}/\beta$. If $\nrm{\vec{x}} \leq \beta \mathcal{L}$, then $ P_< = 0$. If $\nrm{\vec{x}} > \beta \mathcal{L}$, we follow the same process as above, and we arrive at
    \begin{align}
        P_{<} &\leq \frac{1}{2\ln(\mathcal{R}/\mathcal{L})+2}
        \left(\frac{4}{\beta^2+1} + \frac{4\eta^2}{(1+\eta)^2}\left(
        \sin^2(\theta_{\mathcal{L}}) + 2\int_{\ln(\mathcal{L})}^{\ln(\nrm{\vec{x}}/\beta)} d \tau \sin^2(\theta_t)
        \right)
        \right)
    \end{align}
    We will add these two expressions to get an upper bound on the overall probability that the output deviates by a factor $\beta$. To bound the $O(\eta^2)$ terms under the integral, we simply assert that $\sin^2(\theta_t) \leq 1$. First, suppose that $\mathcal{L}\beta \leq \nrm{\vec{x}} \leq \mathcal{R}/\beta$. This implies that both $P_<$ and $P_>$ are nonzero, and also that $\sin^2(\theta_{\mathcal{R}}) \leq 1/(\beta^2+1)$ and $\sin^2(\theta_{\mathcal{L}}) \leq \beta^2/(\beta^2+1)$, so $\sin^2(\theta_{\mathcal{L}}) + \sin^2(\theta_{\mathcal{R}}) \leq 1$. We find that
    \begin{equation}
        P_> + P_< \leq \frac{1}{\ln(\mathcal{R}/\mathcal{L})+1}\left(\frac{4}{\beta^2+1} + \frac{2\eta^2}{(1+\eta)^2}\left(1+2\ln(\mathcal{R}/\mathcal{L})-4\ln(\beta))\right)\right)
    \end{equation}
    If $\mathcal{L}\beta \leq \nrm{\vec{x}} \leq \mathcal{R}/\beta$ does not hold, then this implies that at least one of $P_<$ and $P_>$ is zero and it is easy to see that the above expression still holds. 
    
    Now, we divide $P_< + P_>$ by the lower bound on $q_{\rm succ}$ to yield an upper bound on the probability of deviating by $\beta$, conditioned on the algorithm succeeding. The result is Eq.~\eqref{eq:prob_deviate_more_than_beta}.

    Finally, we show the lower bound on $\bra{\vec{x}}\tilde{\rho}\ket{\vec{x}}$. 
    Let the symbol $\EV$ denote expectation over the ensemble of $(t,\ket{\tilde{\vec{x}}})$. We can equivalently upper bound $\EV[1-|\braket{\vec{x}}{\vec{\tilde{x}}}|^2]$. For any particular value of $\tau$, denote this quantity by $\mu_t^2$, which we can compute from Eq.~\eqref{eq:post_step3_KR}:
    \begin{equation}
        \mu_t^2 = 1-|\braket{\vec{x}}{\vec{\tilde{x}}}|^2 = \frac{\delta_2^2\sin^2(\theta_t)}{\mathcal{Q}_t} \leq \frac{4\eta^2}{(1+\eta)^2\mathcal{Q}_t}\sin^2(\theta_t)
    \end{equation}
    The probability a certain $\tau$ is observed is proportional to the probability $\tau$ is randomly chosen, times the probability $\mathcal{Q}_t$ that \autoref{algo:main_algo} succeeds. Thus, we have
    \begin{align}
        1-\bra{\vec{x}}\tilde{\rho}\ket{\vec{x}} &= \EV[1-|\braket{\vec{x}}{\vec{\tilde{x}}}|^2] = \frac{1}{q_{\rm succ}(2\ln(\mathcal{R}/\mathcal{L})+2)}\left(\mathcal{Q}_{\mathcal{L}}\mu_{\mathcal{L}}^2 + \mathcal{Q}_{\mathcal{R}}\mu_{\mathcal{R}}^2 + 2\int_{\ln(\mathcal{L})}^{\ln(\mathcal{R})} d\tau \mathcal{Q}_t \mu_t^2 \right) \\
        &\leq  \frac{4\eta^2}{(1+\eta)^2 q_{\rm succ}(2\ln(\mathcal{R}/\mathcal{L})+2)}\left(\sin^2(\theta_{\mathcal{L}}) + \sin^2(\theta_{\mathcal{R}}) + 2\int_{\ln(\mathcal{L})}^{\ln(\mathcal{R})} d\tau \sin^2(\theta_t) \right) 
    \end{align}
    We can use the same substitution as above to say
    \begin{align}
        2\int_{\ln(\mathcal{L})}^{\ln(\mathcal{R})} d\tau \sin^2(\theta_t) = 2\int_{\theta_{\mathcal{R}}}^{\theta_{\mathcal{L}}} d\theta_t \tan(\theta_t) = 2\ln\left(\frac{\cos(\theta_{\mathcal{R}})}{\cos(\theta_{\mathcal{L}})}\right) = \ln\left(\frac{\mathcal{R}^2(\nrm{\vec{x}}^2+\mathcal{L}^2)}{\mathcal{L}^2(\nrm{\vec{x}}^2+\mathcal{R}^2)}\right) \leq \ln\left(\frac{\mathcal{R}^2+\mathcal{L}^2}{2\mathcal{L}^2}\right)\label{eq:evaluation_int_sin^2(theta_t)}
    \end{align}
    We also have $\sin^2(\theta_{\mathcal{L}}) \leq 1$ and $\sin^2(\theta_{\mathcal{R}}) \leq 1/2$, a fact that follows from the assumption $b \geq \nrm{\vec{x}}$. We can also impose the upper bound on $q_{\rm succ}$. This gives
  \begin{align}
        1-\bra{\vec{x}}\tilde{\rho}\ket{\vec{x}} 
        &\leq  \frac{4\eta^2}{(1+\eta)^2q_{\rm succ}(2\ln(\mathcal{R}/\mathcal{L})+2)}\left(\frac{3}{2} + \ln\left(\frac{\mathcal{R}^2+\mathcal{L}^2}{2\mathcal{L}^2}\right) \right) \\
        &\leq \frac{2\eta^2}{(1-\eta)^2}\left(\frac{3}{2} + \ln\left(\frac{\mathcal{R}^2+\mathcal{L}^2}{2\mathcal{L}^2}\right) \right)
    \end{align}
    
\end{proof}

Briefly, we now discuss how this theorem enables us to learn the norm up to a constant approximation ratio. Suppose $\beta$ is a fixed constant greater than 3. We may take $\eta \geq 0.1(2 \ln(\mathcal{R}/\mathcal{L})+1-4\ln(\beta))^{-1/2} \geq \Omega(\kappa^{-1/2})$ so that the right-hand-side of Eq.~\eqref{eq:prob_deviate_more_than_beta} is upper bounded by 0.45. Thus, there is a strictly greater than 1/2 chance the output of \autoref{algo:random_t_learn_norm} is a 3-approximation of $\nrm{\vec{x}}$. The expected number of times we need to run \autoref{algo:random_t_learn_norm} to observe success is upper bounded by $q_{\rm succ}^{-1} \approx \ln(\mathcal{R}/\mathcal{L})+1 \leq \ln(\kappa) + 1$. Thus, the overall complexity to learn a 3-approximation to the norm is $O(\kappa\log(\kappa)\log\log(\kappa))$. The probability the output is a 3-approximation can be boosted to $1-\delta$ with median amplification at the cost of multiplicative $O(\log(1/\delta))$ overhead. The approximation ratio can be improved from 3 to $1+\varepsilon$ using the method from \autoref{sec:improving_approx_ratio}.

Now we describe how the $O(\log(\kappa))$ complexity coming from the required number of repetitions can be improved to $O(\sqrt{\log(\kappa)})$ with fixed-point amplitude amplification. 

\begin{lemma}\label{lem:estimate_norm_random_t_FPAA}
    Let $\eta>0$, $\delta > 0$, $\mathcal{L}$, and $\mathcal{R}$ be fixed constants and suppose $\nrm{\vec{x}} \in [\mathcal{L},\mathcal{R}]$, and let $d$ be a positive integer. Suppose we restrict values of $t$ such that $\ln(t)-\mathcal{L}$ an integer multiples of $2^{-d}$, and we run \autoref{algo:random_t_learn_norm} coherently by creating a superposition over a set of discrete $t$ values, rather than a classical random guess. Then, we may wrap \autoref{algo:random_t_learn_norm} in fixed-point amplitude amplification, and boost its probability of success to $1-\delta$. The output $(t,\vec{\tilde{x}})$ of this protocol is identical to that of \autoref{algo:random_t_learn_norm} up to discretization error of order $2^{-d}$. Specifically
        \begin{align}
            \Pr\left[t \not\in [\beta^{-1}\nrm{\vec{x}}, \beta \nrm{\vec{x}}] \right]
            \leq{}& \frac{(1+\eta)^2}{(1-\eta)^2}\frac{4}{\beta^2+1} +\ln(\mathcal{R}/\mathcal{L})^2 2^{-d-1} +  \frac{2\eta^2(2\ln(\mathcal{R}/\mathcal{L})+1-4\ln(\beta))}{(1-\eta)^2}\\
            \bra{\vec{x}}\tilde{\rho}\ket{\vec{x}} \geq{}& 1- \frac{2\eta^2}{(1-\eta)^2}\left(\frac{3}{2} + \ln\left(\frac{\mathcal{R}^2+\mathcal{L}^2}{2\mathcal{L}^2}\right) + \ln(\mathcal{R}/\mathcal{L})^22^{-d} \right)\left(1-\ln(\mathcal{R}/\mathcal{L})^22^{-d-1}\right)^{-1} \,.
        \end{align}
    The query cost of the algorithm is $Q$ total queries to $U_A$, $U_A^{\dagger}$ and their controlled versions, and $2Q$ total queries to $U_{\vec{b}}$, $U_{\vec{b}}^\dagger$, and their controlled versions, where 
    \begin{equation}
        Q \leq \left(2\left\lceil \frac{(1+\eta)\sqrt{\ln(\mathcal{R}/\mathcal{L})+1}}{2(1-\eta)} \left(1-\ln(\mathcal{R}/\mathcal{L})^22^{-d-1}\right)^{-1/2} \ln(\frac{2}{\sqrt{\delta}})-\frac{1}{2} \right\rceil +1\right)2\left\lceil \frac{\kappa}{2}\ln(\frac{2}{\eta}) \right\rceil 
    \end{equation}
    The algorithm requires $d+\lceil \log_2\ln(\mathcal{R}/\mathcal{L})\rceil $ additional ancilla qubits, compared to \autoref{algo:random_t_learn_norm}. 
\end{lemma}
\begin{proof}
    First, we describe in more detail the modifications that allow fixed-point amplitude amplification to be applied. Let $q = d+\lceil \log_2 \ln(\mathcal{R}/\mathcal{L}) \rceil +1 $ be the number of ancilla qubits. Let $j_{\max} = \lceil \ln(\mathcal{R}/\mathcal{L})2^d\rceil  \leq 2^{q-1}$.  We require a unitary $\mathcal{P}$ that prepares the state
    \begin{equation}
        \mathcal{P}\ket{0}\ket{0^{q}} = \frac{1}{\sqrt{2 \ln(\mathcal{R}/\mathcal{L})+2}}\left(\ket{0}\ket{0^q}+\ket{0}\ket{10^{q-1}} + \frac{\sqrt{2\ln(\mathcal{R}/\mathcal{L})}}{\sqrt{j_{\max}}}\sum_{j=0}^{j_{\max}-1}\ket{1}\ket{j} \right)
    \end{equation}
    We require another unitary $\mathcal{B}$ which applies the block-encoding $U_{A_t}$ for matrix $A_t$, where the value of $t$ is controlled by the setting of the $q+1$ qubits depicted in the equation above. Referencing \autoref{fig:block_encoding_At}, to control the value of $t$, we merely need to control the value of the rotation angle in the block-encoding $U_{A_t}$. Let $t_j = e^{a+j2^{-d}}$, so that $\ln(t_j)$ is the left edgepoint of the interval $[a+j2^{-d}, a+(j+1)2^{-d}]$. Then, we require that
    \begin{equation}
        \mathcal{B} = \ketbra{0^{q+1}} \otimes U_{A_\mathcal{L}} + \ketbra{010^{q-1}} \otimes U_{A_{\mathcal{R}}} + \sum_{j=0}^{2^d-1} \ketbra{1j} \otimes U_{A_{t_j}}
    \end{equation}
    Then, we may run \autoref{algo:random_t_learn_norm} by using $\mathcal{P}$ to generate the random distribution of values of $t$, and replace each occurrence of $U_{A_t}$ by $\mathcal{B}$. The analysis of the algorithm is the same in \autoref{lem:estimate_norm_random_t}, except that the integrals are approximated by discrete sums. To bound the error in this approximation, we invoke the well-known formula that the difference between an integral on the interval $[L,R]$ and a sum of the integrand evaluated at the left edge point of $N$ equally sized subintervals is at most $K_1(R-L)^2/(2N)$, where $K_1$ is an upper bound on the derivative of the integrand on the interval. We note the (non-optimized) upper bounds
    \begin{align}
        \left\lvert \frac{d}{d\tau}(\sin^2(2\theta_t))\right\rvert  = \left\lvert - \sin(2\theta_t)\sin(4\theta_t)/2\right\rvert &\leq 1/2 \\
        \left\lvert \frac{d}{d\tau}(\sin^2(\theta_t))\right\rvert  = \left\lvert - \sin^2(2\theta_t)\right\rvert &\leq 1 \\
    \end{align}
    The approximation errors satisfy
    \begin{align}
        \left\lvert \int_{\ln(\mathcal{L})}^{\ln(\mathcal{R})} d\tau \sin^2(2\theta_t) - 2^{-d}\sum_{j=0}^{j_{\max-1}} \sin^2(2\theta_{t_j})\right\rvert &\leq \ln(\mathcal{R}/\mathcal{L})^2 2^{-d-2} \\
        \left\lvert \int_{\ln(\mathcal{L})}^{\ln(\mathcal{R})} d\tau \sin^2(\theta_t) - 2^{-d}\sum_{j=0}^{j_{\max-1}} \sin^2(\theta_{t_j})\right\rvert &\leq \ln(\mathcal{R}/\mathcal{L})^2 2^{-d-1}
    \end{align}
    These integrals are used in \autoref{lem:estimate_norm_random_t} in Eqs.~\eqref{eq:q_succ_ab_final} and \eqref{eq:evaluation_int_sin^2(theta_t)}. Thus, we instead obtain
    \begin{align}
        q_{\rm succ} &\geq  \frac{(1-\eta)^2}{(1+\eta)^2(\ln(\mathcal{R}/\mathcal{L})+1)} \left(1-\ln(\mathcal{R}/\mathcal{L})^22^{-d-1}\right) \label{eq:q_succ_discrete_lowerbound}\\
        1-\bra{\vec{x}}\tilde{\rho}\ket{\vec{x}} 
        &\leq \frac{2\eta^2}{(1-\eta)^2}\left(\frac{3}{2} + \ln\left(\frac{\mathcal{R}^2+\mathcal{L}^2}{2\mathcal{L}^2}\right) + \ln(\mathcal{R}/\mathcal{L})^22^{-d} \right)\left(1-\ln(\mathcal{R}/\mathcal{L})^22^{-d-1}\right)^{-1} 
    \end{align}
    Through a similar method, we also obtain
        \begin{equation}
        P_> + P_< \leq \frac{1}{\ln(\mathcal{R}/\mathcal{L})+1}\left(\frac{4}{\beta^2+1} + \ln(\mathcal{R}/\mathcal{L})^2 2^{-d-1} + \frac{2\eta^2}{(1+\eta)^2}\left(1+2\ln(\mathcal{R}/\mathcal{L})-4\ln(\beta))\right)\right)\,.
    \end{equation}
    Using this coherent version of \autoref{algo:random_t_full_QLSS}, we can wrap the algorithm in fixed-point amplitude amplification \cite{yoder2014FixedPointSearch}. Given a lower bound $\lambda$ on the success probability (as in Eq.~\eqref{eq:q_succ_discrete_lowerbound}), fixed-point amplitude amplification allows the success probability to be boosted to at least $1-\delta$ with using $L$  calls to the algorithm, where $L$ is the smallest odd integer greater than $\lambda^{-1/2} \ln(2/\sqrt{\delta})$, and an equal number of calls to the oracle that determines success. Here, determining success is simply checking whether the ancilla qubits are $\ket{0}$. Thus, the overall query complexity is the number $L$ times the query complexity of \autoref{algo:main_algo}, yielding the quoted statement. 
\end{proof}

\section{Explicit complexity bounds for  QLSS}

\subsection{Near-optimal QLSSs} \label{app:explicit_near_optimal_QLSS}

Now, we present a full QLSS that does not require knowing an estimate for the norm ahead of time. 
All it needs to know is that $\nrm{\vec{x}} \in [\mathcal{L},\mathcal{R}]$ (we can always take $[\mathcal{L},\mathcal{R}] = [1,\kappa]$ if we have no additional information). In this subsection, we do not yet implement the adiabatic-inspired ideas described in \autoref{sec:estimate_by_adiabatic} to achieve linear-in-$\kappa$ complexity; we so so in \autoref{app:explicit_optimal_QLSS}. The algorithm presented here is a randomized algorithm that outputs an ensemble of pairs $(t,\ket{\vec{\tilde{x}}})$, where $t$ is an estimate for $\nrm{\vec{x}}$ and $\ket{\vec{\tilde{x}}}$ is the associated output of \autoref{algo:main_algo} for that choice of $t$. It works by running \autoref{algo:random_t_learn_norm} to generate a good norm estimate $t$ and associated anstatz state $\ket{\vec{x}_{\rm ans}}$, and then using KP to refine that state. The KP step should be run with precision 
\begin{equation}\label{eq:eta_prime}
    \eta_{\rm KP} = \frac{\varepsilon}{\mu}\frac{\sqrt{1-\mu^2}}{\sqrt{1-\varepsilon^2}}
\end{equation}
where
\begin{equation}
    \mu = \frac{\eta}{1-\eta}\sqrt{3 + 2\ln\left(\frac{\mathcal{R}^2+\mathcal{L}^2}{2\mathcal{L}^2}\right)}
\end{equation}
We show that the ensemble $\tilde{\rho}$ of $\ketbra{\vec{\tilde{x}}}$ that the algorithm outputs is $\varepsilon$-close to $\ketbra{\vec{x}}$.  

\begin{algorithm}
\caption{QLSS using random choice of $t$ + postselection + KP}\label{algo:random_t_full_QLSS}
\DontPrintSemicolon
\SetAlgoLined
\LinesNumbered
\Input{$(A, \vec{b}, [\mathcal{L},\mathcal{R}], \kappa, \eta, \varepsilon)$}
\Output{$(t, \ket{\vec{\tilde{x}}})$}
\BlankLine
Let $\eta_{\rm KP}$ be given by Eq.~\eqref{eq:eta_prime}. \;
Run \autoref{algo:random_t_learn_norm} with parameters $(A,\vec{b},[\mathcal{L},\mathcal{R}],\kappa,\eta)$. If it fails, repeat this step. If it succeeds, denote its output by $(t,\ket{\vec{x_{\rm ans}}})$\; \label{line:find_random_t_log}
Apply KP to $\ket{\vec{x_{\rm ans}}}$ using the matrix $G=Q_{\vec{b}} A$ with parameters $(\kappa,\eta_{\rm KP})$. If it fails, go to \autoref{line:find_random_t_log}. If it succeeds, denote its output by $\ket{\vec{\tilde{x}}}$.\;\label{line:run_KP_log}
\Return{$(t , \ket{\vec{\tilde{x}}})$}
\end{algorithm}

\begin{theorem}\label{thm:QLSS_random_t}
Suppose that $\nrm{\vec{b}} = 1$, that $\vec{b}$ is in the column space of $A$, and that all nonzero singular values of $A$ lie in the interval $[\kappa^{-1},1]$. Let $U_A$ be a $(1,a)$-block-encoding of $A$, and $U_{\vec{b}}$ be a state-preparation unitary for $\ket{\vec{b}}$.  Let $\vec{x}$ denote the unique vector of minimum norm $\nrm{\vec{x}}$ for which $A\vec{x} = \vec{b}$. Fix $\eta > 0$ and $\varepsilon > 0$, and assume $\eta$ satisfies $\eta/(1-\eta) < (3+2\ln((\mathcal{R}^2+\mathcal{L}^2)/(2\mathcal{L})^2))^{-1/2}$. Let $\eta_{\rm KP} > 0$ be given by Eq.~\eqref{eq:eta_prime}.  Consider \autoref{algo:random_t_full_QLSS}, denoting its output by $(t,\ket{\vec{\tilde{x}}})$. 

Then, the ensemble of outputs $\ketbra{\vec{\tilde{x}}}$, denoted by $\tilde{\rho}$, satisfies $\frac{1}{2} \nrm{\ketbra{\vec{x}}-\tilde{\rho}}_1 \leq \varepsilon$, meaning the output solves the QLSP. Furthermore, the expected number of queries required is $Q$ total queries to $U_A$, $U_A^\dagger$ and their controlled versions, as well as $2Q$ total queries to $U_{\vec{b}}$, $U_{\vec{b}}^\dagger$, and their controlled versions, where
\begin{align}\label{eq:query_complexity_exhaustive_QLSS}
    Q \leq \frac{2(1+\eta)^2(\ln(\mathcal{R}/\mathcal{L})+1)}{(1-\eta)^2}\left\lceil \frac{\kappa}{2}\ln(\frac{2}{\eta}) \right\rceil + \frac{1}{1-\frac{\eta^2}{(1-\eta)^2}\left(3 + 2\ln\left(\frac{\mathcal{R}^2+\mathcal{L}^2}{2\mathcal{L}^2}\right)\right)}2\left\lceil \frac{\kappa}{2}\ln(\frac{2}{\eta_{\rm KP}}) \right\rceil
\end{align}
\end{theorem}

\begin{proof}
First, we verify that the output solves the QLSP. From \autoref{lem:estimate_norm_random_t}, the ensemble of outputs of \autoref{algo:random_t_learn_norm}, denoted by $\tilde{\rho}_{\rm ans}$, satisfies
    \begin{equation}
        \bra{\vec{x}}\tilde{\rho}_{\rm ans}\ket{\vec{x}} \geq 1- \frac{2\eta^2}{(1-\eta)^2}\left(\frac{3}{2} + \ln\left(\frac{\mathcal{R}^2+\mathcal{L}^2}{2\mathcal{L}^2}\right) \right)
    \end{equation}
The subsequent KP step can then be analyzed with \autoref{lem:KP}, with 
\begin{equation}
    \mu = \frac{\eta}{1-\eta}\sqrt{3 + 2\ln\left(\frac{\mathcal{R}^2+\mathcal{L}^2}{2\mathcal{L}^2}\right)}
\end{equation}
It asserts that the output of the KP step satisfies
\begin{equation}
    \frac{1}{2}\nrm{\ketbra{\vec{x}}-\tilde{\rho}}_1 \leq \frac{\mu\eta_{\rm KP}}{\sqrt{1-\mu^2+\mu^2\eta_{\rm KP}^2}} = \varepsilon\,.
\end{equation}
where the last equality is ensured by the choice of $\eta_{\rm KP}$ in Eq.~\eqref{eq:eta_prime}

Next we verify the expected query complexity. We examine the expected complexity from each of the two lines separately. For \autoref{line:find_random_t_log}, each time it is called, it has the complexity quoted in \autoref{thm:QLSS_known_norm}. The expected number of times it is called is $r_{\rm succ}^{-1}$, where $r_{\rm succ}$ is the probability that both lines \autoref{line:find_random_t_log} and \autoref{line:run_KP_log} succeed. We can see that this satisfies the same lower bound as $q_{\rm succ}$ from Eq.~\eqref{eq:q_succ_ab_final}, since this lower bound comes entirely from the portion of the state that is left invariant by the KP step. 
\begin{equation}
    r_{\rm succ} \geq \frac{(1-\eta)^2}{(1+\eta)^2(\ln(\mathcal{R}/\mathcal{L})+1)}\,. 
\end{equation}
Thus, the contribution to the expected query complexity from \autoref{line:find_random_t_log} is at most 
\begin{equation}
    \frac{(1+\eta)^2(\ln(\mathcal{R}/\mathcal{L})+1)}{(1-\eta)^2}2\left\lceil \frac{\kappa}{2}\ln(\frac{2}{\eta}) \right\rceil
\end{equation}

On the other hand, \autoref{line:run_KP_log} only executes when \autoref{line:find_random_t_log} succeeds, and the algorithm terminates immediately once \autoref{line:run_KP_log} succeeds for the first time. By \autoref{lem:KP}, the success probability of \autoref{line:run_KP_log} is at least $1-\mu^2$. Thus, the expected number of times \autoref{line:run_KP_log} executes is at most $(1-\mu^2)^{-1}$, and its contribution to the query complexity is upper bounded by 
\begin{equation}
    \frac{2\left\lceil \frac{\kappa}{2}\ln(\frac{2}{\eta_{\rm KP}}) \right\rceil}{1-\frac{\eta^2}{(1-\eta)^2}\left(3 + 2\ln\left(\frac{\mathcal{R}^2+\mathcal{L}^2}{2\mathcal{L}^2}\right)\right)}\,.
\end{equation}
This completes the proof. 
\end{proof}

If one takes $[\mathcal{L},\mathcal{R}] = [1,\kappa]$ and $\eta = \Theta(1/\log(\kappa))$, then $\mu = \Theta(1)$ and $\eta_{\rm KP} = \Theta(\varepsilon)$. The overall query complexity is $O(\kappa \log(\kappa)\log\log(\kappa) + \kappa \log(1/\varepsilon))$. Furthermore, we may perform a closer inspection to derive a simpler formula that holds over a wide range of values of $\kappa$. Suppose that $\kappa \in [3,10^6]$, so that $\sqrt{3+2\ln((\kappa^2+1)/2)} \in [2.49,7.55]$. Choose $\eta$ so that $\mu = 0.25$. This implies 
\begin{align}
    \left(\frac{1+\eta}{1-\eta}\right)^2 = \left(1+\frac{2\eta}{1-\eta}\right)^2 = \left(1+0.5/\sqrt{3+2\ln((\kappa^2+1)/2)}\right)^2 \leq 1.442\,.
\end{align} 
Furthermore 
\begin{align}
    \ln(2/\eta) = \ln(2/\mu) + \ln(\sqrt{3+2\ln((\kappa^2+1)/2)})+ \ln(1+\mu/\sqrt{3+2\ln((\kappa^2+1)/2)}) \leq  4.14
\end{align}
Then, working from Eq.~\eqref{eq:query_complexity_exhaustive_QLSS}, we have
\begin{align}
    Q &\leq 2\cdot 1.442(\ln(\kappa)+1)\lceil (4.14)\frac{\kappa}{2}\rceil + \frac{16}{15}\cdot 2 \cdot \lceil\frac{\kappa}{2}\ln(\varepsilon^{-1}\sqrt{1-\varepsilon^2})-\frac{\kappa}{2}\ln(\sqrt{15}/2) \rceil 
\end{align}
Noting that $\lceil x \rceil \leq x+1$, we can then write
\begin{align}\label{eq:Q_exhaustive_search_practical}
    Q &\leq 5.97 \ln(\kappa)\kappa +5.27\kappa + 1.07 \kappa \ln(\varepsilon^{-1}\sqrt{1-\varepsilon^2}) + 2.89\ln(\kappa) + 5.02\,.
\end{align}
This will be a loose bound compared to Eq.~\eqref{eq:query_complexity_exhaustive_QLSS}, especially for smaller $\kappa$, but it is a genuine upper bound that holds whenever $\kappa \in [3,10^6]$.

We can reduce the prefactor on the first term by a quantity $\sqrt{\log(\kappa)}$ by using (fixed-point) amplitude amplification to boost the chances of successfully preparing $\rho_{\rm ans}$, prior to performing KP. 

\begin{theorem}\label{thm:QLSS_random_t_FPAA}
    Consider \autoref{algo:random_t_full_QLSS}, with  \autoref{line:find_random_t_log} replaced by the fixed-point amplitude amplified version described in \autoref{lem:estimate_norm_random_t_FPAA} for a certain amplification parameter $\delta$, and discretization parameter $d$. Then one can achieve the same results with
    \begin{align}
        Q ={}& \frac{\frac{1}{1-\delta}\left(2\left\lceil \frac{(1+\eta)\sqrt{\ln(\mathcal{R}/\mathcal{L})+1}}{2(1-\eta)} \left(1-\ln(\mathcal{R}/\mathcal{L})^22^{-d-1}\right)^{-1/2} \ln(\frac{2}{\sqrt{\delta}})-\frac{1}{2} \right\rceil +1\right)2\left\lceil \frac{\kappa}{2}\ln(\frac{2}{\eta}) \right\rceil+2\left\lceil \frac{\kappa}{2}\ln(\frac{2}{\eta_{\rm KP}}) \right\rceil}{1-\frac{\eta^2}{(1-\eta)^2}\left(3 + 2\ln\left(\frac{\mathcal{R}^2+\mathcal{L}^2}{2\mathcal{L}^2}\right) + \ln(\mathcal{R}/\mathcal{L})^22^{-d+1} \right)\left(1-\ln(\mathcal{R}/\mathcal{L})^22^{-d-1}\right)^{-1} } 
\end{align}
\end{theorem}
\begin{proof}
    This follows by following the same analysis as in \autoref{thm:QLSS_known_norm}, except using \autoref{lem:estimate_norm_random_t_FPAA} in place of \autoref{lem:estimate_norm_random_t}. The probability that both \autoref{line:find_random_t_log} and \autoref{line:run_KP_log} succeed is at least $(1-\delta)(1-\mu^2)$, so the expected number of times the \autoref{line:find_random_t_log} must run is at most the inverse of this quantity. As before \autoref{line:run_KP_log} must be repeated an expected number $(1-\mu^2)^{-1}$ times, leading to the quoted expression. 
\end{proof}

We now turn the complexity statement of \autoref{thm:QLSS_random_t_FPAA} into a looser but simpler upper bound. First of all, take $d \rightarrow \infty$, as the bound on the query complexity is monotonically decreasing with $d$. Let $[\mathcal{L},\mathcal{R}] = [1,\kappa]$. Choose $\delta = 1/4$ and 
note $\ln(2/\sqrt{\delta})/(1-\delta) \leq 1.85$. As previously, choose $\eta$ so that $\mu = 0.25$, and note $\ln(2/\eta) \leq 4.14$ when $\kappa \in [3, 10^6]$. We have
\begin{align}\label{eq:Q_FPAA_practical}
    Q &\leq\frac{\left(1.85 \cdot \sqrt{1.45} \cdot \sqrt{\ln(\kappa)+1} + \frac{8}{3} \right)\left(4.14 \kappa + 2\right) + \kappa \ln(\varepsilon^{-1}\sqrt{1-\varepsilon^2})-\kappa\ln(\sqrt{15}/2) +2}{15/16} \\
    &\leq 9.84 \sqrt{\log(\kappa)+1}\kappa + 11.1 \kappa + 1.07 \kappa \ln(\varepsilon^{-1}\sqrt{1-\varepsilon^2}) + 4.76 \sqrt{\ln(\kappa)+1} + 7.83
\end{align}

\subsection{Optimal QLSS}\label{app:explicit_optimal_QLSS}

In this subsection, we provide a rigorous and explicit version of the full QLSS with optimal dependence on $\kappa$, following a similar idea as that presented in \autoref{sec:estimate_by_adiabatic}. The algorithm leverages the family of matrices $\bar{A}_\sigma$ defined in Eq.~\eqref{eq:barA_sigma} to estimate $\nrm{\vec{x}}$. The algorithm is described in \autoref{algo:full_QLSS_constant_prefactors}, and we give explicit bounds on the expected query complexity in \autoref{thm:optimal_QLSS}. 

The algorithm takes in the QLSP parameters $A, \vec{b}, \kappa, \varepsilon$. It also chooses a few additional free parameters, $\hat{c}, \hat{q}, \hat{\beta},\hat{r}, \hat{\chi}, \hat{\Delta}$; we denote these free parameters with hats for easy reference. These are independent of $\kappa$ and $\varepsilon$, and they must satisfy $\hat{c} > 1$, $\hat{q} > 1$,  $\hat{\beta} > 1$, $0 < \hat{\chi} < 1 $, and $0 < \hat{\Delta} <1$. We give an expression for the query complexity that depends on these parameters. A numerical optimization of this expression suggests a locally optimal choice for the parameters that we report in \autoref{tab:free_params}. 
\begin{table}[h!]
    \centering
    \begin{tabular}{|c|c|c|}
        \hline
        \textbf{param} & \textbf{value} & \textbf{role} \\
        \hline
        $\hat{\beta}$ & 15.4 & target approximation ratio for $\nrm{\vec{x}}$ prior to final step \\
        \hline
        $\hat{\chi}$ & 0.0398 & magnitude of errors in kernel reflection on final step \\
        \hline
        $\hat{c}$ & 20.0 & factor by which effective condition number $\sigma^{-1}$ is increased each step  \\
        \hline
        $\hat{r}$ & 3.37 & factor by which target approximation ratio decreases each step \\
        \hline 
        $\hat{q}$ & 5.41 & factor by which kernel reflection error parameter is increased each step \\
        \hline 
        $\hat{\Delta}$ & 0.00424 & upper bound on total probability of repeated failure of kernel reflection\\
        \hline
    \end{tabular}
    \caption{Summary of free parameters in \autoref{algo:full_QLSS_constant_prefactors}, along with their role. A specific choice for each parameter is provided which leads the final complexity statement to achieve a local minimum.}\label{tab:free_params}
\end{table}

Additionally, in its operation and analysis, we define the following parameters to depend on the above. The index $j$ ranges from 1 to $J$, where $J$ is defined in the first line. 
\begin{equation}\label{eq:optimal_QLSS_param_relations}
\begin{split}
    J &= \lceil \ln(\kappa)/ \ln(\hat{c}    )\rceil \\
    \sigma_j &= \kappa^{-1}\hat{c}^{J-j} \in [\kappa^{-1},1] \\
    \beta_j &= \hat{\beta}\hat{r}^{J-j-1} \\
    \chi_j &= \hat{\chi} \hat{q}^{-J+j} \\
    \eta_j &= \chi_j / (1+\chi_j) \\ 
    m_j &= (J-j+1)\ln(
         \hat{\Delta}^{-1}+1)\frac{(1+\eta_j)^2(\ln(\hat{c}\hat{\beta}^2 \hat{r}^{2J-2j})+1)}{(1-\eta_j)^2}  \\
    \bar{P} &= \frac{4}{\hat{\beta}^2(1-\hat{r}^{-2})}
        +\frac{16\hat{\chi}}{\hat{q}\hat{\beta}^2(1-\hat{r}^{-2}\hat{q}^{-1})}
        +\frac{16\hat{\chi}^2}{\hat{q}^2\hat{\beta}^2(1-\hat{r}^{-2}\hat{q}^{-2})}
        +  \frac{2\hat{\chi}^2\ln(e\hat{c}^2\hat{r}^4)}{\hat{q}^2-1} \\
    \mu &= \sqrt{1-\left(1-\hat{\Delta} - \bar{P}\right)\left(1 - 2\hat{\chi}^2\left(\frac{3}{2} + \ln\left(\frac{\hat{c}\hat{\beta}^2+1}{2}\right) \right)\right)}\\
    \eta_{\mathrm KP} &= \frac{\varepsilon}{\mu}\frac{\sqrt{1-\mu^2}}{\sqrt{1-\varepsilon^2}} 
\end{split}
\end{equation}

\begin{algorithm}
\caption{Full optimal  QLSS}\label{algo:full_QLSS_constant_prefactors}
\DontPrintSemicolon
\SetAlgoLined
\LinesNumbered
\Input{$(A, \vec{b}, \kappa, \varepsilon)$ and tunable parameters $(\hat{c}, \hat{q}, \hat{\beta},\hat{r}, \hat{\chi}, \hat{\Delta})$}
\Output{$(t, \ket{\vec{\tilde{x}}})$}
\BlankLine
Let $J$ be given as in Eq.~\eqref{eq:optimal_QLSS_param_relations} \label{line:fullQLSS_line1}\; 
Let $\beta_0 = 1$ and $t_0 = 1$\;
\For{$j = 1,\ldots, J$}
{
Let $\sigma_{j}$, $\eta_j$, $m_j$, $\beta_j$, be given as in Eq.~\eqref{eq:optimal_QLSS_param_relations}\;
Let $[\mathcal{L}_j, \mathcal{R}_j] = [\max(1,t_{j-1}/\beta_{j-1}), \min(\sigma_j^{-1}, \hat{c}t_{j-1} \beta_{j-1})]$\;
Let $\bar{A}_{\sigma_j}$ be given as in Eq.~\eqref{eq:barA_sigma} and its block-encoding be constructed as in \autoref{fig:block_encoding_barAsigma}\;
If $j < J$, run \autoref{algo:random_t_learn_norm} with parameters $(\bar{A}_{\sigma_j},\vec{b},[\mathcal{L}_j,\mathcal{R}_j],\sigma_j^{-1},\eta_j)$, and if $j=J$ use parameters $(A,\vec{b},[\mathcal{L}_j,\mathcal{R}_j],\kappa,\eta_j)$. If it fails, repeat this step up to $m_j$ times. If all $m_j$ repetitions fail, return to \autoref{line:fullQLSS_line1}. Upon first success, denote its output by $(t_j,\ket{\vec{x_j}})$ and continue. \label{line:optimal_QLSS_stepj}\; 
}
Let $\eta_{\mathrm{KP}}$ be given as in Eq.~\eqref{eq:optimal_QLSS_param_relations} \;
Apply KP to $\ket{\vec{x_{J}}}$ using the matrix $G$ with parameters $(\kappa,\eta_{\rm KP})$. If it fails, go to \autoref{line:fullQLSS_line1}. If it succeeds, denote its output by $\ket{\vec{\tilde{x}}}$.\label{line:optimal_QLSS_KP}\;\label{line:run_KP_optimal_QLSS}
\Return{$(t = t_J , \ket{\vec{\tilde{x}}})$}
\end{algorithm}

\begin{theorem}\label{thm:optimal_QLSS}
Suppose that $\nrm{\vec{b}} = 1$, that $\vec{b}$ is in the column space of $A$, and that all nonzero singular values of $A$ lie in the interval $[\kappa^{-1},1]$. Let $U_A$ be a $(1,a)$-block-encoding of $A$, and $U_{\vec{b}}$ be a state-preparation unitary for $\ket{\vec{b}}$.  Let $\vec{x}$ denote the unique vector of minimum norm $\nrm{\vec{x}}$ for which $A\vec{x} = \vec{b}$. Fix $\varepsilon > 0$. Consider \autoref{algo:full_QLSS_constant_prefactors}, denoting its output by $(t,\ket{\vec{\tilde{x}}})$, using the parameter choices listed in \autoref{tab:free_params}.  

Then, the ensemble of outputs $\ketbra{\vec{\tilde{x}}}$, denoted by $\tilde{\rho}$, satisfies $\frac{1}{2} \nrm{\ketbra{\vec{x}}-\tilde{\rho}}_1 \leq \varepsilon$, meaning the output solves the QLSP. Furthermore, the expected number of queries required is $Q$ total queries to $U_A$, $U_A^\dagger$ and their controlled versions, as well as $2Q$ total queries to $U_{\vec{b}}$, $U_{\vec{b}}^\dagger$, and their controlled versions, where
\begin{align}\label{eq:query_complexity_optimal_QLSS}
    Q \leq 56.0 \kappa + 1.05 \kappa \ln(\frac{\sqrt{1-\varepsilon^2}}{\varepsilon})  + 2.78 \ln(\kappa)^3 + 3.17
\end{align}
\end{theorem}

\begin{proof}

    The first part of the algorithm consists of computing a sequence of estimates $t_j$ for $j=1,\ldots,J = \lceil \ln(\kappa)/\ln(\hat{c})\rceil$. The final step of this sequence also produces a state $\ket{\vec{x_J}}$. For any value of $j$, we say that step $j$ ``succeeds'' if the algorithm advances to step $j+1$, and step $j$ ``fails'' if all $m_j$ repetitions within \autoref{line:optimal_QLSS_stepj} fail, causing the algorithm to restart and return to \autoref{line:fullQLSS_line1}. We say that the algorithm begins a new ``cycle'' when one of the steps fails and it returns to \autoref{line:fullQLSS_line1}, with each cycle being completely independent of the previous cycles.  
    
    Conditioned on steps $1,2,\ldots,j$ all succeeding, there is some fixed ensemble of outputs $t_j$. We argue that this ensemble is concentrated near $\nrm{\vec{\bar{x}_{\sigma_j}}}$ for all $j$. In particular, we will call an estimate $t_j$ ``good'' if $t_j \in [\beta_j^{-1}\nrm{\vec{\bar{x}_{\sigma_j}}}, \beta_j \nrm{\vec{\bar{x}_{\sigma_j}}}]$ and ``bad'' otherwise. Furthermore, we recall from Property 3 shown in \autoref{sec:estimate_by_adiabatic} that \begin{equation}
        \nrm{\vec{\bar{x}_{\sigma_{j}}}} \leq \nrm{\vec{\bar{x}_{\sigma_{j+1}}}} \leq \nrm{\vec{\bar{x}_{\sigma_{j}}}} \sigma_{j} / \sigma_{j+1} = \hat{c}\nrm{\vec{\bar{x}_{\sigma_{j-1}}}}
    \end{equation}
    Thus, if $t_j$ is good, it holds that $\nrm{\vec{\bar{x}_{\sigma_{j+1}}}} \in [\beta_j^{-1}t_j, \beta_j \hat{c}t_j]$. Separately we know that $\nrm{\vec{\bar{x}_{\sigma_{j+1}}}} \in [1,\sigma_{j+1}^{-1}]$ since $\bar{x}_{\sigma_{j+1}}$ is the solution to a linear system with condition number at most $\sigma_{j+1}^{-1}$, so we conclude that $\nrm{\vec{\bar{x}_{\sigma_{j+1}}}} \in [\beta_j^{-1}t_j, \beta_j \hat{c}t_j] \cap [1,\sigma_{j+1}^{-1}] = [\mathcal{L}_{j+1}, \mathcal{R}_{j+1}]$. 

    Let $p_j$ denote the probability that $t_j$ is bad conditioned on $\nrm{\vec{\bar{x}_{\sigma_j}}} \in [\mathcal{L}_j, \mathcal{R}_j]$, and note that by \autoref{lem:estimate_norm_random_t}, we can bound $p_j$ by 
    \begin{equation}\label{eq:pj_bound}
        p_j := \Pr[t_j \text{ is bad}\;\; \big|\;\; \nrm{\vec{\bar{x}_{\sigma_j}}} \in [\mathcal{L}_j, \mathcal{R}_j] ] \leq \frac{(1+\eta_j)^2}{(1-\eta_j)^2}\frac{4}{\beta_j^2+1} + \frac{2\eta_j^2(2\ln(\mathcal{R}_j/\mathcal{L}_j)+1-4\ln(\beta_j))}{(1-\eta_j)^2}\,.
    \end{equation}

    Moreover, let $\delta_j$ be the probability that all $m_j$ repetitions of \autoref{algo:random_t_learn_norm} within step $j$ fail, causing the algorithm to return to \autoref{line:fullQLSS_line1}, conditioned on $\nrm{\vec{\bar{x}_{\sigma_j}}} \in [\mathcal{L}_j, \mathcal{R}_j]$. We know from \autoref{lem:estimate_norm_random_t} that the probability of success of each run of \autoref{algo:random_t_learn_norm} is at least $q_{\rm succ}$ where $q_{\rm succ}$ is bounded in Eq.~\eqref{eq:q_succ_lowerbound}, and thus $\delta_j \leq (1-q_{\rm succ})^{m_j}$. That is,
    \begin{align}\label{eq:deltaj_bound}
        \delta_j \leq \left(1-\frac{(1-\eta_j)^2}{(1+\eta_j)^2(\ln(\mathcal{R}_j/\mathcal{L}_j)+1)}\right)^{m_j}\,.
    \end{align}

    Define
    \begin{align}
        P &= \sum_{j=1}^{J-1} p_j \\
        \Delta &= \sum_{j=1}^{J} \delta_j\,.
    \end{align}
    If the relation $\nrm{\vec{\bar{x}_{\sigma_{j}}}} \in [\mathcal{L}_{j}, \mathcal{R}_{j}]$ is satisfied, then step $j$ will succeed and output a $t_{j}$ that is good, except with probability at most $\delta_j + p_j$. Then, as argued above, when $t_j$ is good, $\nrm{\vec{\bar{x}_{\sigma_{j+1}}}} \in [\mathcal{L}_{j+1}, \mathcal{R}_{j+1}]$ is satisfied. By the union bound, the probability that one of the steps fails (resetting the cycle before it completes) or all succeed but one of the $t_1, \ldots, t_{J-1}$ is bad, can be upper bounded by $P+\Delta$. \textit{In other words, with probability at least $1-P-\Delta$, steps $1,\ldots,j$ succeed and all outputs $t_1,\ldots, t_{j-1}$ are good. }

    Now, we verify that the output of the algorithm solves the QLSP. With probability $1-P-\Delta$, in any given cycle, step $J$ is executed on input parameters satisfying $\nrm{\vec{\bar{x}_{\sigma_J}}} = \nrm{\vec{x}} \in [\mathcal{L}_J, \mathcal{R}_J]$, and it succeeds. In this scenario, \autoref{lem:estimate_norm_random_t} applies, and the ensemble $\tilde{\rho}_{J, \rm good}$ of ouptuts $\ketbra{\vec{\tilde{x}_J}}$ satisfies
    \begin{align}
        \bra{\vec{x}}\tilde{\rho}_{J, \rm good}\ket{\vec{x}} &\geq 1- \frac{2\eta_J^2}{(1-\eta_J)^2}\left(\frac{3}{2} + \ln\left(\frac{\mathcal{R}_J^2+\mathcal{L}_J^2}{2\mathcal{L}_J^2}\right) \right) 
    \end{align}
    With probability at most $P+\Delta$, either the cycle does not reach the KP step (\autoref{line:optimal_QLSS_KP}), or it does reach the KP step but $\nrm{\vec{\bar{x}_{\sigma_J}}} \in [\mathcal{L}_J, \mathcal{R}_J]$ does not hold, and thus \autoref{lem:estimate_norm_random_t} does not apply. In this scenario, we may still say that the ensemble $\tilde{\rho}_{J, \rm bad}$ of outputs of step $J$ (conditioned on an output existing) satisfies  $
        \bra{\vec{x}}\tilde{\rho}_{J, \rm bad}\ket{\vec{x}} \geq 0$. 
    Weighting these two ensembles by their (worst-case) probabilities, we conclude that the ensemble $\tilde{\rho}_J$ of outputs of step $J$ satisfies
    \begin{equation}\label{eq:<x|rho_J|x>_prelim}
        \bra{\vec{x}}\tilde{\rho}_J\ket{\vec{x}} \geq (1-P-\Delta)\bra{\vec{x}}\tilde{\rho}_{J, \rm good}\ket{\vec{x}} \geq  (1-P-\Delta)\left(1- \frac{2\eta_J^2}{(1-\eta_J)^2}\left(\frac{3}{2} + \ln\left(\frac{\mathcal{R}_J^2+\mathcal{L}_J^2}{2\mathcal{L}_J^2}\right) \right) \right)
    \end{equation}

    We now compute some bounds on the above using the particular parameter choices we have made, with the goal of showing that $\bra{\vec{x}}\tilde{\rho}_J\ket{\vec{x}} \geq 1-\mu^2$, with $\mu$ given in Eq.~\eqref{eq:optimal_QLSS_param_relations}. First, we note that by our definitions, for all $j$, $\eta_j /(1+\eta_j) = \chi_j$, and $(1+\eta_j)/(1-\eta_j) = 1+2\chi_j$. Furthermore, for $j \geq 2$, we have $\mathcal{R}_j/\mathcal{L}_j \leq \hat{c} \hat{\beta}_{j-1}^2 = \hat{c}\hat{\beta}^2 \hat{r}^{2J-2j}$. For $j=1$, we have $\mathcal{R}_1/\mathcal{L}_1 = \sigma_1^{-1} \leq \hat{c} \leq \hat{c}\hat{\beta}^2 \hat{r}^{2J}$ (since $\hat{\beta} > 1$ and $\hat{r} > 1$). Thus, referencing Eq.~\eqref{eq:pj_bound} and substituting $k = J-j-1$, we have
    \begin{align}
        P &= \sum_{j=1}^{J-1} p_j \leq \sum_{j=1}^{J-1} \left[\frac{(1+\eta_j)^2}{(1-\eta_j)^2}\frac{4}{\beta_j^2+1} + \frac{2\eta_j^2(2\ln(\mathcal{R}_j/\mathcal{L}_j)+1-4\ln(\beta_j))}{(1-\eta_j)^2} \right]\\
        &\leq \sum_{j=1}^{J-1}\left[4(1+2\chi_j)^2\beta_j^{-2}+ 2\chi_j^2\left(2\ln(\hat{c}\hat{\beta}_{j-1}^2)+1-4\ln(\beta_j)\right) \right] \\
        &= \sum_{j=1}^{J-1}\left[ 4(1+2\hat{\chi}\hat{q}^{-J+j})^2\hat{\beta}^{-2}\hat{r}^{-2J+2j} + 2\hat{\chi}^2\hat{q}^{-2J+2j}\left(2\ln(\hat{c}\hat{\beta}^2 \hat{r}^{2J-2j})+1-4\ln(\hat{\beta}^2 \hat{r}^{J-j-1})\right)\right] \\
        &= \sum_{k=0}^{J-2} 
        \left[4(1+2\hat{\chi}\hat{q}^{-k-1})^2\hat{\beta}^{-2}\hat{r}^{-2k} + 2\hat{\chi}^2\hat{q}^{-2k-2}\ln(e\hat{c}^2\hat{r}^{4})\right] \\
        &\leq \left[\frac{4}{\hat{\beta}^2}\sum_{k=0}^\infty \hat{r}^{-2k} \right] +\left[\frac{16\hat{\chi}}{\hat{q}\hat{\beta}^2}\sum_{k=0}^\infty \hat{r}^{-2k}\hat{q}^{-k} \right] +\left[\frac{16\hat{\chi}^2}{\hat{q}^2\hat{\beta}^2}\sum_{k=0}^\infty \hat{r}^{-2k}\hat{q}^{-2k} \right] +  \left[\frac{2\hat{\hat{\chi}}^2\ln(e\hat{c}^2\hat{r}^4)}{\hat{q}^2}\sum_{k=0}^\infty \hat{q}^{-2k} \right] \\
        &= \frac{4}{\hat{\beta}^2(1-\hat{r}^{-2})}
        +\frac{16\hat{\chi}}{\hat{q}\hat{\beta}^2(1-\hat{r}^{-2}\hat{q}^{-1})}
        +\frac{16\hat{\chi}^2}{\hat{q}^2\hat{\beta}^2(1-\hat{r}^{-2}\hat{q}^{-2})}
        +  \frac{2\hat{\chi}^2\ln(e\hat{c}^2\hat{r}^4)}{\hat{q}^2(1-\hat{q}^{-2})} \\
        &= \bar{P}
    \end{align}
    where $\bar{P}$ was defined in Eq.~\eqref{eq:optimal_QLSS_param_relations}. Similarly, working from Eq.~\eqref{eq:deltaj_bound}, we bound $\Delta$, here utilizing the relation $1/x \geq -1/\ln(1-x)$, the substitution $k=J-j+1$, and the definition of $m_j$ in Eq.~\eqref{eq:optimal_QLSS_param_relations},
    \begin{align}
        \Delta &= \sum_{j=1}^J \delta_j \leq \sum_{j=   }^J\left(1-\frac{(1-\eta_j)^2}{(1+\eta_j)^2(\ln(\mathcal{R}_j/\mathcal{L}_j)+1)}\right)^{m_j} \\
        &\leq \sum_{j=1}^J\left(1-\frac{(1-\eta_j)^2}{(1+\eta_j)^2\left(\ln(\hat{c}\hat{\beta}^2 \hat{r}^{2J-2j})+1\right)}\right)^{(J-j+1)\ln(
        \hat{\Delta}^{-1}+1)\frac{(1+\eta_j)^2(\ln(\hat{c}\hat{\beta}^2 \hat{r}^{2J-2j})+1)}{(1-\eta_j)^2}} \\
        &\leq \sum_{j=1}^J\left(1-\frac{(1-\eta_j)^2}{(1+\eta_j)^2(\ln(\hat{c}\hat{\beta}^2 \hat{r}^{2J-2j})+1)}\right)^{(J-j+1)\frac{\ln(\frac{1}{
        \hat{\Delta}^{-1}+1})}{\ln(1-\frac{(1-\eta_j)^2}{(1+\eta_j)^2(\ln(\hat{c}\hat{\beta}^2 \hat{r}^{2J-2j})+1)})}} \\
        &= \sum_{j=1}^J \left(\frac{\hat{\Delta}}{1+\hat{\Delta}}\right)^{J-j+1} \leq \sum_{k=1}^\infty \frac{\hat{\Delta}^k}{(1+\hat{\Delta})^k} = \hat{\Delta}
    \end{align}
    Thus we have
    \begin{align}
        \bra{\vec{x}}\tilde{\rho}_J\ket{\vec{x}} &\geq   (1-\hat{\Delta} - \bar{P})\left(1 - \frac{2\eta_J^2}{(1-\eta_J)^2}\left(\frac{3}{2} + \ln\left(\frac{\mathcal{R}_J^2+\mathcal{L}_J^2}{2\mathcal{L}_J^2}\right) \right)\right) \\
        &\geq 1-\left(1-\left(1-\hat{\Delta} - \bar{P}\right)\left(1 - 2\hat{\chi}^2\left(\frac{3}{2} + \ln\left(\frac{\hat{c}\hat{\beta}^2+1}{2}\right) \right)\right)\right)\\
        &= 1-\mu^2\,.
    \end{align}
    where $\mu$ was defined in Eq.~\eqref{eq:optimal_QLSS_param_relations}. 
    The subsequent KP step can then be analyzed with \autoref{lem:KP}. We observe that the output $\ket{\vec{x_J}}$ has no overlap with the kernel of $A$ (see \autoref{thm:QLSS_known_norm}). Thus, \autoref{lem:KP} asserts that the output of the KP step satisfies
\begin{equation}
    \frac{1}{2}\nrm{\ketbra{\vec{x}}-\tilde{\rho}}_1 \leq \frac{\mu\eta_{\rm KP}}{\sqrt{1-\mu^2+\mu^2\eta_{\rm KP}^2}} = \varepsilon\,,
\end{equation}
where the last equality is ensured by the choice of $\eta_{\rm KP}$ in Eq.~\eqref{eq:optimal_QLSS_param_relations}. Thus we have verified that the ensemble output by the KP step, when it succeeds, solves the QLSP. 

Next, we give an upper bound on the overall expected query complexity. Since each cycle is independent of the previous cycle, this query complexity is given by the expected number of cycles times the expected number of queries per cycle. The expected number of cycles is the inverse of the probability that all steps $1, \ldots, J$ succeed, and then the KP step also succeeds. This quantity is at least the probability that all steps $1, \ldots, J$ succeed and that all estimates $t_1,\ldots, t_{J-1}$ are good, times the probability the KP step succeeds conditioned on the output $\tilde{\rho}_{J, \rm good}$. This can then be lower bounded by  $(1-P-\Delta)\bra{\vec{x}}\tilde{\rho}_{J,\rm good}\ket{\vec{x}}$, which was itself lower bounded by $1-\mu^2$ above. Thus, we conclude that the expected  number of cycles is at most $(1-\mu^2)^{-1}$. 

Let $\mathcal{Q}_j$ be the expected query complexity incurred by step $j$ within any given cycle, for $j=1,\ldots, J$, and let $\mathcal{Q}_{\rm KP}$ be the expected query complexity incurred by the KP step (\autoref{line:optimal_QLSS_KP}) in any given cycle. We may then make the upper bound on the total expected query complexity 
\begin{equation}\label{eq:Qbound_fromQjQKP}
    Q \leq \frac{\mathcal{Q}_{\rm KP} + \sum_{j=1}^J \mathcal{Q}_{j}}{1-\mu^2}
\end{equation}
We now provide expressions upper bounding $\mathcal{Q}_j$ and $\mathcal{Q}_{\rm KP}$. We know that when step $j$ executes,  the probability that $\nrm{\vec{\bar{x}_{\sigma_j}}} \in [\mathcal{L}_j, \mathcal{R}_j]$ is at least $1-P-\Delta$. When this is the case, each call to \autoref{algo:random_t_learn_norm} within step $j$ succeeds with probability lower bounded by Eq.~\eqref{eq:q_succ_lowerbound}, and the expected number of calls to  \autoref{algo:random_t_learn_norm} is at most the inverse of that quantity. With probability at most $P+\Delta$, we have no such guarantee, but we may still say that the expected number of calls to \autoref{algo:random_t_learn_norm} is $m_j$. Each time \autoref{algo:random_t_learn_norm} is called, it has a fixed query complexity equal to $2\lceil \sigma_j^{-1} \ln(2/\eta_j)/2\rceil$. Thus, in total we have
\begin{equation}\label{eq:Qj_bound}
    \mathcal{Q}_j \leq \left((1-P-\Delta)\frac{(1+\eta_j)^2\left(\ln(\hat{c}\hat{\beta}_{j-1}^2)+1\right)}{(1-\eta_j)^2} + (P+\Delta) m_j\right) 2\lceil \sigma_j^{-1} \ln(2/\eta_j)/2\rceil
\end{equation}
On the other hand, the KP step executes at most once per cycle, and it has the same query complexity every time it executes, given by $2\lceil \kappa \ln(2/\eta_{\rm KP})/2\rceil $. Thus we have
\begin{equation}\label{eq:QKP_bound}
    \mathcal{Q}_{\rm KP} \leq 2\lceil \kappa \ln(2/\eta_{\rm KP})/2\rceil \leq \kappa \ln(\frac{\sqrt{1-\varepsilon^2}}{\varepsilon}) - \kappa \ln(\frac{\sqrt{1-\mu^2}}{2\mu})+2
\end{equation}
where we have used the bound $\lceil x \rceil \leq x+1$ and substituted the definition of $\eta_{\rm KP}$ from Eq.~\eqref{eq:optimal_QLSS_param_relations}. 

We now update the bound on $\mathcal{Q}_j$ for our particular parameter choices. We note the relations $(1+\eta_j)/(1-\eta_j) \leq 1+2\hat{\chi}$ and $1/\eta_j = (1+\chi_j)/\chi_j \leq (1+\hat{\chi})/\chi_j = (1+\hat{\chi})\hat{q}^{J-j}/\hat{\chi}$. We work from Eq.~\eqref{eq:Qj_bound}, using the definition of $m_j$ in Eq.~\eqref{eq:optimal_QLSS_param_relations}, to write
\begin{align}
    \mathcal{Q}_j \leq{}&\left(\frac{\ln(\hat{c}\hat{\beta}^2\hat{r}^{2J-2j})+1}{(1-\eta_j)^2(1+\eta_j)^{-2}} + (P+\Delta)\left(m_j- \frac{\ln(\hat{c}\hat{\beta}^2\hat{r}^{2J-2j})+1}{(1-\eta_j)^2(1+\eta_j)^{-2}}\right)\right) 2\lceil \sigma_j^{-1} \ln(2/\eta_j)/2\rceil \\
    \leq{}& \frac{\ln(\hat{c}\hat{\beta}^2\hat{r}^{2J-2j})+1}{(1-\eta_j)^2(1+\eta_j)^{-2}}\left(1 + (\bar{P}+\hat{\Delta})\left(\ln(\hat{\Delta}^{-1}+1)(J-j+1)-1\right)\right) \left( \sigma_j^{-1} \ln(2/\eta_j)+2\right) \\
    \nonumber ={}& (1+2\hat{\chi})^2\left(\ln(e\hat{c}\hat{\beta}^2) + (2J-2j)\ln(\hat{r})\right)\left(1 + (\bar{P}+\hat{\Delta})\left(\ln(\hat{\Delta}^{-1}+1)(J-j+1)-1\right)\right) \\
        &\qquad \times \left( \kappa \hat{c}^{-J+j}\ln(2\hat{\chi}^{-1}(1+\hat{\chi}))+\kappa \hat{c}^{-J+j}\ln(\hat{q})(J-j) +2\right)
\end{align}
Working from this last expression, we take the sum over $j$ and group terms by their $j$ dependence, as follows. 
\begin{align}\label{eq:sum_Qj_bound}
    \nonumber \sum_{j=1}^J \mathcal{Q}_j 
    &\leq{} \kappa Z_0 \left[\sum_{j=1}^J \hat{c}^{-J+j}\right] + \kappa Z_1 \left[\sum_{j=1}^J (J-j)\hat{c}^{-J+j}\right] + \kappa Z_2 \left[\sum_{j=1}^J (J-j)^2\hat{c}^{-J+j}\right] \\
    & \qquad + \kappa Z_3\left[\sum_{j=1}^J (J-j)^3\hat{c}^{-J+j}\right] + JZ_4 + J^2 Z_5  \left[\sum_{j=1}^J\frac{J-j}{J^2} \right]+ J^3 Z_6 \left[\sum_{j=1}^J\frac{(J-j)^2}{J^3} \right]
\end{align}
where $Z_p$ are constants (independent of $\kappa$ and $\varepsilon$), given by
\begin{align}
     Z_0 &= (1+2\hat{\chi})^2\ln(\frac{2(1+\hat{\chi})}{\hat{\chi}})\ln(e\hat{c}\hat{\beta}^2)\left(1-\bar{P}-\hat{\Delta} + (\bar{P} + \hat{\Delta})\ln(\hat{\Delta}^{-1}+1)\right)  \\
    Z_1 &=  (1+2\hat{\chi})^2\ln(\frac{2(1+\hat{\chi})}{\hat{\chi}})\ln(e\hat{c}\hat{\beta}^2)(\bar{P} + \hat{\Delta})\ln(\hat{\Delta}^{-1}+1) \nonumber\\
    & \qquad + 2(1+2\hat{\chi})^2\ln(\frac{2(1+\hat{\chi})}{\hat{\chi}})\ln(\hat{r})\left(1-\bar{P}-\hat{\Delta} + (\bar{P} + \hat{\Delta})\ln(\hat{\Delta}^{-1}+1)\right) \nonumber \\
    & \qquad + (1+2\hat{\chi})^2\ln(\hat{q})\ln(e\hat{c}\hat{\beta}^2)\left(1-\bar{P}-\hat{\Delta} + (\bar{P} + \hat{\Delta})\ln(\hat{\Delta}^{-1}+1)\right) \\
    Z_2 &=  2(1+2\hat{\chi})^2\ln(\frac{2(1+\hat{\chi})}{\hat{\chi}})\ln(\hat{r})(\bar{P} + \hat{\Delta})\ln(\hat{\Delta}^{-1}+1) \nonumber\\
    & \qquad + 2(1+2\hat{\chi})^2\ln(\hat{q})\ln(\hat{r})\left(1-\bar{P}-\hat{\Delta} + (\bar{P} + \hat{\Delta})\ln(\hat{\Delta}^{-1}+1)\right) \nonumber \\
    & \qquad + (1+2\hat{\chi})^2\ln(\hat{q})\ln(e\hat{c}\hat{\beta}^2)(\bar{P} + \hat{\Delta})\ln(\hat{\Delta}^{-1}+1) \\
    Z_3 &= 2(1+2\hat{\chi})^2\ln(\hat{q})\ln(\hat{r})(\bar{P} + \hat{\Delta})\ln(\hat{\Delta}^{-1}+1) \\
    Z_4 &= 2(1+2\hat{\chi})^2\ln(e\hat{c}\hat{\beta}^2)\left(1 -\bar{P}- \hat{\Delta} + (\bar{P}+\hat{\Delta})\ln(\hat{\Delta}^{-1}+1)\right)\\
    Z_5 &= 4(1+2\hat{\chi})^2\ln(\hat{r})\left(1 -\bar{P}- \hat{\Delta} + (\bar{P}+\hat{\Delta})\ln(\hat{\Delta}^{-1}+1)\right) + 2(1+2\hat{\chi})^2\ln(e\hat{c}\hat{\beta}^2)(\bar{P}+ \hat{\Delta})\ln(\hat{\Delta}^{-1}+1)\\
    Z_6 &= 4(1+2\hat{\chi})^2\ln(\hat{r})(\bar{P} + \hat{\Delta})\ln(\hat{\Delta}^{-1}+1)
\end{align}
Furthermore, each of the sums in brackets evaluates to a quantity that is upper bounded by a constant, independent of $\kappa$. We can compute upper bounds with the substitution $k = J-j$ and extending the sum to infinity.
\begin{align}
    &\sum_{j=1}^J \hat{c}^{-J+j} \leq \sum_{k=0}^\infty \hat{c}^{-k} = \frac{1}{1-\hat{c}^{-1}}\\
    &\sum_{j=1}^J (J-j) \hat{c}^{-J+j} \leq \sum_{k=0}^\infty k \hat{c}^{-k} = \frac{\hat{c}^{-1}}{(1-\hat{c}^{-1})^2} \\
    &\sum_{j=1}^J (J-j)^2 \hat{c}^{-J+j} \leq \sum_{k=0}^\infty k^2 \hat{c}^{-k} = \frac{\hat{c}^{-2} + \hat{c}^{-1}}{(1-\hat{c}^{-1})^3} \\
    &\sum_{j=1}^J (J-j)^3 \hat{c}^{-J+j} \leq \sum_{k=0}^\infty k^2 \hat{c}^{-k} = \frac{\hat{c}^{-3} + 4 \hat{c}^{-2} + \hat{c}^{-1}}{(1-\hat{c}^{-1})^4} \\
    &\sum_{j=1}^J \frac{J-j}{J^2} = \frac{1}{2} - \frac{1}{2J} \leq \frac{1}{2} \\
    &\sum_{j=1}^J \frac{(J-j)^2}{J^2} = \frac{(J-1)(2J-1)}{6J^2} \leq  \frac{1}{3}
\end{align}

We are now ready to conclude. We plug our bound on $\sum_{j=1}^J \mathcal{Q}_j$ from Eq.~\eqref{eq:sum_Qj_bound} and our bound on $\mathcal{Q}_{\rm KP}$ in Eq.~\eqref{eq:QKP_bound} into Eq.~\eqref{eq:Qbound_fromQjQKP}. We find an upper bound on $Q$ equal to
\begin{align}
    &\frac{\kappa}{1-\mu^2}\left[ \ln(\frac{\sqrt{1-\varepsilon^2}}{\varepsilon}) + \frac{Z_0}{1-\hat{c}^{-1}} + \frac{Z_1\hat{c}^{-1}}{(1-\hat{c}^{-1})^2} + \frac{Z_2(\hat{c}^{-2}+\hat{c}^{-1})}{(1-\hat{c}^{-1})^{3}} + \frac{Z_3(\hat{c}^{-3} +4\hat{c}^{-2}+\hat{c}^{-1})}{(1-\hat{c}^{-1})^4} -\ln(\frac{\sqrt{1-\mu^2}}{2\mu})\right] \nonumber \\
    & \qquad + 
    \lceil \log_{\hat{c}}(\kappa)\rceil \left(\frac{Z_4}{1-\mu^2}\right)+ \lceil \log_{\hat{c}}(\kappa)\rceil^2\left(\frac{Z_5}{2(1-\mu^2)}\right) + \lceil \log_{\hat{c}}(\kappa)\rceil^3\left(\frac{Z_6}{2(1-\mu^2)}\right)+ \left(\frac{3}{1-\mu^2}\right)
\end{align}
where $Z_p$ are given above. As long as $\kappa \geq \hat{c}$, we can say $\lceil \log_{\hat{c}}(\kappa) \rceil \leq 2(\ln(\kappa)/\ln(\hat{c}))^3$, $\lceil \log_{\hat{c}}(\kappa) \rceil^2 \leq 4(\ln(\kappa)/\ln(\hat{c}))^3$ and $\lceil \log_{\hat{c}}(\kappa) \rceil^2 \leq 8(\ln(\kappa)/\ln(\hat{c}))^3$. This allows us to upper bound the final line (sublinear in $\kappa$) by
\begin{equation}
    \frac{\ln(\kappa)^3}{\ln(\hat{c})^3} \left(\left(\frac{2Z_4}{1-\mu^2}\right)+\left(\frac{2Z_5}{1-\mu^2}\right) + \left(\frac{4Z_6}{1-\mu^2}\right)\right)+ \left(\frac{3}{1-\mu^2}\right) 
\end{equation}
This verifies that the final complexity is $O(\kappa) + O(\kappa \log(1/\varepsilon))$. By plugging in the parameter choices from \autoref{tab:free_params}, the numerical values reported can be verified from this final expression. 
\end{proof}

\section{Lower bounds on the complexity of norm estimation}\label{app:norm_query_lower_bound}

Our method has demonstrated that estimating the norm $\nrm{\vec{x}}$ to within a constant factor is a key step toward achieving an optimal QLSS with $O(\kappa)$ complexity. One might hope that estimating the norm could be easier than producing the state $\ket{\vec{x}}$. However, here we show an $\Omega(\kappa)$ lower bound on estimating the norm to within a constant factor less than 5/4. Our proof extends the method in Ref.~\cite{Orsucci2021solvingclassesof}, which showed an $\Omega(\kappa)$ lower bound on the query complexity of the QLSP even for positive semi-definite matrices. 

The key idea of that method was to reduce the \texttt{PromiseMajority} problem to the QLSP, such that solving the QLSP yields a solution to \texttt{PromiseMajority}. Known lower bounds on \texttt{PromiseMajority} then imply a lower bound the QLSP. Here we do the same, modifying the construction of Ref.~\cite{Orsucci2021solvingclassesof} so that we reduce to the problem of estimating the norm, rather than the QLSP.  Technically, as stated, our bound leaves open the possibility of achieving a multiplicative-factor approximation worse than $5/4$ in $o(\kappa)$ complexity. 

\begin{theorem}\label{thm:norm_query_lower_bound}
    Let $\mathcal{A}$ be a quantum algorithm for esitmating the norm in the following sense. On any input $\kappa \in [3,\infty)$, $\varepsilon \in (0,1/4]$, and given (controlled) access to a $(1,a)$-block-encoding $U_A$ for the $N \times N$ matrix $A$ with singular values contained in $[\kappa^{-1},1]$, and state-preparation unitary $U_{\vec{b}}$ for the $N$-dimensional vector $\vec{b}$, the algorithm $\mathcal{A}$ outputs a value $t$, where, with probability at least 2/3, $t \in [(1+\varepsilon)^{-1}\nrm{\vec{x}}, (1+\varepsilon) \nrm{\vec{x}}]$ where $
    \vec{x}$ is the solution of minimum norm to the equation $A\vec{x} =\vec{b}$.  Then $\mathcal{A}$ must make at least
    \begin{equation}
        \Omega( \min(\kappa \varepsilon^{-1}, N))
    \end{equation}
    queries to $U_A$ and to $U_{\vec{b}}$. 
\end{theorem}
\begin{proof}
    First we define the $\texttt{PromiseMajority}(M,N')$ problem following Ref.~\cite{Orsucci2021solvingclassesof}: Given a vector $\vec{y} \in \{0,1\}^{N'}$ and a value of $M \in \{1,\ldots,N'\}$ (where $M+N'$ is even), and the promise that either (i) $y_i = 0$ for $(N'+M)/2$ of the entries, or (ii) $y_i=1$ for $(N'+M)/2$ of the entries, determine whether (i) or (ii) is the case.  Suppose a quantum algorithm has query access to the entries of $\vec{y}$ via a unitary $\mathcal{P}_{\vec{y}}$ acting as $\mathcal{P}_{\vec{y}}\ket{i}\ket{z} = \ket{i}\ket{ z \oplus y_i}$. Then, the number of queries to $\mathcal{P}_{\vec{y}}$ that the quantum algorithm must make  to solve the $\texttt{PromiseMajority}(M,N')$ problem with at least 2/3 probability of correctness is $\Omega(N'/M)$ (see \cite[Lemma 19]{Orsucci2021solvingclassesof} and \cite[Corollary 1.2]{nayak1999queryComplexityMedian}). 

    Now, we will define a family of linear systems, parameterized by $\kappa$, $\varepsilon$, and $N$. The linear systems will be size $N \times N$ (without loss of generality, we let $N$ be a power of 2) and have condition number $\kappa$. We will show that learning the norm of the solution to these linear systems to within factor $1+\varepsilon$ yields a solution to the $\texttt{PromiseMajority}(M,N')$ problem, where the relationship between parameters $(M,N')$ and $(\kappa, \varepsilon, N)$ is
    \begin{align}
        N' &= N/2 \\
        M &= 2 \max(1, 4 N\varepsilon/\kappa)\,.
    \end{align}
    Furthermore, we will show that the linear systems can be constructed such that the block-encoding $U_A$ can be carried out with two queries to the oracle $\mathcal{P}_{\vec{y}}$, while $U_{\vec{b}}$ requires zero queries; or it can be constructed such that $U_A$ requires zero queries, and $U_{\vec{b}}$ can be accomplished in one query. Thus, the lower bound of $\Omega(N'/M)$ queries to $\mathcal{P}_{\vec{y}}$ will imply a query lower bound on the norm estimation problem of $\Omega(\min(\kappa/\varepsilon,N))$. 
    
    To define the family, let $\vec{1} = (1;1;\ldots;1)/\sqrt{N'}$ denote the unit vector of length $N'$ with equal entries. Let $D$ be the $N' \times N'$ diagonal unitary matrix for which the $i$th diagonal entry of $D$ is equal to $(-1)^{y_i}$, and define unit vector
    \begin{equation}
        \vec{d} = D\vec{1} = \frac{1}{\sqrt{N'
        }}\begin{bmatrix}
            (-1)^{y_0} \\
            \vdots \\
            (-1)^{y_{N'-1}}
        \end{bmatrix}\,.
    \end{equation}
    Finally, let 
    \begin{equation}
        C = I_{N'} -(1-\kappa^{-1})\vec{1}\vec{1}^\dagger
    \end{equation}
    where $I_{N'}$ is the $N' \times N'$ identity matrix. We can see that all singular values of $C$ lie in the interval $[\kappa^{-1},1]$. We can also see that $C$ is invertible and that its inverse is given by 
    \begin{equation}
        C^{-1} = I_{N'} + (\kappa-1)\vec{1}\vec{1}^\dagger \,.
    \end{equation}
    Then, define the Hermitian $N\times N$ matrix $A$ and vector $\vec{b}$ by
    \begin{equation}
       A = 
       \frac{1}{2}\begin{bmatrix}
            C + I_{N'} & (C-I_{N'})D \\
            D(C-I_{N'}) & C+I_{N'}
        \end{bmatrix} \qquad \qquad  
        \vec{b} = \frac{1}{\sqrt{1 + \varepsilon^2N^2/M^2}}\begin{bmatrix}
            \vec{1} \\
            \varepsilon\frac{N}{M} \vec{1}
        \end{bmatrix}\,.
    \end{equation}
    We note the factorizations
    \begin{align}
        A = &\underbrace{\begin{bmatrix}
           I_{N'} & 0 \\
           0 & D \end{bmatrix}}_{\text{unitary}}
        \underbrace{\begin{bmatrix}
           I_{N'}/\sqrt{2} & I_{N'}/\sqrt{2} \\
           I_{N'}/\sqrt{2} & -I_{N'}/\sqrt{2}
        \end{bmatrix}}_{\text{unitary}}
        \begin{bmatrix}
           C & 0 \\
           0 & I_{N'} 
       \end{bmatrix} 
        \underbrace{\begin{bmatrix}
           I_{N'}/\sqrt{2} & I_{N'}/\sqrt{2} \\
           I_{N'}/\sqrt{2} & -I_{N'}/\sqrt{2}
        \end{bmatrix}}_{\text{unitary}}
       \underbrace{\begin{bmatrix}
           I_{N'} & 0 \\
           0 & D 
       \end{bmatrix}}_{\text{unitary}}\label{eq:A_factorization}\\
       \vec{b} = \frac{1}{\sqrt{2}\sqrt{1 + \varepsilon^2N^2/M^2}}&\underbrace{\begin{bmatrix}
           I_{N'} & 0 \\
           0 & D 
       \end{bmatrix}}_{\text{unitary}}
       \underbrace{\begin{bmatrix}
           I_{N'}/\sqrt{2} & I_{N'}/\sqrt{2} \\
           I_{N'}/\sqrt{2} & -I_{N'}/\sqrt{2}
        \end{bmatrix}}_{\text{unitary}}
        \begin{bmatrix}
            \vec{1} + \varepsilon\frac{N}{M} \vec{d} \\
            \vec{1} - \varepsilon\frac{N}{M} \vec{d}
        \end{bmatrix}
    \end{align}
    Since all factors of $A$ are unitary except for the third factor, and $C$ has singular values in $[\kappa^{-1},1]$, we conclude that $A$ has singular values in $[\kappa^{-1},1]$. This also implies that $A$ is invertible and $A\vec{x} = \vec{b}$ has a unique solution.  Furthermore, from the factorization we see that this unique solution is
    \begin{align}
        \vec{x} &= \frac{1}{\sqrt{2}{\sqrt{1 + \varepsilon^2N^2/M^2}}}\underbrace{\begin{bmatrix}
           I_{N'} & 0 \\
           0 & D \end{bmatrix}}_{\text{unitary}}
        \underbrace{\begin{bmatrix}
           I_{N'}/\sqrt{2} & I_{N'}/\sqrt{2} \\
           I_{N'}/\sqrt{2} & -I_{N'}/\sqrt{2}
        \end{bmatrix}}_{\text{unitary}}
        \begin{bmatrix}
           C^{-1} & 0 \\
           0 & I_{N'} 
       \end{bmatrix} 
        \begin{bmatrix}
            \vec{1} + \varepsilon\frac{N}{M} \vec{d} \\
            \vec{1} - \varepsilon\frac{N}{M} \vec{d}
        \end{bmatrix} \\
        &= \frac{1}{\sqrt{2}\sqrt{1+\varepsilon^2N^2/M^2}}\underbrace{\begin{bmatrix}
           I_{N'} & 0 \\
           0 & D \end{bmatrix}}_{\text{unitary}}
        \underbrace{\begin{bmatrix}
           I_{N'}/\sqrt{2} & I_{N'}/\sqrt{2} \\
           I_{N'}/\sqrt{2} & -I_{N'}/\sqrt{2}
        \end{bmatrix}}_{\text{unitary}}
        \begin{bmatrix}
            \varepsilon\frac{N}{M}\vec{d} +\left(\kappa + \varepsilon(\kappa-1)\frac{N}{M}\vec{1}^\dagger\vec{d}\right) \vec{1}  \\
            \vec{1} - \varepsilon\frac{N}{M} \vec{d}
        \end{bmatrix}
    \end{align}
We recall that unitary matrices do not change norms, and compute the norm of $\vec{x}$ as
\begin{align}
    \nrm{\vec{x}}^2 &= \frac{
    1+2\left(\varepsilon\frac{N}{M}\right)^2
    +\left(\kappa + \varepsilon(\kappa-1)\frac{N}{M}\vec{1}^\dagger \vec{d}\right)^2
    -2\left(\varepsilon\frac{N}{M}\right)\vec{1}^\dagger \vec{d} 
    +2\left(\varepsilon\frac{N}{M}\right)\left(\kappa + \varepsilon(\kappa-1)\frac{N}{M}\vec{1}^\dagger \vec{d}\right)\vec{1}^\dagger \vec{d}}
    %
    {2(1+\varepsilon^2N^2/M^2)} \\
\end{align}
Now we observe that the value of $\vec{1}^\dagger\vec{d}$ differs in case (i) and case (ii) of the \texttt{PromiseMajority} problem. 
\begin{equation}
    \vec{1}^\dagger\vec{d} = \begin{cases}
        M/N' = 2M/N & \text{in case (i) }\\
        -M/N' = -2M/N & \text{in case (ii)}
    \end{cases}
\end{equation}
Let $\nrm{\vec{x}}_{(i)}$ and $\nrm{\vec{x}}_{(ii)}$ denote the norm values in the two caes. We have
\begin{align}
    \frac{\nrm{\vec{x}}_{(i)}^2}{\nrm{\vec{x}}_{(ii)}^2} &= 
    \frac{
    %
    1+2\left(\varepsilon\frac{N}{M}\right)^2+\left(\kappa + 2\varepsilon(\kappa-1)\right)^2-4\varepsilon +4\varepsilon\left(\kappa + 2\varepsilon(\kappa-1)\right)}
    %
    {1+2\left(\varepsilon\frac{N}{M}\right)^2+\left(\kappa - 2\varepsilon(\kappa-1)\right)^2+4\varepsilon -4\varepsilon\left(\kappa - 2\varepsilon(\kappa-1)\right)} \\
    &= \frac{
    \left(\kappa^2+1+2\left(\varepsilon\frac{N}{M}\right)^2+4\varepsilon^2 (\kappa^2-1)\right) + 4\varepsilon(\kappa^2-1)}
    {\left(\kappa^2+1+2\left(\varepsilon\frac{N}{M}\right)^2+4\varepsilon^2 (\kappa^2-1)\right) - 4\varepsilon(\kappa^2-1)} \\
    &=\frac{
    1+ \frac{4\frac{\kappa^2-1}{\kappa^2+1}\varepsilon}{1+2\left(\varepsilon\frac{N}{M}\right)^2\frac{1}{\kappa^2+1}+4\frac{\kappa^2-1}{\kappa^2+1}\varepsilon^2}}
    {1- \frac{4\frac{\kappa^2-1}{\kappa^2+1}\varepsilon}{1+2\left(\varepsilon\frac{N}{M}\right)^2\frac{1}{\kappa^2+1}+4\frac{\kappa^2-1}{\kappa^2+1}\varepsilon^2}} 
\end{align}
Now we recall some of our parameter relations. We have chosen $M$ such that $M \geq 4\varepsilon N/\kappa$. We have assumed $\kappa \geq 3$, so $(\kappa^2-1)/(\kappa^2+1) \geq 4/5$. Finally, we have assumed $\varepsilon \leq 1/4$, so $4\frac{\kappa^2-1}{\kappa^2+1}\varepsilon^2 \leq 1/4$. This allows us to assert
\begin{align}
    \frac{\nrm{\vec{x}}_{(i)}^2}{\nrm{\vec{x}}_{(ii)}^2} &\geq 
    \frac{1+\frac{\frac{16}{5}\varepsilon}{1+\frac{1}{8} + \frac{1}{4}}}{1+\frac{\frac{16}{5}\varepsilon}{1+\frac{1}{8} + \frac{1}{4}}} \\
    &= \frac{1+\frac{128}{55}\varepsilon}{1-\frac{128}{55}\varepsilon}\label{eq:ratio_norms_bound}
\end{align}
Now, if the algorithm $\mathcal{A}$ produces an estimate $t$ satisfying $t \in [(1+\varepsilon)^{-1} \nrm{\vec{x}}, (1+\varepsilon)\nrm{\vec{x}}]$, 
then 
\begin{equation}
    t^2 \in [(1+2\varepsilon+\varepsilon^2)^{-1} \nrm{\vec{x}}^2, (1+2\varepsilon +\varepsilon)\nrm{\vec{x}}^2] \subseteq [(1+9\varepsilon/4)^{-1}\nrm{\vec{x}}^2,(1+9\varepsilon/4)\nrm{\vec{x}}^2 ]\label{eq:tsquared_inclusion}
\end{equation}
where the last inclusion follows under the assumption $\varepsilon \leq 1/4$. Since $9/4 < 128/55$,  Eqs.~\eqref{eq:ratio_norms_bound} and \eqref{eq:tsquared_inclusion} together imply that by taking the output $t$ and determining whether $t^2$ is closer to $\nrm{\vec{x}}_{(i)}^2$ or $\nrm{\vec{x}}_{(ii)}^2$, we can determine whether we are in case (i) or case (ii) and solve the \texttt{PromiseMajority} problem. 

It remains to show that we can construct a block-encoding $U_A$ and state-preparation untiary $U_{\vec{b}}$ for $A$ and $\vec{b}$ as defined. Constucting $U_{\vec{b}}$ is simple as the vector does not depend on $\vec{y}$ and need not query $\mathcal{P}_{\vec{y}}$, and all entries of the superposition are known. To construct $U_A$, we refer to its factorization in Eq.~\eqref{eq:A_factorization}. The block-encoding for $A$ multiplies block-encodings for each of the factors:
\begin{itemize}
    \item The unitary operation $\left[\begin{smallmatrix}
        I_{N'} & 0 \\
        0 & D
    \end{smallmatrix}\right] = \ketbra{0}\otimes I_{N'} + \ketbra{1}\otimes \mathcal{P}_{\vec{y}}$ is accomplished with a controlled-$P_{\vec{y}}$ query. Controlled-$P_{
    \vec{y}}$ can be built from $P_{\vec{y}}$ as follows. Prior to beginning the algorithm, we make make $T = O(1)$ queries $P_{\vec{y}}$ on random inputs $i_1, \ldots, i_T$, learning $y_{i_1}, \ldots, y_{i_T}$. With high probability we find a value $i^*$ for which $y_{i^*} = 0$ (unless $M \approx N$ and we are in case (ii), in which case we can already easily differentiate case (i) and case (ii) in $O(1)$ classical queries). We prepare the state $\ket{i^*}$ in an ancilla register. In order to apply controlled-$P_{\vec{y}}$ on the state $\ket{i^*}\ket{c}\ket{i}\ket{z}$, where the first qubit is the ancilla, the second qubit is the control, and the final two qubits are the target of the query, we perform a controlled swap between the registers holding $\ket{i^*}$ and $\ket{i}$, controlled on the register holding $\ket{c}$. Then, we query the final two registers, and undo the controlled swap. Since $y_{i^*} = 0$, there is no action on the final register when the control is set to 1. 
    \item The unitary operation $\frac{1}{\sqrt{2}}\left[\begin{smallmatrix}
        I_{N'} & I_{N'} \\
        I_{N'} & -I_{N'}
    \end{smallmatrix}\right]$ is equivalent to a Hadamard gate, $H \otimes I_{N'}$. No queries to $\mathcal{P}_{\vec{y}}$ are required. 
    \item The matrix $\left[\begin{smallmatrix}
        C & 0 \\
        0 & I_{N'}
    \end{smallmatrix}\right]= \ketbra{0}\otimes C + \ketbra{1} \otimes I_{N'}$ is accomplished by controlled-$U_C$ operation, where $U_C$ is a block-encoding of $C$. We block-encode $C$ as follows. Let $\theta = \arccos(\kappa^{-1})$. First, perform a parallel layer of Hadmard gates, which maps $\ket{\vec{1}}$ to $\ket{\vec{e_0}}$. Then, introduce an ancilla qubit and apply $\ketbra{\vec{e_0}}\otimes e^{i\theta Y} + (I_{N'}-\ketbra{0}) \otimes I_2$. Finally, apply another layer of Hadamards to send $\ket{0}$ back to $\ket{\vec{1}}$. Since $\bra{0}e^{i\theta Y}\ket{0} = \kappa^{-1}$, we have the correct $\bra{\vec{1}}C \ket{\vec{1}}$ matrix element. For any vector orthogonal to $\ket{\vec{1}}$, the $R_y$ gate is not triggered, so $U_C$ acts as identity.  
\end{itemize}
Overall, each $U_A$ operation requires two $\mathcal{P}_{\vec{y}}$ queries, and each $U_{\vec{b}}$ operation requires no queries. This completes the proof that the number of queries to $U_A$ must be at least $\Omega(\min(\kappa \varepsilon^{-1}, N))$.

The method above does not lower bound the number of times $\mathcal{A}$ must query the state-preparation unitary $U_{\vec{b}}$. To make this lower bound, we slightly modify the construction to
    \begin{align}
        A = &
        \underbrace{\begin{bmatrix}
           I_{N'}/\sqrt{2} & I_{N'}/\sqrt{2} \\
           I_{N'}/\sqrt{2} & -I_{N'}/\sqrt{2}
        \end{bmatrix}}_{\text{unitary}}
        \begin{bmatrix}
           C & 0 \\
           0 & I_{N'} 
       \end{bmatrix} 
        \underbrace{\begin{bmatrix}
           I_{N'}/\sqrt{2} & I_{N'}/\sqrt{2} \\
           I_{N'}/\sqrt{2} & -I_{N'}/\sqrt{2}
        \end{bmatrix}}_{\text{unitary}}
        \\
       \vec{b} = \frac{1}{\sqrt{2}\sqrt{1 + \varepsilon^2N^2/M^2}}
       &\underbrace{\begin{bmatrix}
           I_{N'} & 0\\
           0 & D
        \end{bmatrix}}_{\text{unitary}}
        \begin{bmatrix}
            \vec{1}\\
             \varepsilon\frac{N}{M} \vec{1}
        \end{bmatrix}
    \end{align}
    Here, the block-encoding $U_A$ requires zero queries to $\mathcal{P}_{\vec{y}}$, and the unitary $U_{\vec{b}}$ can be implemented with one query. Furthermore, the solution $\vec{x}$ is unitarily related to the solution stated above, and thus has the same norm in case (i) and case (ii). The same analysis then implies that the number of queries to the state preparation unitary cannot be smaller than $\Omega(\kappa\varepsilon^{-1},N)$.
\end{proof}

\section{Optimal QLSS based on quantum Zeno effect}\label{app:Zeno}

In this section, we sketch an analysis for an optimal QLSS based on the Zeno method. This is essentially the same as the near-optimal method from Ref.~\cite{lin2019OptimalQEigenstateFiltering}, applied to the adiabatic path  \autoref{sec:estimate_by_adiabatic}, and incorporating fixed-point amplitude amplifciation and log log trick to achieve optimality. 

We assume without loss of generality that $\kappa$ is a power of 2. First, we note the identity that when $\kappa^{-1} \leq \varsigma \leq 1$ and $\kappa^{-1} \leq \sigma < \sigma' \leq 1$, we have
\begin{align}
    \frac{f(\sigma)f(\sigma') + \varsigma^2 \sqrt{(1-f(\sigma)^2)(1-f(\sigma')^2)}}{f(\sigma')^2 + \varsigma^2(1-f(\sigma')^2)} &\geq \frac{f(\sigma)f(\sigma') + \varsigma^2 (1-f(\sigma')^2) }{f(\sigma')^2 + \varsigma^2(1-f(\sigma')^2)}  \\
    &= \frac{f(\sigma')^2 + \varsigma^2 (1-f(\sigma')^2)  - f(\sigma')(f(\sigma')-f(\sigma))}{f(\sigma')^2 + \varsigma^2(1-f(\sigma')^2)}\,, \\
    &= 1- (f(\sigma') - f(\sigma))\frac{f(\sigma')}{f(\sigma')^2 + \varsigma^2(1-f(\sigma')^2)} \\
    &\geq 1- (f(\sigma') - f(\sigma))\frac{f(\sigma')}{\sigma'^2} \\
    &\geq 1- \frac{f(\sigma') - f(\sigma)}{\sigma'} \label{eq:Zeno_identity}
\end{align}
where the first line uses monotonicity of $f$ (i.e., $f(\sigma) < f(\sigma')$), the second to last line uses the identity $f(\sigma')^2 + \kappa^{-2}(1-f(\sigma')^2) = \sigma'^2$ and the fact that $\varsigma \geq \kappa^{-1}$, and the final line uses $\sigma' \geq f(\sigma')$. The key fact is that the overlap between nearby  states along the adiabatic path remains close to one (see \autoref{sec:estimate_by_adiabatic} for definitions):
\begin{align}
    \lvert \braket{\vec{\bar{x}_{\sigma}}}{\vec{\bar{x}_{\sigma'}}} \rvert &= 
    \frac{\vec{\bar{x}_{\sigma'}}^\dagger \vec{\bar{x}_{\sigma}} }{\nrm{\vec{\bar{x}_{\sigma'}}}\nrm{\vec{\bar{x}_{\sigma}}}} \geq 
    \frac{\vec{\bar{x}_{\sigma'}}^\dagger \vec{\bar{x}_{\sigma}} }{\nrm{\vec{\bar{x}_{\sigma}}}^2}
    =  \frac{\sum_{j} |w_j|^2\frac{f(\sigma)f(\sigma') + \varsigma_j^2\sqrt{(1-f(\sigma)^2)(1-f(\sigma')^2)}}{(f(\sigma)^2 + (1-f(\sigma)^2) \varsigma_j^2)(f(\sigma')^2 + (1-f(\sigma')^2) \varsigma_j^2)}}{\sum_j \frac{|w_j|^2}{f(\sigma)^2 + (1-f(\sigma)^2)\varsigma_j^2}  } \\
    &\geq  \left(1-\frac{f(\sigma')-f(\sigma)}{\sigma'}\right)\frac{\sum_{j} \frac{|w_j|^2}{f(\sigma)^2 + (1-f(\sigma)^2) \varsigma_j^2}}{\sum_j \frac{|w_j|^2}{f(\sigma)^2 + (1-f(\sigma)^2)\varsigma_j^2}  } \\
    &= 1-\frac{f(\sigma')-f(\sigma)}{\sigma'}\,. \label{eq:overlap_zeno}
\end{align}
where we have used the identity of Eq.~\eqref{eq:Zeno_identity} in the second line.  Note that since $\sigma' \geq f(\sigma')$, the right-hand side is greater than $f(\sigma)/f(\sigma')$. When $\sigma, \sigma' \geq 2 \kappa^{-1}$, $f(\sigma) \approx \sigma$ holds reasonably well, so we may roughly think of the overlap as being lower bounded by approximately $\sigma/\sigma'$, as mentioned in the main text. 

We may now construct a Zeno-inspired algorithm using a similar sequence of $\sigma$ values as was used for norm estimation in \autoref{sec:estimate_by_adiabatic}.  There, we chose the geometric sequence: $1, 2^{-1},2^{-2},\ldots, 2^{-\log_2(\kappa)}\equiv \kappa^{-1}$. Here, we choose the same sequence with the inverse of $f$ applied, appended with a final point where $\sigma = \kappa^{-1}$:
\begin{align}
    1\equiv f^{-1}(1), f^{-1}(2^{-1}), f^{-1}(2^{-2}), \ldots, f^{-1}(\kappa^{-1}), f^{-1}(0.5\kappa^{-1}), f^{-1}(0) \equiv \kappa^{-1}
\end{align}
This sequence has the property that the overlap between subsequent states in the sequence is always lower bounded by 1/2. We can see this by applying the identity in Eq.~\eqref{eq:overlap_zeno} and obtaining
\begin{align}
    \lvert \braket{\vec{\bar{x}_{f^{-1}(2^{-j})}}}{\vec{\bar{x}_{f^{-1}(2^{-j-1})}}} \rvert \geq  \frac{f(f^{-1}(2^{-j-1}))}{f(f^{-1}(2^{-j}))}= \frac{1}{2}
\end{align}
and also (for the last jump in the sequence)
\begin{align}
    \lvert \braket{\vec{\bar{x}_{f^{-1}(0.5\kappa^{-1})}}}{\vec{\bar{x}_{f^{-1}(0)}}} \rvert \geq 1-\frac{0.5\kappa^{-1}}{f^{-1}(0.5\kappa^{-1})} \geq \frac{1}{2}\,,
\end{align}
where the last inequality is true since the the function $f^{-1}$ maps all inputs to values greater than $\kappa^{-1}$. 

The main idea of the Zeno algorithm is to jump from one state in the sequence to the next. Namely, given that we have prepared the state $\ket{\vec{\bar{x}_{\sigma'}}}$,  we approximately implement the jump $\ket{\vec{\bar{x}_{\sigma'}}}\mapsto \ket{\vec{\bar{x}_{\sigma}}}$ by  fixed-point amplitude amplification \cite{yoder2014FixedPointSearch}; this requires the ability to reflect about both $\ket{\vec{\bar{x}_{\sigma'}}}$
 and $\ket{\vec{\bar{x}_{\sigma'}}}$, which is supplied by KR to error $\delta/2$ at cost $O(\sigma^{-1}\log(1/\delta))$. Since the overlap is lower bounded by a constant, we may amplify the jump to success probability $1-\delta/2$ at cost $O(\log(1/\delta))$ calls to KR. Noting that $\sigma \geq f(\sigma)$, we may conclude that each jump is performed to error $\delta$ at cost $O(f(\sigma)^{-1} \log^2(1/\delta))$. As in the analysis in \autoref{sec:estimate_by_adiabatic}, we invoke the log log trick, choosing the value of $\delta$ on the jump from $f^{-1}(2^{-j})$ to $f^{-1}(2^{-j-1})$ to be $e^{\Omega(2+\log_2(\kappa)-j)}$, and choosing the value of $\delta$ on the final jump to be $O(1)$.  The total error over the $\log_2(\kappa)+2$ jumps is at most $O(1)$. The total cost is given by the sum (cf.~Eq.~\eqref{eq:Q_adiabatic_norm_optimal})
\begin{align}
 Q&=\underbrace{O(\kappa)}_{\text{final jump}} + \sum_{j=1}^{\log_2(\kappa)+1} O(2^j(2+\log_2(\kappa)-j)^2) \\
 &\leq  O(\kappa) \sum_{j'=0}^{\infty } 2^{-j'}(1+j')^2
 \leq O(\kappa)\,.\label{eq:Q_Zeno_norm_optimal}
\end{align}
This shows that a Zeno-based QLSS equipped with fixed-point amplitude amplification and the log log trick can achieve optimal complexity. 

\end{document}